\documentclass[a4paper, 11pt]{article}

\usepackage{bm}
\usepackage[colorlinks,citecolor=blue,urlcolor=blue]{hyperref}
\usepackage[sort]{natbib}
\usepackage[labelfont=bf]{caption}
\usepackage{bbm}

\usepackage[a4paper,top=2.5cm,bottom=2.5cm,left=2.5cm,right=2.5cm]{geometry}

\def\spacingset#1{\renewcommand{\baselinestretch}%
{#1}\small\normalsize} \spacingset{1.5}

\usepackage{graphicx}

\usepackage{hyperref}
\usepackage{tikz}
\usepackage{xcolor}
\usetikzlibrary{arrows.meta,positioning,calc,shapes.geometric,shapes.callouts,backgrounds}
\usepackage{colortbl}
\usepackage{array}
\usepackage{graphicx}
\usepackage{amsmath,amsthm,amssymb}
\usetikzlibrary{matrix}

\definecolor{lowprob}{RGB}{255,255,255}
\definecolor{highprob}{RGB}{255,102,102}

\usepackage{caption}
\usepackage{subcaption}
\usepackage{algorithm}
\usepackage{algorithmic}
\usepackage{comment}
\usepackage{booktabs}

\usepackage{lscape}
\newcolumntype{P}[1]{>{\centering\arraybackslash}p{#1}}

\newtheorem{Theorem}{Theorem}
\newtheorem{Proposition}{Proposition}
\newtheorem{Assumption}{Assumption}
\newtheorem{Lemma}{Lemma}
\newtheorem{remark}{Remark}

\newcommand{\E}{\mathbb{E}}
\newcommand{\V}{\mathbb{V}}
\newcommand{\tr}{\operatorname{tr}}
\newcommand{\bbP}{\mathbb{P}}
\newcommand{\dmin}{\delta_{\min}}

\def\mb{\mathbf}
\def\c{\cal}
\def\t{\tilde}

\newcommand{\Pbb}{\mathbb{P}}
\newcommand{\Var}{\operatorname{Var}}
\newcommand{\Cov}{\operatorname{Cov}}
\newcommand{\vecop}{\operatorname{vec}}

\newcommand{\R}{\mathbb{R}}

\newcommand{\norm}[1]{\left\lVert #1\right\rVert}
\newcommand{\toP}{\xrightarrow{p}}

\newcommand{\argmin}{\operatorname*{arg\,min}}
\newcommand{\tod}{\xrightarrow{d}}
\newcommand{\PP}{\mathrm{PPI}}
\newcommand{\Human}{\mathrm{H}}
\newcommand{\AI}{\mathrm{A}}

\begin{document}
\title{Recidivism Prediction, Peer Effect Estimation, and Prediction-Powered Inference with LLM Text Measures}
\author{Shanjukta Nath$^{1}$, Jiwon Hong$^{2}$, Jae Ho Chang$^{2}$, Keith Warren$^{3}$, and \\Subhadeep Paul$^{2}$\thanks{$^{1}$Department of Agricultural and Applied Economics, University of Georgia; shanjukta.nath@uga.edu, $^{2}$Department of Statistics, The Ohio State University, $^{3}$College of Social Work, The Ohio State University. We thank Sergio Urzua, Guido Kuersteiner, Susan Athey, Emil Palikot, Aritra Guha, and seminar participants at the Twenty-Ninth Annual Meetings of the Society of Labor Economists and the AI Summit at The Ohio State University for helpful comments and feedback. SP's research was partially funded by a grant from NSF DMS (2529302). We thank the Ohio Supercomputer Center (OSC) for providing us with computing resources through an academic account. The Therapeutic Community data used for this paper were exempted from IRB review based on IRB Protocol 2007B0162 by The Ohio State University. The human surveys conducted on Prolific to validate LLM-generated output were exempted from IRB review under Protocol STUDY20260794 at The Ohio State University.}}
\date{}

\maketitle

\spacingset{1.5}

\begin{abstract}
We provide a new framework for estimating peer effects when outcomes are multivariate behavioral measures derived from written text using an LLM and the network formation is endogenous. We obtain LLM embeddings and zero shot classification of more than 200,000 written exchanges among residents of low-security correctional facilities.  We find that LLM embeddings improve out-of-sample recidivism prediction by up to 30\% over pre-entry covariates alone using LASSO and LoRA fine-tuning, showing that text representations capture meaningful signals. For peer effect estimation,  we develop a novel instrumental variable estimator that accommodates multivariate outcomes, sparse networks, and multidimensional latent homophily. We show that this estimator is $\sqrt{N}$-consistent and asymptotically normal under sparsity conditions that relax dense-network assumptions prevalent in the peer effect literature. Limited human annotations are then combined with LLM zero-shot vectors in a new prediction-powered peer inference (PPPI) approach to obtain de-biased estimates and valid inference. Results reveal significant peer effects in the behavioral profiles. 
\end{abstract}

\medskip

\medskip

\medskip


\section{Introduction}

Language exchanges often convey behavioral traits, beliefs, norms, accountability, and conflict in several socio-economic settings \citep{fouka2020backlash,grogger2019speech,niederhoffer2002linguistic}. Measuring and modeling language exchanges gives researchers a way to understand behavioral traits in high-stakes environments, such as correctional facilities, where behavior significantly affects rehabilitation outcomes and recidivism. However, researchers can rarely observe direct data on these exchanges, limiting what can be learned about behavior transmitted through such exchanges. Moreover, language interactions are, by nature, non-numerical, unstructured, high-dimensional, and context-dependent, making them difficult to analyze mathematically.

Our data comes from a unique and rare setting called Therapeutic Communities (TCs), which are low-security correctional facilities where such language exchanges between residents could be recorded. These facilities are for individuals with histories of criminal behavior and substance use disorders, and are programs based on mutual aid that encourage residents to collaborate, practice group therapy, and provide regular feedback to peers through affirmations and corrections \citep{de2000therapeutic}. Corrections are targeted at individuals who violate TC norms, and affirmations are given to encourage positive and prosocial behavior \citep{de2000therapeutic,perfas2014therapeutic}. These written exchanges provide a behavioral record of interactions inside the facility. The dataset includes message text, sender, receiver, timing of message exchanges, resident entry and exit dates, pre-entry covariates, and linked three-year recidivism records.

Using this correctional facility data, we answer two empirical questions. First, can we use the written exchanges to improve recidivism prediction beyond what pre-entry covariates can achieve in this setting? Second, do the peers in these facilities meaningfully influence each other's behavior? These two research questions are connected because the written exchanges are both signals of underlying behavioral traits and records of the content of social interactions within the facility. Consequently, the same text data can be used to answer both questions, although each empirical exercise requires a different text representation and econometric methodology.

The objective of the first empirical exercise is to maximize the predictive accuracy of recidivism.  Traditional text analysis methods deployed in economics, such as bag-of-words, dictionary-based methods, and topic models \citep{kalamara2022making,lippmann2022gender,brehm2025vaccines,alsan2025something}, have limitations in capturing implicit and context-dependent semantic meaning of text. Bag-of-words and conventional topic models rely primarily on patterns of word occurrence or co-occurrence, while dictionary-based methods typically use the explicit presence of researcher-specified words or phrases. Some recent works, such as \cite{gennaro2022emotion,ash2022ideas}, have improved on these earlier models and instead deployed earlier text embedding methods such as Word2Vec \citep{mikolov2013efficient}, for text analysis. However, even these models are limited by their inability to incorporate context in their embeddings. We use output from transformer-based Large Language Models (LLMs), which improve upon early text embedding models by incorporating both an attention mechanism \citep{vaswani2017attention} and bidirectional processing \citep{devlin2019bert}, to obtain context-aware vector representations of text. We predict recidivism using open-source open-weight BERT-based text embeddings in two ways, first using these embeddings as covariates in a penalized logistic regression (lasso) model, and second by building a predictive model with parameter-efficient fine-tuning approaches such as adapter tuning \citep{houlsby2019parameter} and LoRA fine-tuning \citep{hu2022lora}.

For the second empirical exercise on estimation of peer effects in language exchanges, we face two challenges. First, although LLM embedding vectors provide a meaningful numerical representation of language, they are not well-suited for peer effect estimation due to their high dimensionality (700-1000). Such high dimensionality limits the scope of interpretation in peer effect estimation. Second, 
similar to peer effect estimation in other settings, the identification of peer effects in language is complex due to multiple issues. These include simultaneity in the outcomes and confounding arising from unobserved homophily that affects both the outcome and peer selection \citep{manski1993identification,shalizi2011homophily,johnsson2021estimation,goldsmith2013social}. 

We address the first challenge to peer effects estimation using zero-shot classification, which provides a suitable way to map high-dimensional text embeddings to interpretable user-defined input label categories. Zero-shot classification reframes the classification task as a language inference problem that determines entailment between a text and the user-supplied label categories \citep{yin2019benchmarkingzeroshottextclassification}. Notably, this language inference task does not require labeled training data. We therefore obtain zero-shot probabilities at the message level using pretrained BART-large model, which are subsequently aggregated at the resident level to obtain resident behavioral profiles on which peer effect is measured.

In terms of the second challenge, identification of peer effects can be achieved using instrumental variables, where peers interact in network settings, using the identification results outlined in \cite{bramoulle2009identification}. However, \cite{bramoulle2009identification}'s results assume that the network is uncorrelated with the error term conditional on covariates, while networks are typically endogenous because of latent homophily. These are unobserved characteristics that drive both network formation (peer selection) and the outcomes. Recent work by \cite{johnsson2021estimation} addresses this problem through a latent variable model, but requires unrealistically dense networks (expected degree scaling as $N$) and relies on unidimensional latent variables. Other work accommodates multivariate correlated outcomes \citep{zhu2020multivariate,cohen2018multivariate} but assumes exogenous network formation or does not consider the multivariate simultaneous peer effect problem \citep{nath2022identifying}. Consequently, the identification and estimation of peer effects in language require new econometric methods and theory that accommodate multivariate correlated outcomes, multidimensional latent variables, and realistic, sparse networks.

To address the aforementioned challenge, we develop new econometric methods based on instrumental variables (IV) and latent space network models to accommodate correlated multivariate outcomes and multidimensional latent variables. Our methods do not require dense networks and accommodate flexible nonparametric modeling of multidimensional latent homophily vectors. The proposed estimator is shown to be $\sqrt{N}$ consistent and asymptotically normal. Moreover, the data-set in the empirical application is rich enough to provide the necessary variation for identification. The peer network is based on realized written exchanges and not simply co-residence, making the setting especially well suited to study peer influence through language.

We also contribute to modeling networks using multidimensional additive and multiplicative latent-variable models for peer-effect estimation. While the works of \cite{goldsmith2013social,johnsson2021estimation} integrated network formation with peer effects estimation, these approaches have used univariate additive latent factors to account for the unobserved confounders. We advocate the use of Random dot product graph (RDPG) models \citep{athreya2017statistical}, which are multiplicative latent variable models, and the latent space models, which accommodate additive and multiplicative latent effects along with observed covariates \citep{hoff2002latent,ma2020universal,li2023statistical}. In contrast to the one-dimensional latent-variable model of \cite{graham2017econometric} that \cite{johnsson2021estimation} employed, the models we deploy use multidimensional latent variables to model network data flexibly. We build methodological and theoretical tools needed to enable the use of this rich class of additive and multiplicative latent space models in peer effect identification, using estimated latent positions with a nonparametric sieve adjustment.

Because the paper uses LLM-generated behavioral profiles for peer-effect inference, validating these profiles is a critical step. We implement a multi-part validation strategy \citep{ludwig2024large,modarressi2025causal,egami2022make} for this purpose.\footnote{We thank an anonymous reviewer for suggesting a series of these validation checks.} We use sample splitting so that the data used to generate the input labels is distinct from the held-out data used for zero-shot classification and downstream analysis. We then evaluate the resulting categories using three human validation surveys run using the platform Prolific that compare model-based classifications with human judgments. While the first survey validates the zero-shot labels given to the LLM, the other two surveys are used to obtain human-judgment-based behavioral profiles of the residents. We use these human evaluations to calibrate the LLM model predictions using prediction-powered inference (PPI) methods \citep{angelopoulos2023prediction} for the peer-effect analysis. We further use placebo-label checks and prediction of affirmation and correction message types as additional validation exercises. Together, these exercises show that the text-derived measures used in the paper capture meaningful behavioral variation in these correctional facilities. 

Despite extensive recent interest in PPI, methods needed to apply it in network settings for peer-effect estimation are unavailable to the best of our knowledge. Network settings present unique challenges for PPI because, unlike the iid case that \cite{angelopoulos2023prediction} considers, it is not straightforward to construct a peer-effect estimator using only human-labeled data (classical inference) when labels are available only for a subset of nodes, not the entire network. In particular, the human-reference endogenous regressor \(\mb G\mb Y^{\Human}\) depends on the outcomes of a resident's network neighbors, many of whom may not be human-validated. Simply restricting the analysis to validated residents would change the network and the peer-effect estimand. Further, in the network setting, calibrating the AI-based outcome of an individual who was human-validated changes not only that resident's outcome but also the endogenous peer outcome \(\mb G\mb Y\) for their neighbors. Hence, the calibration propagates through the network. Moreover, standard PPI variance calculations do not account for the additional uncertainty generated by sampling residents for human validation and its propagation through the network and instrumental-variable estimator. We therefore develop a new prediction-powered peer-effect inference (PPPI) approach that corrects AI-predicted outcome labels using human validation data through a moment correction that weights residents' contributions by the inverse of their selection probability. We show that the resulting estimator is consistent for the ``gold-standard'' human-only peer effect parameter and obtain its asymptotic limiting variance that includes the resident sampling uncertainty.

Using only the pre-entry covariates, we can predict 3-year recidivism with an out-of-sample AUC only slightly above 0.5, which is just marginally better than random guessing. However, using transformer-based LLM embeddings of the 80,000 (male unit) to 120,000 (female unit) written exchanges, prediction accuracy improves dramatically to an AUC of 0.70-0.72, approximately 30\% higher than using pre-entry characteristics alone. This performance is comparable to or slightly better than published AUC values from state-of-the-art commercial software, such as COMPAS, which uses hundreds of individual characteristics in an undisclosed machine learning model \citep{dressel2018accuracy}. An AUC of 0.70 is also generally considered ``good'' in the context of predicting criminal recidivism \citep{laqueur2024algorithmic}. These results establish that LLM representations of written exchanges capture economically meaningful behavioral variation and can be useful to build a powerful predictive model in settings such as TCs with limited individual-level predictors. 

Although the prediction exercise does not by itself yield causal conclusions, it is still policy relevant for resource allocation in correctional facilities. In settings where rehabilitative resources are limited, language-based risk measures can help target additional attention toward residents at higher predicted risk of recidivism, and can be updated continuously during their time in the facility rather than being fixed at intake \citep{kleinberg2015prediction}.

Using zero-shot classification, we map the unstructured language exchanges to four user-defined input label categories that represent distinct behavioral dimensions. These chosen input labels are community support, personal growth, rule violations, and disruptive conduct. Applying the peer effects methodology developed in this paper to the zero-shot profiles (predicted probabilities for language categories) reveals significant peer effects. Residents affect each other's language profiles even after accounting for unobserved homophily using the new peer effect estimation methods developed in this paper. Peer effects within the same behavioral dimension are substantially larger than cross-category spillovers across dimensions, and are often statistically significant at the conventional levels, a pattern observed consistently across both male and female units.

We contribute to a small but growing economics literature that uses LLM derived measures for economic applications. \cite{hansen2023remote} and \cite{juhasz2025measuring} fine-tune BERT models on human-annotated data to measure remote work arrangements and industrial policy, respectively, and other papers such as \cite{bajari2023hedonic} use LLMs to construct price-relevant representations of products. 
Our contribution to this literature includes prediction models for economic outcomes using LLM embeddings and parameter-efficient fine-tuning of pre-trained LLMs, and peer effect estimation using LLM zero-shot classification with administrative language data. Further, with the new PPI for peer effects approach, we demonstrate how AI measures can be incorporated into econometric models such as peer effect model, with limited gold-standard human or experimental data.

\section{Context and Institutional Setting}
\label{empiricalsetup}

This paper uses datasets from two male units and one female unit of Therapeutic Communities (TCs). These TCs located in a Midwestern state in the United States are minimum-security, community-based correctional facilities.  Thus, while they are locked facilities, they are not prison units.  These TCs had residents from urban, suburban, and rural areas. Each unit had 80 beds.  Since data are collected over time, the female unit had 982 unique residents, and the male units had 1649 unique residents. The two male units are located on two floors of the same TC; henceforth, they will be collectively referred to as the male unit. We observe the entry and exit dates of these residents. The residents may stay in these units for up to 180 days. They may either successfully graduate from the program at the end of their stay or leave it unsuccessfully. We find that, on average, the residents stayed for about 120 days. Table \ref{summarystatsdata} details the residents used in the analysis. The graduation rate is high from these units, 87\% in the female unit and 89\% in the male unit, respectively.

Entry and exit timestamps show substantial variation in these dates and residents' length of stay in these units. Figure \ref{time} (a,b) shows the variation in these entry and exit dates, and panel (c) shows substantial differences in the length of stay. The resident information from TCs is mapped with the administrative dataset maintained by the State's Department of Rehabilitation and Corrections. This administrative dataset provides the date of reincarceration of TC residents, tracked up to 3 years. Reincarceration is one of several possible measures of recidivism, but it is the most commonly used in the criminology literature \citep{bouchard2024seeing} and has been used in previous studies of TC outcomes \citep{doogan2016semantic,warren2020tightly}. Table \ref{summarystatsdata} shows that 22\% of the residents had recidivism in the female unit and 25\% in the male unit (among all residents, irrespective of graduation/success from TC treatment status). Interestingly, these numbers align with global figures obtained by \cite{yukhnenko2023criminal} across 33 countries and various types of crimes.

During their stay in TCs, the residents were encouraged to exchange written affirmations and corrections (very short messages) with each other as part of their therapy. We observe the timestamps, text, sender, and receiver IDs for each of these exchanges in these TCs. The digitized records of these texts consist of 123,000 exchanges between residents in the female unit and about 83,000 exchanges between residents in the male unit. Figure \ref{counttexts}  shows the distribution of the count of affirmations (push) and corrections (pull) aggregated by sender ID in panels (a,b) and by receiver ID in panels (c,d), respectively. Figure \ref{fig:resident-timeline} provides a few examples of these messages and the data structure. In this figure, we create a sample of four residents who entered the facility around a similar time and exchanged messages. Note that this figure shows examples of only a few messages exchanged between these residents due to space constraints.

\begin{figure}[!h]
\caption{Timeline of resident stays and exchanges (2006)}
\label{fig:resident-timeline}
\centering
\begin{tikzpicture}[
    scale=0.65, 
    transform shape,
    resident/.style={rectangle, draw, fill=#1!20, minimum width=2cm, minimum height=0.7cm, rounded corners},
    message/.style={rectangle callout, draw, fill=#1!30, rounded corners=2pt, text width=3.2cm, align=left, font=\footnotesize}
]

\draw[-{Latex[length=3mm]}] (0,2) -- (14,2) node[right] {};

\foreach \x/\month/\day in {0/Feb/1, 4/Mar/1, 8/Apr/1, 12/May/1, 14/Jun/1}
    \draw (\x,2.3) -- (\x,1.7) node[below] {\month \day};

\node[resident=blue, anchor=south east] (R1) at (-2,-1.6) {Resident A};
\draw[thick, blue] (0,-2) -- (12,-2);
\node[left, align=right, font=\footnotesize] at (0,-2) {Entry: Jan 5, 2006};
\node[right, align=left, font=\footnotesize] at (12,-2) {Exit: May 3, 2006};

\node[resident=green, anchor=south east] (R2) at (-2,-4.1) {Resident B};
\draw[thick, green] (0.2,-4.5) -- (14,-4.5);
\node[left, align=right, font=\footnotesize] at (0,-4.5) {Entry: Jan 20, 2006};
\node[right, align=left, font=\footnotesize] at (14,-4.5) {$\rightarrow$ Exit: Jul 13, 2006};

\node[resident=red, anchor=south east] (R3) at (-2,-6.6) {Resident C};
\draw[thick, red] (0.7,-7) -- (14,-7);
\node[below, align=center, font=\footnotesize] at (0.7,-7) {Feb 9, 2006};
\node[right, align=left, font=\footnotesize] at (14,-7) {$\rightarrow$ Exit: Aug 7, 2006};

\node[resident=purple, anchor=south east] (R4) at (-2,-9.1) {Resident D};
\draw[thick, purple] (1.8,-9.5) -- (14,-9.5);
\node[below, align=center, font=\footnotesize] at (1.8,-9.5) {Feb 22, 2006};
\node[right, align=left, font=\footnotesize] at (14,-9.5) {$\rightarrow$ Exit: Aug 15, 2006};

\node[message=yellow](pushAC) at (2,-0.8) {Feb 13, 2006\\A → C: "Willingness to change your life"\\(push)};

\node[message=yellow](pushBC) at (3,-3.5) {Feb 28, 2006\\B → C: "Being a good listener and processing"\\(push)};

\node[message=yellow](pushDC) at (4,-11) {Mar 2, 2006\\D → C: "Studying hard for the board"\\(push)};

\node[message=orange](pullBA) at (5.3,-5.5) {Mar 26, 2006\\B → A: "Talking on quiet time"\\(pull)};

\node[message=orange](pullCD) at (11.5,-6) {Jun 23, 2006\\C → D: "Talking on quiet time"\\(pull)};

\draw[thick, dashed, yellow!70!black, -latex] (1.5,-1.3) -- (1.5,-7);
\draw[thick, dashed, yellow!70!black, -latex] (3,-4) -- (3,-7);
\draw[thick, dashed, yellow!70!black, -latex] (4,-10.5) -- (4,-7);
\draw[thick, dotted, orange!70!black, -latex] (5.3,-4.5) -- (5.3,-2);
\draw[thick, dotted, orange!70!black, -latex] (11.5,-7) -- (11.5,-9.5);

\end{tikzpicture}
\begin{minipage}{13.0 cm}{\footnotesize{Notes: The dates of the entry, exit, and exchange of messages have been altered slightly for the construction of this figure to maintain the anonymity of residents.}}
\end{minipage} 
\end{figure}

We deploy transformer-based LLMs to create text embeddings for the 123,000 and 83,000 written exchanges in the female and male units. Specifically, we use the open-weight and open-source BERT (\textit{bert-base-uncased}) language model from \cite{devlin2019bert}\footnote{The implementation is done using the \texttt{text} package in R \citep{kjell2023text} that allows users to access open-weight open-source LLMs easily from \href{https://huggingface.co/}{huggingface.co} and run these models locally in their environment.}. We use these embeddings to create sender and receiver profile vectors for the residents based on the language they use to interact with their peers. We describe in Section \ref{results} that these word embeddings have substantial predictive ability for recidivism when used in a LASSO model, with a deep neural network adapter, or with LoRA finetuning.

However, it is hard to interpret and validate these 768-dimensional embeddings meaningfully, especially for the second empirical question on peer effects in behavioral profiles. Therefore, our second method uses an LLM-based zero-shot approach \citep{yin2019benchmarkingzeroshottextclassification,meshkin2024harnessing,wang2023large} using the \textit{BART} model of \cite{lewis2019bart}. In this method, the user provides a fixed class of words to the LLM along with the text messages, and the method generates predictions for the user-supplied classes for each message. Note that the classification is performed without seeing any labeled training data (and hence ``zero-shot'') and with only the natural language knowledge that the LLM has acquired. 
\begin{figure}[!h]
\centering
\caption{Raw messages and their predicted class probabilities using Zero-Shot classifier}
\label{fig:zero-shot-classification}
\begin{tikzpicture}
\tikzset{
    box/.style={
        rectangle,
        draw=black,
        thick,
        minimum width=2.5cm,
        minimum height=0.8cm,
        text centered,
        text width=6cm,
        align=center
    },
    arrow/.style={
        ->,
        thick,
        >=stealth
    },
    matrixbox/.style={
        rectangle,
        draw=black,
        thick,
        minimum width=6.5cm,
        minimum height=6cm,
        align=center
    },
    title/.style={
        font=\large\bfseries,
        align=center
    },
    subtitle/.style={
        font=\footnotesize\bfseries,
        align=center
    }
}

\node[matrixbox] (input) at (-3.5, 2) {};
\node[align=center, font=\small\bfseries] at (-3.5, 6) {\textbf{Raw Messages}};

\node[align=left, font=\footnotesize] at (-3.5, 3.5) {\textbf{1.} ``holding family hostage"};
\node[align=left, font=\footnotesize] at (-3.5, 2.5) {\textbf{2.} ``talking in quiet zone in bathroom"};
\node[align=left, font=\footnotesize] at (-3.5, 1.5) {\textbf{3.} ``getting promotion"};
\node[align=left, font=\footnotesize] at (-3.5, 0.5) {\textbf{4.} ``dancing on dayroom floor"};
\node[align=left, font=\footnotesize] at (-3.5, -0.5) {\textbf{5.} ``always giving me positive feedback"};

\node[matrixbox] (output) at (4, 2) {};
\node[align=center, font=\small\bfseries] at (4, 6) {\textbf{Class Probabilities}};

\node[align=center, font=\tiny] at (1.5, 4.3) {\textbf{Personal} \\ \textbf{Growth}};
\node[align=center, font=\tiny] at (3, 4.3) {\textbf{Community} \\ \textbf{Support}};
\node[align=center, font=\tiny] at (4.5, 4.3) {\textbf{Rule} \\ \textbf{Violations}};
\node[align=center, font=\tiny] at (6, 4.3) {\textbf{Disruptive} \\ \textbf{Conduct}};

\node[rectangle, fill=lowprob, draw=black, minimum width=1cm, minimum height=0.6cm, font=\tiny] at (1.5, 3.5) {0.02};
\node[rectangle, fill=lowprob, draw=black, minimum width=1cm, minimum height=0.6cm, font=\tiny] at (3, 3.5) {0.01};
\node[rectangle, fill=lowprob, draw=black, minimum width=1cm, minimum height=0.6cm, font=\tiny] at (4.5, 3.5) {0.17};
\node[rectangle, fill=highprob, draw=black, minimum width=1cm, minimum height=0.6cm, font=\tiny] at (6, 3.5) {0.79};

\node[rectangle, fill=lowprob, draw=black, minimum width=1cm, minimum height=0.6cm, font=\tiny] at (1.5, 2.5) {0.10};
\node[rectangle, fill=lowprob, draw=black, minimum width=1cm, minimum height=0.6cm, font=\tiny] at (3, 2.5) {0.04};
\node[rectangle, fill=highprob, draw=black, minimum width=1cm, minimum height=0.6cm, font=\tiny] at (4.5, 2.5) {0.68};
\node[rectangle, fill=lowprob, draw=black, minimum width=1cm, minimum height=0.6cm, font=\tiny] at (6, 2.5) {0.19};

\node[rectangle, fill=highprob, draw=black, minimum width=1cm, minimum height=0.6cm, font=\tiny] at (1.5, 1.5) {0.84};
\node[rectangle, fill=lowprob, draw=black, minimum width=1cm, minimum height=0.6cm, font=\tiny] at (3, 1.5) {0.08};
\node[rectangle, fill=lowprob, draw=black, minimum width=1cm, minimum height=0.6cm, font=\tiny] at (4.5, 1.5) {0.02};
\node[rectangle, fill=lowprob, draw=black, minimum width=1cm, minimum height=0.6cm, font=\tiny] at (6, 1.5) {0.06};

\node[rectangle, fill=lowprob, draw=black, minimum width=1cm, minimum height=0.6cm, font=\tiny] at (1.5, 0.5) {0.10};
\node[rectangle, fill=lowprob, draw=black, minimum width=1cm, minimum height=0.6cm, font=\tiny] at (3, 0.5) {0.04};
\node[rectangle, fill=lowprob, draw=black, minimum width=1cm, minimum height=0.6cm, font=\tiny] at (4.5, 0.5) {0.22};
\node[rectangle, fill=highprob, draw=black, minimum width=1cm, minimum height=0.6cm, font=\tiny] at (6, 0.5) {0.64};

\node[rectangle, fill=lowprob, draw=black, minimum width=1cm, minimum height=0.6cm, font=\tiny] at (1.5, -0.5) {0.32};
\node[rectangle, fill=highprob, draw=black, minimum width=1cm, minimum height=0.6cm, font=\tiny] at (3, -0.5) {0.61};
\node[rectangle, fill=lowprob, draw=black, minimum width=1cm, minimum height=0.6cm, font=\tiny] at (4.5, -0.5) {0.03};
\node[rectangle, fill=lowprob, draw=black, minimum width=1cm, minimum height=0.6cm, font=\tiny] at (6, -0.5) {0.04};

\draw[arrow] (0, 6) -- (1, 6) node[midway, above, font=\footnotesize] {Zero-Shot} node[midway, below, font=\footnotesize] {Classification};

\end{tikzpicture}
\end{figure}

The classification mapping is explained in Figure \ref{fig:zero-shot-classification}. The figure provides examples of 5 text messages from TCs and their associated LLM-based Zero-Shot classifier probabilities.\footnote{This method is implemented using the \texttt{transforEmotion} package in R \citep{christensen2024transforemotion}. Similar to the \texttt{text} package, any LLM-based Zero-Shot classifier with a pipeline on Huggingface can be implemented on the local machine without sending any data to an external server.} The four user-supplied classes are selected using iterative prompts to an AI model and further validated using human annotators and cosine similarity in the semantic space. See section \ref{BERTandBART} for further details on the label generation and human validation of labels.

Using the low-dimensional interpretable class probabilities, we set up a multivariate outcome model of peer effects that allows for contextual and endogenous peer effects and endogeneity in the network formation.  A resident is defined as a network neighbor or ``peer'' of another resident if they exchanged a written message. This construction matches the common practice of using self-nominated friends as ``peers'', and not all students in the same classroom, in popular peer effect datasets such as the Add Health or the student conflict network  in \cite{paluck2016changing}. The underlying intuition is that individuals are more likely to be influenced by a smaller set of realized interactions than by the full set of people who simply share the same institutional environment \citep{calvo2009peer,bramoulle2009identification,olivetti2020mothers}. In contrast, defining all residents in the same facility as peers would impose a much denser and less informative network structure, since co-residence captures only potential exposure and not actual interaction. 

In the same spirit, we use the ``at least one exchange'' threshold because written exchanges provide an observed, behaviorally meaningful measure of realized interaction.  The written message records are directed, since each exchange has a sender and a receiver. However, in the peer-effects specification, we use a \emph{symmetrized} version of this network. Specifically, two residents are treated as linked if there is at least one recorded exchange in either direction during the relevant observation period. This is similar to the undirected network assumption used in other related papers, such as \cite{johnsson2021estimation,olivetti2020mothers,calvo2009peer}. Drawing a parallel to the popular Add Health dataset, this choice is akin to making two students connected if at least one has nominated the other as a friend in the friend nomination survey, which is a commonly made choice in the literature.

The outcome vector on which we measure the peer effect comprises the additive log ratios of the probability vector for the residents on the behavioral input labels. However, since these behavioral profiles are LLM generated, they may lead to bias in the peer effect estimate. We therefore conduct two human surveys to collect data from paid human participants and create ``gold standard'' behavioral profiles for a subset of the participants. These gold-standard profiles are expensive to create and therefore scarce, and we do not observe them for the entire network.\footnote{To provide an approximate cost for generating human-validated profiles at scale, generating empirical probabilities for 1,000 messages on Prolific with 5 annotators for each message cost approximately \$750 for the sender-profile survey and \$750 for the receiver-profile survey. Extending this to approximately 200,000 messages for the two units would cost approximately \$300,000 in total. This shows how prohibitive the cost can be to generate human-validated profiles and motivates our approach of combining limited human-validated profiles with PPPI methods for such a large corpus of messages.} We therefore infer peer effects using the Prediction powered peer effect inference (PPPI) framework developed in this paper that uses the scarce human profiles to debias the LLM generated profiles.

\section{Methods}
\label{methodsdetails}
While the previous section described our analysis plan in the context of our data, in this section we describe our methodology in detail, over four subsections: (1) LLM embeddings, zero-shot classification, and human validation, (2) Recidivism prediction and fine-tuning, (3) Multivariate peer effects, and (4) prediction powered inference for validation and peer effect estimation.

\subsection{LLM embeddings, zero-shot classification, and human validation}
\label{BERTandBART}

We deploy pre-trained open-weight open-source LLMs on the corpus of written messages to derive two representations: (1) context-sensitive embedding vectors and (2) probability vectors for meaningful behavioral categories using zero-shot classification. We use the \textit{bert-base-uncased} model \citep{devlin2019bert} for (1), and the \textit{bart-large-mnli} model \citep{lewis2019bart} for (2). We briefly review here text embedding methods and discuss our rationale for adopting transformer-based AI models, such as BERT and BART.\footnote{BERT has 110 M parameters, BART MNLI has 400 M parameters. So they may not be considered large by current measure, but were large when they were introduced.} For details refer to Appendix \ref{attn}. 

A text embedding is a high-dimensional numerical vector representation of text. If two words convey similar meanings, they are represented closely in the semantic space. However, earlier word embedding models, such as Word2Vec \citep{mikolov2013efficient}, did not consider the context for words, limiting their ability to handle words with multiple context-specific meanings. Transformer-based models, like BERT and BART, overcome these limitations using self-attention mechanisms \citep{vaswani2017attention}, where the embedding for each word is dynamically influenced by the embedding of the entire input sequence, providing potentially different embeddings for each word based on context. The generative pre-training (GPT) framework \citep{Radford2018ImprovingLU}  used semi-supervised learning in the transformer architecture. The model is pre-trained unsupervised on a large data corpus to learn embeddings or representations, followed by supervised fine-tuning for specific tasks. The language representation model BERT in \cite{devlin2019bert} further improved the capabilities of pre-trained language models by introducing a bi-directional transformer, which can use both left and right context for a word as opposed to unidirectional processing (typically from left to right) in OpenAI GPT-1 of \cite{Radford2018ImprovingLU}.

To numerically encode the semantic content of each message, we use the output from the penultimate block of \textit{bert-base-uncased} model to construct an ``embedding profile" for each participant. This is a high-dimensional vector summarizing their language use across multiple messages during their stay in TC. We choose the penultimate layer as that layer is widely thought to contain the unsupervised representation of data in the representation learning framework \citep{bengio2013representation}, while the last layer is trained or fine-tuned to be task-specific. One can, in principle, also choose an embedding representation from an earlier layer. Let $\mb{f}^{\mathrm{token}}_{c,k} \in \mathbb{R}^H$ denote the embedding of the $k$-th token in the $c$-th message, extracted from the penultimate encoder. We obtain the message-level embedding by taking the average of these vectors over all tokens in the message as
$\mb{f}^{\mathrm{msg}}_c=\frac{1}{N_c} \sum_{k=1}^{N_c} \mb{f}^{\mathrm{token}}_{c,k}$, where $N_c$ is the number of tokens in message $c$. We then compute the ``sender embedding profile'' for a sender $s$ by taking an average of the the message-level embeddings, $\mb{f}^{\mathrm{sender}}_s=\frac{1}{N_s} \sum_{c \in \mathcal{C}_s} \mb{f}_c^{\mathrm{msg}}$, where $\mathcal{C}_s$ denotes the set of messages sent by sender $s$, and $N_s = |\mathcal{C}_s|$ is the number of such messages. We can obtain the embedding profile for each receiver analogously. 

Even though BERT embeddings are very useful for recidivism prediction, they are not suitable for peer effect estimation. The embedding coordinates are high-dimensional and do not have a direct behavioral interpretation, making it difficult to summarize what kind
of behavior is transmitted across peers or to validate the measures with human annotators. This motivates us to use zero-shot classification \citep{yin2019benchmarkingzeroshottextclassification} to map each message into a small set of
interpretable behavioral categories. In zero-shot classification, a user provides an arbitrary number of labels, and the LLM assigns given texts to the most appropriate label without any additional training.

In zero-shot learning, the classification problem is turned into a language inference problem. We use the \textit{bart-large} model fine-tuned on the MNLI (Multi-Genre Natural Language Inference) task dataset, a dataset designed for the language inference task.\footnote{The company Meta fine-tuned the \textit{bart-large} model on the MNLI dataset, and we use this pre-trained model directly from hugging face. See link \url{https://huggingface.co/facebook/bart-large-mnli}.} In MNLI, each input consists of a premise and a hypothesis, and the pre-trained BART model is trained to classify the semantic relationship between them as one of three predefined MNLI classes: entailment, neutral, or contradiction. A premise-hypothesis combination is concatenated into a single input sequence using special tokens following the format: \texttt{<s>  premise </s>  hypothesis </s>}, which allows the model to process both segments jointly while preserving their roles.

The encoder uses a self-attention mechanism that allows the model to capture contextual dependencies between tokens across both segments. This enables the model to assess how strongly the premise supports, contradicts, or is unrelated to the hypothesis. During training, the embedding of the first token \texttt{<s>} from the final encoder layer is passed to a classification head, which outputs a logit vector over the three MNLI classes. For details on the zero-shot method, see Appendix \ref{zeroshotappendix1}. 

\paragraph{Validation of Input Labels for zero-shot:} Since the empirical analysis in this paper uses LLM-generated output, we follow a rigorous pipeline shown in Figure \ref{fig:validation_pipeline} to construct, validate, and interpret the behavioral dimensions and their zero-shot probabilities. The first step is to split the messages into two data sets. The first data-set ``\textit{label generation data}" is used for generating the input labels that describe the underlying behavior of the messages exchanged in the correction units. The second data-set ``\textit{held-out data}" comprises the remaining messages that will be used for zero-shot classification using the labels generated from label generation data. This sample splitting follows the steps outlined in \cite{egami2022make,modarressi2025causal}, which shows that sample splitting is necessary for valid inference that removes the correlation between the mapping function discovery between raw text and input labels and the downstream estimation of the treatment effect.

\begin{figure}[!htbp]
\centering
\caption{Validation pipeline for input labels and zero-shot output}
\label{fig:validation_pipeline}
\begin{tikzpicture}[
    font=\footnotesize,
    line cap=round,
    line join=round,
    box/.style={
        rectangle,
        rounded corners=3pt,
        draw=black!75,
        thick,
        align=center,
        minimum height=0.95cm,
        fill=gray!6
    },
    leftbox/.style={
        box,
        text width=2.9cm,
        fill=blue!5
    },
    rightbox/.style={
        box,
        text width=2.9cm,
        fill=green!5
    },
    smallright/.style={
        box,
        text width=2.55cm,
        fill=green!5
    },
    arrow/.style={
        ->,
        thick,
        >=stealth,
        draw=black!80
    }
]

\node[box, text width=2.6cm] (all) at (0,0) {All messages};
\node[font=\scriptsize, align=center] (splitnote) at (0,-0.85) {Sample-splitting};

\node[leftbox]  (labelgen)  at (-4.5,-1.8) {Label generation\\data};
\node[leftbox]  (finalleft) at (-4.5,-3.55) {Final 4 behavioral\\labels};
\node[leftbox]  (survey1)   at (-4.5,-5.45) {Human Survey 1\\Validate input\\label names};

\node[rightbox] (heldout)    at (4.5,-1.8) {Held-out\\data};
\node[rightbox] (zeroshot)   at (4.5,-3.55) {Zero-shot\\output};
\node[rightbox] (finalright) at (2.3,-5.3) {Final 4 behavioral\\labels};
\node[smallright] (placebo)  at (6.7,-5.3) {Placebo\\labels};

\node[rightbox] (survey2)    at (0.0,-7.35) {Human Survey 2\\Validate zero-shot\\output: receiver};
\node[rightbox] (survey3)    at (4.0,-7.35) {Human Survey 3\\Validate zero-shot\\output: sender};
\node[rightbox] (affcorr)    at (8.0,-7.35) {Affirmation\\/\ correction\\prediction};

\draw[arrow] (all.south west) to[out=230,in=50] (labelgen.north);
\draw[arrow] (all.south east) to[out=310,in=130] (heldout.north);

\draw[arrow] (labelgen) -- (finalleft);
\draw[arrow] (finalleft) -- (survey1);

\draw[arrow] (heldout) -- (zeroshot);
\draw[arrow] (zeroshot.south west) to[out=235,in=55] (finalright.north);
\draw[arrow] (zeroshot.south east) to[out=305,in=125] (placebo.north);

\draw[arrow] (finalright.south west) to[out=235,in=55] (survey2.north);
\draw[arrow] (finalright) -- (survey3);
\draw[arrow] (finalright.south east) to[out=325,in=145] (affcorr.north west);

\draw[arrow] (placebo) -- (affcorr);

\end{tikzpicture}
\end{figure}

After the sample-splitting exercise, we use an AI model on the label-generation data to generate candidate labels that describe 2 positive and 2 negative behavioral dimensions of these messages. Once we obtain the candidate labels, we combine them with semantic validation and author judgment to select the final four input labels. To validate our choice of the final four input labels, we conduct a human validation survey on Prolific (\url{https://www.prolific.com/}). Human survey 1 focuses on validating the final four labels chosen on the label generation dataset. A subset of 80 messages from the label generation data was provided in the survey, along with five questions (see Table \ref{tab:survey1_questions} for the questions), to human participants on Prolific. In the first four questions, respondents were asked whether each of the final four input labels captures a meaningful behavioral dimension of the messages. Finally, the fifth question tested the joint assessment of all four labels together in capturing relevant behavioral aspects of the shown messages. Table \ref{tab:survey1_validation_main} provides a summary of the responses. We find that more than 91\% of respondents agree that each label individually captures a meaningful behavioral dimension of the messages. The joint assessment of the four labels together is slightly lower at 89.7\%. Overall, there is overwhelming evidence validating the choice of the final four input labels for zero-shot classification. We discuss further validation of the zero-shot profiles using data from the other two human surveys in subsequent sections.

\subsection{Prediction models for recidivism: LASSO, LoRA and Adapter fine-tuning}

We take multiple approaches to predicting recidivism using message text along with resident-level covariates. The first approach uses individual embedding profiles and resident covariates to predict recidivism in a penalized logistic regression lasso method. 
 Suppose $Y_i$ denotes the recidivism status for individual $i$, which takes the value $1$ if the individual recidivates and $0$ if the person does not recidivate. Therefore, we can model $Y_i$ with a Bernoulli distribution with probability of success being $p_i$ and then further model $p_i$ using the covariates and embedding profiles. However, since the embedding profiles are high-dimensional, we use a penalized logistic regression method (logistic-lasso). Suppose $\mb{X}$ is the $N\times p$ matrix of observed covariates and $\mb T$ is the $N\times H$ matrix containing the dimensions of text embeddings obtained from the LLM model as columns. Then the logistic regression model is $\log \left(\frac{p_i}{1-p_i}\right)=\bm\beta_{0}+\mb{X}_i\bm\beta_{1}+\mb T_i\bm\beta_{2} $, where $\bm\beta_1$ and $\bm\beta_2$ are $p$ and $H$ dimensional parameter vectors respectively. We optimize for $\bm\beta_1, \bm\beta_2$ by maximizing the penalized likelihood function with $\ell_1$ penalty as follows,
\[
(\hat{\bm\beta}_1, \hat{\bm\beta}_2) = \text{argmin} \{ - \ell (\bm\beta_1,\bm\beta_2, \mb{X}, \mb T) + \lambda_1 \|\bm\beta_1\|_1 + \lambda_2\|\bm\beta_2\|_1\}. 
\]
The parameters $\lambda_1, \lambda_2$ denote the penalty parameters, which may differ for the observed covariates and the embedding vectors.

A second approach performs parameter-efficient fine-tuning or transfer learning on the pre-trained BERT model to predict recidivism. We employ two popular fine-tuning approaches, namely, low rank (LoRA) \citep{hu2022lora} trainable update matrices added to the transformer blocks, and adding a small classification adapter block to the LLM \citep{houlsby2019parameter}. Both of these methods are designed to adapt a large pre-trained model to a new task by updating only a small number of parameter subsets rather than updating or finetuning all parameters of the model. In LoRA, the weight matrices $\Delta\mb W \in \mathbb{R}^{n \times m}$ that are added to the pre-trained weights ($\mb W \in \mathbb{R}^{n \times m}$) in each layer are product of two low-rank matrices ($\Delta\mb W = \mb A\mb B^\top$, where $\mb A= \mathbb{R}^{n \times r}, \mb B= \mathbb{R}^{m \times r}$), therefore, substantially reducing the trainable parameters. The adapter approach freezes the BERT encoder, and an adaptation module consisting of relatively fewer parameters is trained for the binary classification task (recidivism classification in our setting). 

Conceptually, the adapter approach is very similar to LASSO, but instead of using the embeddings in a separate LASSO model, we directly train a fully connected neural network layer for classification. Parameter-efficient fine-tuning or transfer learning is more appropriate than training a new model from scratch, since the BERT model has about 110 million parameters trained on a large corpus, and our dataset includes only 1649 individuals, even for the more populous of the two units. The LoRA method only updates about 1.5 million parameters, and the adapter approach updates about 1.2 million parameters using our data, thus heavily relying on the predictive ability of the pre-trained BERT model.

\subsection{Multivariate Peer Effects Model with Endogenous Network}
\label{multivariatepeereffectssec}
We consider the multivariate peer effect model, which captures dependencies among multiple outcomes observed across spatially or network-connected entities.  Suppose we have $N$ individuals and we observe $m$ outcome variables, and $p$ covariates for each individual. Further, we observe an undirected network among the individuals whose adjacency matrix is $\mb{A}$, where the elements $A_{ij}$ represent the relationship between individuals $i$ and $j$.  Then we define the multivariate peer effect model as:
\begin{equation}
\mb{Y} = \mb{G} \mb{Y} \mb{D} + \mb{X} \mb{B}_1 + \mb{G} \mb{X} \mb{B}_2  + \mb{E},
\label{MSARmodel}
\end{equation}
where $\mb{Y} \in \mathbb{R}^{N \times m}$ is the matrix of outcome variables, $\mb{G}=(G_{ij})\in\mathbb{R}^{N \times N}$ where $G_{ij}:=\frac{A_{ij}}{\sum_{j\ne i}A_{ij}}$, is the row-normalized counterpart of the network adjacency matrix $\mb{A}$, and  $\mb{X} \in \mathbb{R}^{N \times p}$ is the matrix of observed covariates. 
We write $\mb{Y}_i\in\mathbb{R}^m$ and $\mb{Y}_{.,j}\in\mathbb{R}^N$ to denote the $i$-th row and $j$-column vectors of an $N \times m$ matrix $\mb{Y}$, respectively. Matrix $\mb{D} \in \mathbb{R}^{m \times m}$ is the peer effect parameter matrix. The diagonal of $\mb{D}$ are the direct peer spillover effects on the same dimensions of $\mb{Y}$, and the off-diagonal elements are the indirect peer spillover effects of a dimension on another dimension.  The matrices $\mb{B}_1, \mb{B}_2 \in \mathbb{R}^{p \times m}$ are coefficient matrices for $\mb{X}$ and $\mb{G}\mb{X}$, respectively. In the context of our application, for each resident the outcome $\mb{Y}_i$ is a 3 dimensional vector of log ratios (ALRs) corresponding to categories of community support, rule violation, and personal growth relative to the baseline category of disruptive conduct \citep{aitchison1982statistical}. Since the outcome model is specified on the unrestricted ALR scale and not on the probabilities which are bounded between 0 and 1, the latent components and error terms do not violate the compositional constraint on the probability profile. The undirected network has a link $A_{ij}=1$, if $i$ and $j$ exchanged at least one written message, and $0$ otherwise. 

Since the network is endogenous, we cannot identify and consistently estimate $\mb{D}$ from this model using \cite{bramoulle2009identification}'s IV 2SLS framework. For concreteness, we assume for each individual there is a vector $\mb{U}_i$ of latent homophily variables that is correlated with both the network adjacency matrix $\mb{A}$ and the error matrix $\mb{E}$. We assume the error matrix $\mb{E} \in \mathbb{R}^{N \times m}$ is such that for each row $\mathbb{E} [ \mb{E}_i| \mb{U}_i] =\mb{h}^E(\mb{U}_i)$ for some unknown function $\mb{h}^E$, and $ \mathbb{V}[ \mb{E}_i| \mb{U}_i] = \mb{V}$. The matrix $\mb{V}$ allows dependence across outcome dimensions. The vectors $\mb{E}_i,\mb{X}_i,\mb{U}_i$ are assumed to be i.i.d. across individuals $i$. We let $A_{ij}=f(\mb{U}_i,\mb{U}_j, \mb{W}_i, \mb{W}_j, \xi_{ij})$, where $\mb{W}_i$ are covariates that enter the network formation model, and $\xi_{ij}$s over all $(i,j)$ are assumed to be i.i.d and independent of $\mb{X,W,E,U}$.

To identify the parameter matrices, we take an instrumental variable approach similar to \cite{johnsson2021estimation}, but extend the methodology in several directions as we detail below. 
We define the combined regressor and instrument matrices as 
$\mb{Z} = [\mb{G}\mb{Y},  \mb{X}, \mb{G}\mb{X}] \in \mathbb{R}^{N \times (m + 2p)}, \quad
\mb{K} = [\mb{X}, \mb{G}\mb{X}, \mb{G}^2\mb{X}] \in \mathbb{R}^{N \times 3p}$. The stacked parameter matrix is 
$\bm\beta = [\mb{D},  \mb{B}_1,  \mb{B}_2]^\top \in \mathbb{R}^{(m + 2p) \times m}$. Using this notation, equation \eqref{MSARmodel} can be written compactly as $
\mb{Y} = \mb{Z} \bm\beta + \mb{E}.$
A natural estimator for $\bm\beta$ is the two-stage least squares (2SLS) estimator:
\begin{equation}
\hat{\bm\beta}_{\text{IV}} = (\mb{Z}^\top \mb{K} (\mb{K}^\top \mb{K})^{-1} \mb{K}^\top \mb{Z})^{-1} \mb{Z}^\top \mb{K} (\mb{K}^\top \mb{K})^{-1} \mb{K}^\top \mb{Y}.
\label{IVsol}
\end{equation}

This estimator is a multivariate extension of the estimator proposed in \cite{bramoulle2009identification,kelejian1998generalized} and is a consistent estimator provided the network is exogenous, which is not the case in our setup. For any matrix-valued variable
$\mb V\in\{\mb Y,\mb Z,\mb K,\mb E\}$, let $h^V(\mb U)$ denote
the matrix whose $i$th row is $
h^V(U_i)=\E(V_i\mid U_i)$ and define the residualized matrix $
\widetilde{\mb V}=\mb V-h^V(\mb U)$. With this notation, taking expectations on both sides of Equation \ref{MSARmodel} conditioning on $\mb{U}$ and subtracting this from equation \ref{MSARmodel}, we obtain
$
\widetilde{\mb{Y}} = \widetilde{\mb{Z}} \bm\beta + \widetilde{\mb{E}}.$
 
For this redefined linear model we can solve for $\bm \beta$ using IV-2SLS with the redefined instrument matrix $\tilde{\mb{K}}$. The moment equation for this is given by
\begin{equation}
    \E [\tilde{\mb{K}}_i^\top(\tilde{\mb{Y}}_i -\tilde{\mb{Z}}_i\bm \beta \,)] =0
    \label{momenteq}
\end{equation}
The following proposition provides sufficient conditions for the identification of the true parameter matrix $\bm \beta_0$.
\begin{Proposition}
    Suppose that the conditional exogeneity assumption $
\E\!\left(
\widetilde{\mb K}_i^\top\widetilde{\mb E}_i
\right)=0$ holds. Under assumptions that the $3p \times (m+2p)$ matrix $\E [\tilde{\mb{K}}_i^\top\tilde{\mb{Z}}_i]$ has full column rank, the true parameter $\beta_0$ can be obtained by solving the $3p \times m$ moment matrix condition in Equation \ref{momenteq}.
    \label{identify}
\end{Proposition}

\begin{Assumption}
\label{ass:iv-relevance}
The limiting instrument covariance matrix
$
\Sigma_{\widetilde K}
=
\lim_{N\to\infty}
\frac{1}{N}
\E(\widetilde{\mb K}_N^\top\widetilde{\mb K}_N)
$
is positive definite, and
$
Q_0
:=
\Sigma_{\widetilde Z\widetilde K}
\Sigma_{\widetilde K}^{-1}
\Sigma_{\widetilde K\widetilde Z}
$
is positive definite, where $\Sigma_{\tilde K\tilde Z}:=\lim_{N\to\infty}\frac{1}{N}\E(\tilde{\mb{K}}_N^\top\tilde{\mb{Z}}_N)$.
\end{Assumption}

Proposition~\ref{identify} states the population rank condition
needed to identify $\beta_0$. Note that a necessary condition for the the rank condition in Proposition \ref{identify} is $p\geq m$, i.e., we must have more covariates than behavioral dimension ALRs.  Assumption~\ref{ass:iv-relevance} ensures
that the necessary and sufficient rank condition continues to hold as $N$ grows, so that the
parameters remain strongly identified in the asymptotic analysis.

\begin{remark}
\label{rem:iv-relevance}
The rank condition in Proposition~\ref{identify} can be related
directly to the structural model. Under the stability condition
$\rho(\mb D^\top\otimes \mb G)<1$, Equation~\eqref{MSARmodel} implies $
\mb Y
=
\sum_{\ell=0}^{\infty}
\mb G^\ell(\mb X \mb B_1+\mb G \mb X \mb B_2+\mb E)\mb D^\ell,$
and hence
\[
\mb G \mb Y
=
\mb G \mb X \mb B_1
+
\sum_{\ell=0}^{\infty}
\mb G^{\ell+2} \mb X( \mb B_2+ \mb B_1 \mb D)\mb D^\ell
+
\sum_{\ell=0}^{\infty}
\mb G^{\ell+1}\mb E \mb D^\ell.
\]
Thus, after accounting for the included regressor $\mb G \mb X$, the coefficient on the first excluded
instrument $\mb G^2\mb X$ is the matrix $\mb B_2+\mb B_1 \mb D$. In the scalar model of
\citet{johnsson2021estimation}, the relevance of the instruments is first motivated by requiring the coefficient on $\mb X$ to be nonzero, and later a rank condition similar to our Proposition \ref{identify} is imposed. In our
multivariate model, a single nonzero coefficient in the $\mb B_1$ matrix is not sufficient to
identify all columns of $\mb G \mb Y$. We therefore impose relevance directly through
the full-column-rank condition in Proposition~\ref{identify},
which requires the residualized instruments to identify all columns of the
residualized regressor matrix.
\end{remark}

Therefore, given access to the latent positions $\mb{U}$ and the three conditional mean functions, the parameter $\bm\beta_0$ can be estimated using 2SLS. However, the moment equation still cannot be estimated as we do not observe $\mb{U}$ and do not know the conditional expectations of $\mb{Y,Z,K}$ given $\mb{U}$. We next address how to estimate both of these components.

\paragraph{Latent positions and Non-Parametric estimation with Sieves}

As mentioned earlier, we model the network data using a latent variable model $A_{ij}=f(\mb{U}_i,\mb{U}_j, \mb{W}_i, \mb{W}_j, \xi_{ij})$ and posit that the node level latent variables $\mb{U}_i$ are involved in both the network formation model and are correlated with the error term of the outcome model. These latent variables represent unobserved characteristics of individuals, which may be responsible for tie formation (latent homophily), and an unknown function of these latent variables is part of the error term in the outcome model. Therefore, these latent variables can be estimated from the observed network. The variables $\mb{W}_i$ represent observed covariates if they are included in the network model.

This is a core assumption made in recent methodological advances in peer effect estimation \citep{mcfowland2021estimating,johnsson2021estimation,nath2022identifying}. However, while \cite{mcfowland2021estimating,nath2022identifying} focused on the longitudinal peer effect model, we consider the simultaneous peer effect model. Further, \cite{mcfowland2021estimating,nath2022identifying} assumed the latent variables enter the model for $\mb{Y}$ linearly as $\mb{U\bm\beta}$. In contrast, we assume the conditional expectations of $\mb{Y,Z,K}$ to be unknown functions $\mb{h(U)}$ of the latent variables. In this aspect, our approach resembles \cite{johnsson2021estimation}; however, we consider more general latent variable models for network data with multi-dimensional latent variables, which are more realistic for modeling unobserved individual characteristics.

Further, the theoretical results in \cite{johnsson2021estimation} require the network to be ``dense'' in the sense that every individual is expected to be connected to ${\c O}(N)$ individuals and consequently, the total number of edges in the network is expected to be ${\c O}(N^2)$. This is quite unrealistic as real-world networks tend to be sparse with individuals connected to only a few other people, even when the network size is huge (e.g., in online social networks, even if there are millions of people present in the network, an individual member only has a few hundred or thousands of connections). In contrast, our results here hold for sparse networks.

We define the stacked data matrix $\mb{W} := [\mb{Y}, \mb{Z}, \mb{K}] \in \mathbb{R}^{N \times (2m+5p)}$. The conditional mean functions are defined as $\mb{h}^Y(\mb U) := \mathbb{E}[\mb Y \mid \mb U]$, $\mb{h}^Z(\mb U) := \mathbb{E}[\mb Z \mid \mb U]$, and $\mb{h}^K(\mb U) := \mathbb{E}[\mb K \mid \mb U]$. These are collected into the full conditional mean matrix $\mb{h}(\mb U) := [\mb{h}^Y(\mb U), \mb{h}^Z(\mb U), \mb{h}^K(\mb U)] \in \mathbb{R}^{N \times (2m+5p)}$. The function $\mb{h}: \mathbb{R}^d \to \mathbb{R}^{2m+3p}$ is an unknown non-linear function of $\mb{U}_i$. Let \( h_j(\cdot):\mathbb{R}^N\to\mathbb{R} \) denote the \( j \)th component of \( \mb{h}(\cdot) \), for \( j = 1, \dots, (2m+5p) \). Then,  for each \( \mb U_i \), \( h_j(\mb U_i) \in \mathbb{R} \) is the scalar value of the \( j \)th conditional mean component. 
In what follows, to keep the notation simple, let $\mb{u}$ denote a generic row $\mb{U}_i$.

To estimate each component function $h_j(\mb u)$, we use tensor-product basis functions \citep{zhang2023regression} given by $\phi_k(\mb u) = \prod_{\ell=1}^{d} \psi_{k_\ell}^{(\ell)}(u_\ell)$, where each $\psi_{k_\ell}^{(\ell)}$ is a univariate basis function (polynomial or cosine) on the $\ell$-th coordinate of $\mb{u}$.  With these basis functions, we approximate each component function \( h_j(\mb u) \), by a linear combination of basis functions:
\[
h_j(\mb u) \approx \sum_{k=1}^{L_N} \phi_k(\mb u) \, \alpha_k^j = \sum_{k=1}^{L_N} \prod_{\ell=1}^{d} \psi_{k_\ell}^{(\ell)}(u_\ell) \, \alpha_k^j ,
\]
where \( \bm{\alpha}^j = (\alpha_1^j, \dots, \alpha_{L_N}^j)^\top \in \mathbb{R}^{L_N} \) is the coefficient vector, and the number of basis functions to use $L_N$ is possibly a function of $N$. Note this formulation allows for different univariate basis functions indexed by $\psi_{k_\ell}$ for each component basis $\phi_k$ and each dimension of $u_l$. It also allows for interactions among the predictors.

For the generic vector $\mb{u}$, let \(
\bm{\phi}^{L_N}(\mb u) := \left( \phi_1(\mb u), \dots, \phi_{L_N}(\mb u) \right)^\top \in \mathbb{R}^{L_N}
\). Then we construct the design matrix for all $N$ observations as:  $\bm\Phi_N := \left[\bm{\phi}^{L_N}(\mb u_1), \ldots, \bm{\phi}^{L_N}(\mb u_N)\right]^\top \in \mathbb{R}^{N \times L_N}$. The coefficient vector \( \bm{\alpha}^j \) is obtained by regressing the observed vector \( \mb{w}_{\cdot j} \) on the basis matrix \( \bm\Phi_N \) via ordinary least squares, as $\hat{\bm{\alpha}}^j = (\bm\Phi_N^\top \bm\Phi_N)^{-1} \bm\Phi_N^\top \mb{w}_{\cdot j}$. The fitted values for the \( j \)th component function at all sample points are then given by: $
\hat{h}_j(\mb U) := \bm\Phi_N \hat{\bm{\alpha}}^j = \mb{P}_{\bm\Phi_N} \, \mb{w}_{\cdot j},$ where \( \mb{P}_{\bm\Phi_N} := \bm\Phi_N (\bm\Phi_N^\top \bm\Phi_N)^{-} \bm\Phi_N^\top \) is the projection matrix associated with the basis space, and \( A^{-} \) denotes any symmetric generalized inverse of $A$.

Since latent positions $\mb U$ are unknown, we use a two-step procedure. For estimating $\mb{U} $, we model the network using latent variable models such as the RDPG model \citep{athreya2017statistical,rubin2022statistical,xie2023efficient} and the additive and multiplicative effects latent space model \citep{ma2020universal,hoff2021additive,hoff2002latent,li2023statistical}. 
Then our two-stage procedure is as follows. 
(i) we first construct an estimator \( \hat{\mb U} := (\hat{\mb u}_1, \dots, \hat{\mb u}_N)^\top \) for the latent traits using spectral embedding for RDPG models or maximum likelihood estimation for additive and multiplicative effects latent space models, and  
(ii) we evaluate the basis functions at \( \hat{\mb U} \) and apply the same projection strategy. With these estimated latent vectors, we define the estimated design matrix as $
\hat{\bm\Phi}_N := \bm\Phi_N(\hat{\mb U}) = 
\left[
\bm{\phi}^{L_N}(\hat{\mb u}_1),
\ldots,
\bm{\phi}^{L_N}(\hat{\mb u}_N) \right]
\in \mathbb{R}^{N \times L_N}.$
Then the estimated conditional mean functions are $\hat{\mb{h}}^W(\hat{\mb U}) := \mb{P}_{\hat{\bm\Phi}_N}\bm W$ where $\mb{P}_{\hat{\bm\Phi}_N} := \hat{\bm\Phi}_N (\hat{\bm\Phi}_N^\top \hat{\bm\Phi}_N)^{-1} \hat{\bm\Phi}_N^\top$.

Denote  $\mb{M}_{\hat{\bm\Phi}_N} = \mb{I}_N - \mb P_{\hat{\bm\Phi}_N}$.
Then, our Two-Stage Least Squares (2SLS) estimator is:
\begin{align}
    \hat{\bm\beta}_{2SLS} = \left(\mb{Z}_N^\top \mb{M}_{\hat{\bm\Phi}_N} \mb{K}_N \left(\mb{K}_N^\top \mb{M}_{\hat{\bm\Phi}_N} \mb{K}_N\right)^{-1} \mb{K}_N^\top \mb{M}_{\hat{\bm\Phi}_N} \mb{Z}_N\right)^{-1}\nonumber \\
    \times  \mb{Z}_N^\top \mb{M}_{\hat{\bm\Phi}_N} \mb{K}_N \left(\mb{K}_N^\top \mb{M}_{\hat{\bm\Phi}_N} \mb{K}_N\right)^{-1}\mb{K}_N^\top \mb{M}_{\hat{\bm\Phi}_N} \mb{Y}_N.\label{eq-2sls}
\end{align}

\paragraph{Network models}
\label{sec:lsm}
In terms of latent variable network models, we first consider the Latent Space Model (LSM), which is a general random graph model that includes both additive and multiplicative latent variables and can also accommodate covariates \citep{hoff2002latent,hoff2021additive,ma2020universal,li2023statistical}. The model is parameterized by $d$-dimensional unknown vectors $\mb q_i$ and unknown scalars $v_i$ for all nodes $i=1, \ldots, n$. We also assume that we have edge-level covariates $\mb{X}= (x_{ij})$ available to us. A network adjacency matrix $\mb{A}= (a_{ij})$ from this model is generated as follows \citep{li2023statistical,ma2020universal}:
\[a_{ij}\overset{ind.}{\sim}f_{ij}:= f(a;\theta_{ij}),\qquad \theta_{ij}:=\sigma(\mb q_i^\top\mb q_j+v_i+v_j + x_{ij}\beta), \quad 1\leq i<j\leq n\]
where $f$ is a family of distributions satisfying certain smoothness conditions, and $\sigma:\mathbb{R}\to\mathbb{R}$ is a known link function. In this paper, we consider a sparse version of the model  that assumes $f$ to be a Bernoulli distribution, link function $\sigma$ to be the logistic function, and there exists a sparsity parameter $\rho_N$ such that $\theta_{ij}=logistic(\mb q_i^\top\mb q_j+v_i+v_j + x_{ij}\beta +\rho_N)$. Defining $\omega_N = \exp(\rho_N)$, we assume $\omega_N \to 0$ and $\omega_N=\omega(N^{-1/2})$.
We let $\mb{U}_i\triangleq(\mb q_i^\top,v_i)^\top$ denote the $d+1$ dimensional parameter vector containing all latent variables associated with node $i$.

A special case of this model is Random Dot Product Graph (RDPG) model, which can be expressed as follows,
\[(a_{ij}|\mb{\eta}_i,\mb{\eta}_j)\overset{ind.}{\sim}\operatorname{Bernoulli}(\rho_N \mb{\eta}_i^\top\mb{\eta}_j),\]
where $\rho_N$ again controls the sparsity of the model. For notational convenience we define $\mb{U}_i= \rho_N^{1/2}\mb{\eta}_i$, and will assume $\rho_N =\omega(\frac{\log^4N}{N})$. Note the expected density (i.e., probability of an edge) scales with $N$ as $\rho_N$, or in other words, the expected degree scales as $N\rho_N$ in this model, making the model suitable for sparse networks. This model does not account for the additive latent variables $v_i$'s but contains $d$-dimensional multiplicative latent variables, and the link function $\sigma(\cdot)$ is specified as identity. Clearly, both these models allow for multidimensional latent variables in the network formation model and are more general than the single latent variable model considered in \cite{johnsson2021estimation}. Both models are also capable of modeling sparse networks, which the models in \cite{johnsson2021estimation,graham2017econometric,auerbach2022identification} are not capable of. We compare the model fit of the latent-space network models with the network model considered in \cite{johnsson2021estimation,graham2017econometric} for the network data in our real-data analysis application (see section \ref{comparisonofnet}).

We employ the maximum likelihood estimator with Lagrange adjustment in \citet{li2023statistical} for the estimation of the LSM model.
Let $\hat{\mb{U}}$ denote the maximum likelihood estimator over the constrained parameter space 
\[\Xi_d:=\{\mb{U}\in\mathbb{R}^{N\times(d+1)} \,;\,{\mb Q}^\top\bm 1_N = \bm 0_d, \,{\mb Q}^\top{\mb Q} \text{ is diagonal, } \|\mb{U}\|_{2,\infty} = {\c O}(1)\text{ as }N\rightarrow \infty\}.\]

 Under the set of assumptions laid out in \cite{li2023statistical}, an upper bound on the error rate of estimating the latent positions in Frobenius norm \citep{li2023statistical} is $\frac{1}{N}\|\hat{\mb{U}}-\mb{U}\|_F^2=O_\bbP(\frac{1}{N\omega_N})$.

There is a rich amount of results with better convergence rates of the estimate of $\mb{U}$ when one uses spectral embedding to estimate $\mb{U}$ in the case of RDPG model \citep{athreya2017statistical,cape2019signal,xie2023efficient,chang2024embedding}.
Specifically, the spectral embedding of the adjacency matrix $A$ can be written as 
$\hat{\mb{U}} = \mb{U}_A|\bm S_A|^{1/2}$,
where $\mb U_A$ denotes the matrix containing the leading $d$ eigenvectors of $\mb A$, which are assumed, without loss of generality, to be ordered by decreasing absolute eigenvalue magnitude. The matrix $|\mb S_A|$ is a diagonal matrix containing the absolute values of the corresponding eigenvalues.
Then, a convergence can be obtained in terms of the $2\to \infty$ norm (maximum of row-wise $\ell_2$ norms) as  $\|\hat {\mb{U}} - \mb{U}\mb{h}\|_{2,\infty}  = {\c O}_{\t\bbP}\left(\frac{\log^c N}{N^{1/2}}\right)$ \citep{rubin2022statistical,cape2019signal}.

Our method, combining all the above components, is described in Algorithm \ref{algorithm1}. Note that the step estimating $\hat{U}$ with spectral embedding in the algorithm will be replaced with MLE when the latent space model is used for the underlying  network. We discuss two additional methodological and theoretical issues concerning identifiability under orthogonal transformations and asymptotic relevance of instruments in Remarks \ref{orthogonal}  and \ref{relevance} respectively in the Appendix \ref{details_methods}.

\subsection{Prediction-powered peer-effect inference}
\label{sec:ppi-peer}

The peer-effect analysis uses resident behavioral profiles obtained from
a zero-shot classifier. The expensive human labels are available only for a
sample of residents and, within those residents, a subset of messages. To combine the limited human-validated profile data with LLM data, we use
prediction-powered inference (PPI) \citep{angelopoulos2023prediction} in two stages. The message-level PPI combines
human data on sampled messages with BART probabilities on the resident's
complete message set. The resident-level PPI then uses the human-calibrated
outcomes of the validated residents to rectify the IV moment $\mb g_N(\bm\beta;\mb Y)$ used to estimate the peer effect parameter (defined in Equation \ref{eq:generic-moment}). We call this new method the Prediction Powered Peer Inference (PPPI).

\subsubsection{Human-calibrated resident outcomes}
\label{sec:message-ppi}

Consider a resident $i$ who has sent messages $j=1,\ldots,M_i$. Let
$\mb L_{ij}\in\mathcal S^{C-1}$ denote the $C$-dimensional LLM probability vector
for message $j$, where
$\mathcal S^{C-1}=\{\mb q\in[0,1]^C:\sum_{c=1}^C q_c=1\}$. For a
human-annotated message, let $\mb H_{ij}\in\mathcal S^{C-1}$ denote the human
consensus vector, obtained from the empirical proportions of the human
classifications. We let $\{\mb H_{ij}:j=1,\ldots,M_i\}$ denote the
finite population of human-consensus labels for all messages of resident $i$,
whether or not a particular message is selected for human annotation.\footnote{Here we use $H_{ij}$ both to denote the human consensus vector observed using the Prolific survey as well as the finite population of human consensus vector on all messages if humans were to annotate every message.} The LLM profile based on all messages of resident $i$ is
$\mb q_i^{\AI}=M_i^{-1}\sum_{j=1}^{M_i}\mb L_{ij}$, and  define the corresponding full human profile as
$
\mb q_i^{\Human}=M_i^{-1}\sum_{j=1}^{M_i}\mb H_{ij}.
$

For a resident selected in the Prolific survey, let $\mathcal S_i$ be the set of
$s_i$ messages sampled uniformly without replacement from $M_i$.
The message-level PPI vector is
\begin{equation}
\widehat{\mb q}_i^{\Human}
=
\mb q_i^{\AI}
+
\frac{1}{s_i}\sum_{j\in\mathcal S_i}\left(\mb H_{ij}-\mb L_{ij}\right),
\label{eq:message-ppi}
\end{equation}
which is the standard PPI correction to a sample mean \citep{angelopoulos2023prediction}.
Intuitively, this suggests that the complete LLM profile is debiased, or corrected, by the average human-LLM discrepancy among the validated messages. 

At this stage, we can condition on the complete finite population
$\{(\mb H_{ij},\mb L_{ij}):j=1,\ldots,M_i\}$ and take expectation only with respect
to the random message sample $\mathcal S_i$. Since $\mathcal S_i$ is a simple
random sample without replacement,
\begin{equation}
\E_{\mathcal S_i}\!\left[
\widehat{\mb q}_i^{\Human}
\,\middle|\,
\{\mb H_{ij},\mb L_{ij}\}_{j=1}^{M_i}
\right]
=
\mb q_i^{\AI}
+
\frac{1}{M_i}\sum_{j=1}^{M_i}(\mb H_{ij}-\mb L_{ij})
=
\mb q_i^{\Human}.
\label{eq:message-ppi-unbiased}
\end{equation}
Thus, the message-level PPI estimator is unbiased for the full human profile of resident $i$. Because the $\mb H_{ij}$ and $\mb L_{ij}$ vectors each sum to one, the correction in \eqref{eq:message-ppi} sums to zero over its coordinates, so the $\widehat{\mb q}_i^{\Human}$ vector remains compositional.

Let category $C$ be the reference category and let
$\mb T(\mb q)=(\log(q_1/q_C),\ldots,\log(q_{C-1}/q_C))^\top$ denote the additive
log-ratio (ALR). To avoid zero probabilities, fix $\varepsilon>0$, and let
$\mathcal{S}^{C-1}_\varepsilon = \left\{\mb q \in \mathcal{S}^{C-1} : \min_{c \le C} q_c \ge \varepsilon \right\}$ and
let $\mathcal P_\varepsilon(\mb q)
=
\argmin_{\mb u\in\mathcal S^{C-1}_\varepsilon}\norm{\mb u-\mb q}_2$ denote Euclidean projection onto
$\mathcal S^{C-1}_\varepsilon$. The AI-only ALR outcome, available for every
resident, and the human-calibrated ALR for a validated resident are
\begin{equation}
\mb \mb Y_i^{\AI}=\mb T\!\left\{\mathcal P_\varepsilon(\mb q_i^{\AI})\right\}\in\R^m,
\qquad
\widehat{\mb Y_i}^{\Human}
=\mb T\!\left\{\mathcal P_\varepsilon(\widehat{\mb q}_i^{\Human})\right\},
\qquad m=C-1.
\label{eq:YA-YH}
\end{equation}
If $\min_c q_{ic}^{\Human}>\varepsilon$ the projection is inactive. It also provides
a uniform bound
$\norm{\Delta_i}_\infty\le2\log(1/\varepsilon)$ used in
Assumption~\ref{ass:validation-regularity} below.

\begin{remark}[Conditioning on the first validation stage]
\label{rem:fixed-human}
There are two distinct sampling stages in the PPPI construction. The
unbiasedness result in \eqref{eq:message-ppi-unbiased} concerns the first
stage and takes expectation with respect to the random sample of messages
$\mathcal S_i$, which treats the complete collection of LLM and human message
labels as a fixed finite population. From now on, for the resident-level PPPI method described below, we condition on the
realized first-stage message samples and human classifications. Conditional
on this information, the resident-level human-calibrated ALR
outcomes are treated as fixed gold-standard outcomes. Consequently, whenever resident $i$ is validated, we write the realized
$\mb Y_i^{\Human}=\widehat{\mb Y_i}^{\Human}$ for those residents. The collection
$\{\mb Y_i^{\Human}:i=1,\ldots,N\}$ is regarded as a fixed finite population,
although $\mb Y_i^{\Human}$ is observed only when resident $i$ is selected for
validation. All
asymptotics that follow are conditional on the realized message sample and human
classifications. Therefore, the only source of uncertainty in the
peer-effect standard errors is through the sampling of residents. 
\end{remark}

\subsubsection{Resident-level PPI and moment correction}
\label{sec:resident-ppi}

Let $R_i$ denote the indicator variable that resident $i$ is validated, meaning $i$'s messages were chosen for the human survey. In the empirical application, the validation
sample is a simple random sample without replacement of $n$ residents
from the $N$ residents in a facility. Define the propensity of selecting resident $i$ as $\pi_i=\Pbb(R_i=1)$, which in out setup is $:f_N=:n/N$. Define
the human--AI discrepancy vector for resident $i$ as $\Delta_i=\mb Y_i^{\Human}-\mb Y_i^{\AI}\in\R^m$ and define the
resident-level PPI pseudo-outcome
\begin{equation}
\widehat{\mb Y_i}^{\PP}=\mb Y_i^{\AI}+\frac{R_i}{\pi_i}\Delta_i,
\qquad i=1,\ldots,N.
\label{eq:Yppi}
\end{equation}
This pseudo-outcome satisfies $\E_R(\widehat{\mb Y_i}^{\PP})=\mb Y_i^{\Human}$ conditional on the
complete finite population. For an unvalidated resident
$\widehat{\mb Y_i}^{\PP}=\mb Y_i^{\AI}$, while for a validated resident the discrepancy $\Delta_i$ is
amplified by the inverse inclusion probability $\pi_i$. We emphasize that this quantity is a
survey-weighted pseudo-outcome that is later used to rectify a population estimating
equation, and is not to be interpreted as a behavioral profile.

Now we define the concept of the human-reference peer-effect model as
\[
\mb Y^{\Human}=\mb G\mb Y^{\Human}\mb D_0+\mb X\mb B_{10}+\mb G\mb X\mb B_{20}+\mb E,
\]
with $\mb D_0\in\R^{m\times m}$ the true human-reference peer-effect matrix we wish to infer. This is the model we would fit if we \textit{had human-labeled data on all residents}, which is impractical in our empirical application given the cost of running a human survey on the large corpus of TC messages, and can similarly be impractical in other settings. With
$\mb Z(\mb Y)=[\mb G\mb Y,\mb X,\mb G\mb X]$, $\mb K=[\mb X,\mb G\mb X,\mb G^2\mb X]$ and
$\bm\beta_0=[\mb D_0,\mb B_{10},\mb B_{20}]^\top$ as in Section \ref{multivariatepeereffectssec},
and with $\mb M_{\widehat{\mb \Phi}_N}=\mb I_N-\mb P_{\widehat{\mb \Phi}_N}$ the estimated sieve
residualizer, define the sample IV moment
\begin{equation}
\mb g_N(\bm\beta;\mb Y)
=\frac{1}{N}\mb K^\top \mb M_{\widehat{\mb \Phi}_N}
\left(\mb Y-\mb G\mb Y\mb D-\mb X\mb B_1-\mb G\mb X\mb B_2\right)\in\R^{3p\times m}.
\label{eq:generic-moment}
\end{equation}
We use $\mb g_N^{\Human}(\bm\beta)=\mb g_N(\bm\beta;\mb Y^{\Human})$, as the full human data moment and
$\mb g_N^{\AI}(\bm\beta)=\mb g_N(\bm\beta;\mb Y^{\AI})$ as the full AI data moment.

Let $\mb e_i$ be the $i$th canonical basis vector in $\R^N$ and define the
resident-wise effective correction and its IV-moment contribution as,
\begin{equation}
\Lambda_i(\mb D)=\mb e_i\Delta_i^\top-\mb G\mb e_i\Delta_i^\top\mb D\in\R^{N\times m},
\qquad
\Psi_i(\beta)=\mb K^\top \mb M_{\widehat{\mb \Phi}_N}\Lambda_i(\mb D)\in\R^{3p\times m}.
\label{eq:Lambda-Psi}
\end{equation}
The single row $\mb e_i\Delta_i^\top$ enters the peer lag through $\mb G$, so
$\Lambda_i$ exhibits how a correction to one resident propagates to it's neighbors. Since
$\mb Y^{\Human}=\mb Y^{\AI}+\sum_{i}\mb e_i\Delta_i^\top$ and the residual in
\eqref{eq:generic-moment} is linear in $\mb Y$,
$\mb g_N^{\Human}(\bm\beta)=\mb g_N^{\AI}(\bm\beta)+N^{-1}\sum_i\Psi_i(\bm\beta)$, and the
prediction-powered moment is
\begin{equation}
\widehat{\mb g}_N^{\PP}(\bm\beta)
=\mb g_N^{\AI}(\bm\beta)+\frac{1}{N}\sum_{i=1}^N\frac{R_i}{\pi_i}\Psi_i(\bm\beta).
\label{eq:ppi-moment-rectified}
\end{equation}

\begin{Proposition}[Unbiasedness of the PPI moment]
\label{prop:moment-unbiased}
Conditional on the complete finite population
$(\mb Y^{\Human},\mb Y^{\AI},\mb X,\mb G,\widehat{\mb\Phi}_N)$ and for every fixed $\bm\beta$,
$\E_R\{\widehat{\mb g}_N^{\PP}(\bm\beta)\}=\mb g_N^{\Human}(\bm\beta)$. Moreover
$\widehat{\mb g}_N^{\PP}(\bm\beta)=\mb g_N(\bm\beta;\widehat{\mb Y}^{\PP})$.
\end{Proposition}

Proposition~\ref{prop:moment-unbiased} implies that the PPI-corrected IV moment is the same as the ordinary IV moment with the responses being set as $\widehat{\mb Y}^{\PP}$. However, the asymptotic properties of this estimator are not captured by the existing theory. Hence, the standard errors in Algorithm~\ref{algorithm1} are not valid, and we develop new asymptotic properties of the PPPI estimator.  Writing
$\widehat{\mb Z}^*=\mb M_{\widehat{\mb \Phi}_N}\widehat{\mb Z}^{\PP}$ (for $\widehat{\mb Z}^{\PP}:=\mb Z(\widehat{\mb Y}^\PP)$),
$\widehat{\mb K}^*=\mb M_{\widehat{\mb \Phi}_N}\mb K$ and
$\mb P_{\widehat{\mb K}^*}=\widehat{\mb K}^*\{(\widehat{\mb K}^*)^\top\widehat{\mb K}^*\}^{-1}
(\widehat{\mb K}^*)^\top$, the latent-adjusted PPI--2SLS estimator is
\begin{equation}
\widehat{\bm\beta}_{\PP}
=\mathcal A_N^{-1}(\widehat{\mb Z}^*)^\top \mb P_{\widehat{\mb K}^*}
\mb M_{\widehat{\mb \Phi}_N}\widehat{\mb Y}^{\PP},
\qquad
\mathcal A_N=(\widehat{\mb Z}^*)^\top \mb P_{\widehat{\mb K}^*}\widehat{\mb Z}^*.
\label{eq:ppi-2sls}
\end{equation}
with first $m\times m$ block $\widehat{\mb D}_{\PP}$.

\section{Theoretical results}\label{sec:assmp}

\subsection{General peer effect inference}
For a matrix $\mb{X}$ we denote its spectral norm as $\|\mb{X}\|=\max_{\|x\|=1}\|\mb{X}x\|$, Frobenius norm as $\|\mb{X}\|_F=\sqrt{\sum_{ij}x_{ij}^2}$ and the $ 2 \to \infty$ norm as $\|\mb{X}\|_{2, \infty}=\max_{i}\sqrt{\sum_jx_{ij}^2}$. We use the stochastic order notation $o_\bbP(1)$ to mean that if $x_N = o_\bbP(1)$, then $x_N \overset{p}{\to} 0$. 
We start with a list of assumptions. The first assumption is a collection of conditions described above in the description of the model, making the model well-behaved and identification possible, similar to those in  \citep{johnsson2021estimation}.
\begin{Assumption}\label{assmp:jnm}
    The vectors $(\mb{X}_i,\mb{U}_i,\mb{W}_i,\mb{E}_i)$ are i.i.d. across individuals $i=1, \ldots, N$,  $A_{ij}=f(\mb{U}_i,\mb{U}_j,\mb{W}_i,\mb{W}_j,\xi_{ij})$, where $\xi_{ij}$s over all $(i,j)$ are i.i.d and independent of $\mb{X,E,U,W}$, and $\E(\mb{E}_i|\mb{X}_i,\mb{W}_i,\mb{U}_i)=\E(\mb{E}_i|\mb{U}_i)$.
    In addition, let $\E(\mb{X}_1)=\bm 0_p$, and we have a universal constant $M>0$ such that 
     $\Pr\left(\|\mb{E}_1\|_{\infty}\le M\right)=\Pr\left(\|\mb{X}_1\|_\infty\le M\right)=1.$
\end{Assumption}
We note again that the network $\mb{A}$ is endogenous by the dependence between $\mb{U}_i$ and $\mb{E}_i$. A consequence of Assumption \ref{assmp:jnm} is that
\begin{align}
    \E(\mb{E}_i|\mb{X},\mb{W},\mb{A},\mb{U}_i)& =\E(\mb{E}_i|\mb{X}_i,\mb{W}_i,\mb{A}(\mb{U}_i,\mb{W}_i,\mb{U}_{-i},\mb{W}_{-i},\xi_{ij}),\mb{U}_i,\mb{X}_{-i},\mb{W}_{-i},) \nonumber \\
    & =\E(\mb{E}_i|\mb{X}_i,\mb{W}_i,\mb{U}_i)=\E(\mb{E}_i|\mb{U}_i),
\end{align}
i.e., covariates and observed network are uncorrelated with the model error conditional on the latent variables $\mb{U}_i$.

\begin{Assumption}\label{assmp:sieve}
    Assume the following on the Sieve estimation.
    \begin{enumerate}
        \item[(i)] $L_N=o(N)$.
        \item[(ii)] We have
        \[\lim_{N\to\infty}\bm v^\top\E[\bm{\phi}^{L_N}\left(\mb{U}_1\right) \bm{\phi}^{L_N}\left(\mb{U}_1\right)^\top]\bm v=\lim_{N\to\infty}\sum_{1\le i,j\le L_N}v_iv_j\E\left[\prod_{1\le k,l\le d} \psi_{i_k}^{(k)}(u_{1k})\psi_{j_l}^{(l)}(u_{1l})\right]>0\]
        for any $\bm v=(v_1,\dots,v_{L_N})\in S^{L_N-1}$.
        \item[(iii)] Almost surely, there exists $\zeta_0\left(L_N\right)$ such that 
        \[\sup _{\mb{U}\in \mathcal{U}} \|\bm{\phi}^{L_N}(\mb{U})\|^2=\sup _{\mb{U}\in \mathcal{U}} \sum_{j\le L_N}\left(\prod_{k\le d}\psi^{(k)}_{j_k}(u_k)\right)^2\le
        \zeta_0\left(L_N\right)^2~\&~\frac{\zeta_0(L_N)^2L_N}{N}=o(1),\]
        and a sequence $\boldsymbol{\alpha}_{L_{N}}^j$ and a number $\kappa>0$ such that
        $$
        \sup_{f\in\{Y,Z,K\}}\sup _{j\in {\c I}(f)}\sup _{\mb{U}\in \mathcal{U}}\left|\mb h^f_j(\mb{U})-\bm{\phi}^{L_N}(\mb{U})^\top\bm\alpha_{L_{N}}^j\right|={\c O}\left({L_N}^{-\kappa}\right)
        $$
        and ${L_N}^{-\kappa}=o(1/\sqrt{N})$.
    \end{enumerate}
\end{Assumption}
The above assumption is an extension of the assumption in \cite{johnsson2021estimation} to multivariate basis functions which says the true component-wise functions $h_j(u)$ are well approximated by the sieve approximations for all vectors $u$ and all components $j$.

\begin{Assumption}\label{assmp:lipsch}
    There exists a positive number $\zeta_{1}(k)$ such that for all $k=1, \ldots, L_{N}$,
    \begin{align*}
        |\phi_{k}(\mb{U})-\phi_{k}\left(\mb{U}^\top\right)| =\left|\prod_{l\le d}\psi_{k_l}^{(l)}(u_l)-\prod_{l\le d}\psi_{k_l}^{(l)}(u_l')\right| \leq \zeta_{1}(k)\left\|\mb{U}-\mb{U}^\top\right\|
     \end{align*} 
     and     
\begin{align*}
    \{1\vee\zeta_{0}\left(L_N\right)^2\}\sum_{k=1}^{L_N} \zeta_{1}^{2}(k) =\begin{cases}
    o\left(\frac{N\rho_N}{\log^{2c}N}\right)&\text{if }\mb{A}\sim RDPG\\
    o\left(N\omega_N\right)&\text{if }\mb{A}\sim LSM.
\end{cases}
\end{align*}

\end{Assumption}
Our last assumption restricts $\{\phi_k;k\le L_N\}$ to be a $\zeta_1(k)$-Lipschitz where the order of $\zeta_1(1),\dots,\zeta_1(L_N)$ is determined by the random graph model and appropriate sparsity parameters we assume.

We derive the asymptotic distribution of the 2SLS estimator \( \hat{\bm\beta}_{v,\text{2SLS}} \). 
Here, we write $\dmin:=\min_{i\le N}\sum_{j\ne i}A_{ij}$ to denote the minimum node degree. Let $P_{\max} := \max_{i,j} \bbP(A_{ij}=1 \mid \mb U, \mb X)$ denote the maximum connectivity.
\begin{Theorem}\label{thm:clt}
    Assume the conditions in Section \ref{sec:assmp}. Let the graph $\mb{A}$ be generated via either the RDPG model with sparsity parameter $\rho_N=\omega(\frac{\log^4N}{N})$ or the LSM with sparsity parameter $\omega_N=\omega(N^{-1/2})$, as described in Section \ref{sec:lsm}.
    Further assume there exists a deterministic sequence $d_N>0$ such that:
    \begin{align}
        &\bbP(\delta_{\min} < d_N)=o(1) \quad \text{and} \quad \max(NP_{\max}^2,1)N^6P_{\max}^6=o(d_N^8),\label{eqn:clt-c1}\\
        &NP_{\max}=\omega(1)\quad\text{and}\quad\E(1/\delta_{\min}^2)=o(1).\label{eqn:clt-c2}
    \end{align}
    Then, as $N\to\infty$,
    \[\sqrt{N}\left( \hat{\bm\beta}_{v,2SLS} - \bm\beta_v^0 \right)\Rightarrow {\c N}\left[\bm0_{m(m+2p)},\mb{V}\otimes\left\{(\Sigma_{\tilde Z\tilde K}\Sigma_{\tilde K}^{-1}\Sigma_{\tilde K\tilde Z})^{-1}\right\}\right]\]
    where $\Sigma_{\tilde Z\tilde K},\Sigma_{\tilde K}$, $\Sigma_{\tilde K\tilde Z}$ are defined in Lemma \ref{lem:cov-cons} in the Appendix.
\end{Theorem}
The proofs of all results are in the Appendix \ref{mathproofs}. Several additional technical lemmas are needed, which are all stated and proved in the Appendix \ref{mathproofs}. Note that the deterministic sequence $d_N$ has the interpretation of being a high-probability lower bound on the minimum degree of a node in the graph since $\bbP(\delta_{\min} < d_N)=o(1)$ by condition \eqref{eqn:clt-c1}. The second part of the condition \eqref{eqn:clt-c1} can be thought of as providing a lower bound on $d_N$ in terms of the number of nodes and maximum connectivity of the graph. The conditions in \eqref{eqn:clt-c2} are relatively mild and easily satisfied. The conditions \eqref{eqn:clt-c1} and \eqref{eqn:clt-c2} are required for graphs generated from both RDPG and LSM models, while the difference between density requirements in RDPG and LSM conditions in terms of $\rho_N$ and $\omega_N$ is clarified further later in Lemma \ref{lmm:sieve:est}.

\begin{remark}

The conditions \eqref{eqn:clt-c1} and \eqref{eqn:clt-c2} on the graph density indicate that consistent and asymptotically normal estimation is possible even for \textit{sparse graphs}.  To see this, consider the following special cases. If $\mb{A}$ follows an RDPG model with sparsity parameter $\rho_N$, then $P_{\max} = {\c O}(\rho_N)$ and maximum expected degree is ${\c O}(N\rho_N)$. At the boundary of feasible range for $\rho_N$, where $\rho_N \asymp \frac{\log^4N}{N}$, \eqref{eqn:clt-c1} reduces to $d_N =\omega(\log^3N)$. Therefore, consistent and asymptotically normal peer effect estimation is possible as long as the minimum degree of the graph grows faster than $\omega(\log^3N)$ with high probability, and average expected degree grows faster than $\omega(\log^4N)$.  On the other hand, when $\mb{A}$ follows an LSM, by the LSM link function, $P_{\max}={\c O}(\omega_N)$. At the boundary where $\omega_N \asymp N^{-1/2}$, the condition \ref{eqn:clt-c2} evaluates to $d_N = \omega(N^{3/8})$. Therefore for LSM, the result in Theorem \ref{thm:clt} holds as long as the minimum degree grows faster than $N^{3/8}$ with high probability and the average expected degree grows faster than $N^{1/2}$. This is in contrast to the results in \cite{johnsson2021estimation}, which only hold if the average expected degree of the graph grows as ${\c O}(N)$, making the graph unrealistically dense. In comparison, our results, especially for the RDPG model, have the advantage of being valid for much sparser graphs. 
\end{remark}

An intermediate lemma en route to proving the main theorem further clarifies the distinction between the density requirements for the RDPG and the LSM models. 
\begin{Lemma}\label{lmm:sieve:est}
    Let Assumption \ref{assmp:lipsch} hold.
    Under the conditions on RDPG model in \cite{rubin2022statistical} and sparsity parameter $\rho_N=\omega(\frac{\log^4N}{N})$, we have
    \[\frac{1}{N}\|\hat{\bm\Phi}_N-\bm\Phi_N\|_F^2={\c O}_{\t\bbP}\left(\frac{\log^{2c}N}{N}\sum_{k\le L_N}\zeta_1(k)^2\right)\]
    and under the conditions on the sparse LSM model with sparsity parameter $\omega_N=\omega(N^{-1/2})$ in \cite{li2023statistical}, we have
    \[\frac{1}{N}\|\hat{\bm\Phi}_N-\bm\Phi_N\|_F^2={\c O}_{\t\bbP}\left(\frac{1}{N\omega_N}\sum_{k\le L_N}\zeta_1(k)^2\right).\]
\end{Lemma}
Clearly, Lemma \ref{lmm:sieve:est} provides a result on the concentration of the sieve design matrix when the true unobserved latent effects $\mb{U}$ are replaced with their estimated counterparts $\hat{\mb{U}}$. Both the RDPG and the LSM models can accommodate sparse networks through the sparsity parameters $\rho_N$ and $\omega_N$, respectively. However, our result with the RDPG model has a better concentration of the $\hat{\bm\Phi}_N$ to $\bm\Phi_N$, for whenever $\omega_N$ in LSM is smaller than $\frac{1}{\log^{2c}N}$, i.e., if the LSM model is even slightly sparse. Therefore the main difference between these two models for our results is that estimated sieve design matrix with estimated latent effects concentrate better to the true sieve design matrix with true latent effects better for RDPG model. However, we emphasize that our results for both models accommodate sparse networks with expected density requirement for RDPG being $\frac{\log^4N}{N}$ and LSM being $\frac{N^{1/2}}{N}$ (in addition to what is required to satisfy conditions \eqref{eqn:clt-c1} and \eqref{eqn:clt-c2} in the main theorem).

\subsection{Theory for PPPI}
\label{sec:ppi-theory}

All conditions for identification, latent-position estimation, sieve
approximation, sparse-network concentration and asymptotic normality of the
oracle human-reference estimator are assumed as in Theorem~\ref{thm:clt}. We
state only additional assumptions needed and results established for the PPPI estimator. Define the conditionally centered
oracle quantities $\widetilde{\mb Y}^{\Human},\widetilde{\mb Z}^{\Human},\widetilde{\mb K},
\widetilde{\mb E}$ and the probability limits
$\Sigma_{\widetilde Z\widetilde K}$, $\Sigma_{\widetilde K}$,
$\mathcal Q_0=\Sigma_{\widetilde Z\widetilde K}\Sigma_{\widetilde K}^{-1}
\Sigma_{\widetilde K\widetilde Z}$ as in the proof of
Theorem~\ref{thm:clt}. Then by the theoretical arguments as in Theorem~\ref{thm:clt}, we have, 
$\sqrt N\{\vecop(\widehat{\bm\beta}_{\Human})-\vecop(\bm\beta_0)\}
\tod\mathcal N(0,\mb{V}\otimes\mathcal Q_0^{-1})$ for the
infeasible estimator based on $\mb Y^{\Human}$ for all residents. Let
$\mb M_{\mb\Phi_N}$ be the oracle sieve projection and set
\begin{equation*}
\Lambda_{Ni}^*=\mb M_{\mb\Phi_N}\Lambda_i(\mb D_0),
\quad
\mb K_N^*=\mb M_{\mb\Phi_N}\mb K,
\quad
\xi_{Ni}=\vecop\{(\mb K_N^*)^\top\Lambda_{Ni}^*\},
\end{equation*}
with $\mb S_{\xi,N}$ the finite-population covariance of $\{\xi_{Ni}\}_{i\le N}$. We further set
$\Gamma_0=I_m\otimes\left(\mathcal Q_0^{-1}
\Sigma_{\widetilde Z\widetilde K}\Sigma_{\widetilde K}^{-1}\right)$. 

Let ${\c W}_N = \{\mb{A}, \mb{X}, \mb{U}, \mb{Y}^H, \mb{Y}^A, \boldsymbol{\Phi}_N\}$ denote the complete finite population data, including the network structure, covariates, latent positions, and sieve quantities. We define $\mathcal{F}_N$ as the $\sigma$-field generated by ${\c W}_N$, excluding the resident sampling indicators $\mb{R} = (R_1, \ldots, R_N)^\top$, i.e.,
   $ \mathcal{F}_N := \sigma({\c W}_N).$
The resident-sampling moments evaluated below are conditional on $\mathcal{F}_N$ (e.g., $\mathbb{E}[\cdot \mid \mathcal{F}_N]$ and $\mathbb{V}[\cdot \mid \mathcal{F}_N]$). Since $\xi_{Ni}$ is $\mathcal{F}_N$-measurable, conditional on $\mathcal{F}_N$, the array $\{\xi_{Ni}\}_{i \le N}$ constitutes a \emph{fixed} finite population and the sampling vector $\mb{R}$ is the sole source of randomness.

\begin{Assumption}
\label{ass:resident-design}
For each $N$, exactly $n=n_N$ residents are sampled uniformly without
replacement, independently of the validation outcomes conditional on the
complete data, with $f_N=n_N/N\to f\in(0,1)$.
\end{Assumption}

\begin{Assumption}
\label{ass:validation-regularity}
$\norm{\mb K}_{2,\infty}=O_p(1)$ and $\delta_{\max}/\delta_{\min}=O_p(1)$, and $\mb S_{\xi,N}\toP\Omega_\xi$ with $\Omega_\xi$ nonnegative definite.
\end{Assumption}

The $\varepsilon$-truncation in \eqref{eq:YA-YH} gives
$\norm{\Delta_i}_\infty\le2\log(1/\varepsilon)$ uniformly and
non-stochastically, and $\norm{\mb G\mb e_i}_1\le d_i/\delta_{\min}$, where $d_i$ is the degree of the $i$ th vertex. Together with
$\norm{\mb K}_{2,\infty}=O(1)$ these bound $\xi_{Ni}$ uniformly in $i$ and $N$.
Assumption~\ref{ass:validation-regularity} therefore also provides the
H\'ajek negligibility condition
$$\frac{\max_{i}\{\mb a^\top(\xi_{Ni}-\bar\xi_N)\}^2}
{\sum_i\{\mb a^\top(\xi_{Ni}-\bar\xi_N)\}^2}\toP0$$ required by the
finite-population CLT, without an additional assumption. We are now ready to state our main results for the PPPI estimator.

\begin{Theorem}[Consistency and asymptotic normality of PPI peer-effect estimation]
\label{thm:ppi-asymptotic}
Suppose Assumptions~\ref{ass:resident-design}--\ref{ass:validation-regularity}
hold together with the assumptions of Theorem~\ref{thm:clt}. Then
$\widehat{\bm \beta}_{\PP}\toP\bm\beta_0$ and, with
$\Omega_R=\frac{1-f}{f}\Omega_\xi$,
\begin{equation}
\sqrt N\left\{\vecop(\widehat{\bm \beta}_{\PP})-\vecop(\bm\beta_0)\right\}
\tod\mathcal N\!\left(0,\Sigma_{\beta,\PP}\right),
\qquad
\Sigma_{\beta,\PP}
=\mb{V}\otimes\mathcal Q_0^{-1}
+\Gamma_0\Omega_R\Gamma_0^\top.
\label{eq:ppi-clt}
\end{equation}
\end{Theorem}

It turns out that the two components in PPPI estimator are asymptotically independent because the resident sampling related term has
mean zero conditional on $\mathcal F_N$ while the oracle (human-reference) IV moment is
$\mathcal F_N$-measurable. The second term is nonnegative definite: validating only
a subset of residents weakly increases the asymptotic
variance relative to the oracle. The inflation factor $(1-f)/f$ is large
at the sampling fractions used here ($f\approx0.10$ in the female unit and
$0.06$ in the male unit). This is why the PPPI standard errors
exceed their AI-only counterparts.

\subsubsection{Standard errors}
\label{sec:ppi-se}

The two components of the covariance matrix in \eqref{eq:ppi-clt} must be estimated separately. We cannot use the
the raw residual covariance computed from the PPI pseudo-outcomes since that already
contains resident-sampling noise. Let $\Delta_i=\mb Y_i^{\Human}-\mb Y_i^{\AI}$
for validated residents and
\begin{equation}
\widehat\Lambda_i
=\mb e_i\Delta_i^\top-\mb G\mb e_i\Delta_i^\top\widehat{\mb D}_{\PP},
\qquad
\widehat\Lambda_i^*=\mb M_{\widehat{\mb \Phi}_N}\widehat\Lambda_i,
\qquad
\widehat{\mb b}_i
=\vecop\!\left\{\mathcal A_N^{-1}(\widehat{\mb Z}^*)^\top \mb P_{\widehat{\mb K}^*}
\widehat\Lambda_i^*\right\}.
\label{eq:bihat}
\end{equation}
Then $\widehat{\mb b}_i$ is the coefficient-level validation influence of a single resident $i$, which is obtained by determining the effect of the perturbation $\widehat\Lambda_i$ in the 2SLS estimator. Let $\widehat{\mb S}_b$ be the sample covariance of $\{\widehat{\mb b}_i\}_{i:R_i=1}$.
Then the plug-in estimate for the finite-population variance of
the Horvitz--Thompson correction is
\begin{equation}
\ \widehat{\mb V}_{\beta,R}=\frac{N^2(1-f_N)}{n}\widehat{\mb S}_b,
\qquad\text{so that}\qquad
N\widehat{\mb V}_{\beta,R}\toP\Gamma_0\Omega_R\Gamma_0^\top .
\label{eq:VbetaR}
\end{equation}
For the model-based component, let $\widehat{\mb E}_{\PP}^*$ be the fitted PPI
residual matrix after sieve residualization. Then define,
$\widehat{\mb V}_{E,\mathrm{raw}}=N^{-1}(\widehat{\mb E}_{\PP}^*)^\top
\widehat{\mb E}_{\PP}^*$, and let $\widehat{\mb S}_\Lambda$ be the sample covariance of
$\{\widehat\Lambda_i^*\}_{i:R_i=1}$ in the sense of \eqref{eq:SLambda} in the
appendix. Because $\widehat{\mb E}_{\PP}^*$ contains the resident-validation
perturbation, $\widehat{\mb V}_{E,\mathrm{raw}}$ overstates $\widehat{\mb V}$
by the SRSWOR design contribution, and we estimate the oracle residual
covariance by
\begin{equation}
\widehat{\mb V}
=\widehat{\mb V}_{E,\mathrm{raw}}-\frac{1-f_N}{f_N}\widehat{\mb S}_\Lambda.
\label{eq:VEdebiased}
\end{equation}
Using $\widehat{\mb V}_{E,\mathrm{raw}}$ directly and then adding
\eqref{eq:VbetaR} would instead count validation variation twice. Since
$\mb P_{\widehat{\mb K}_v^*}=\mb I_m\otimes \mb P_{\widehat{\mb K}^*}$, the covariance function of
Algorithm~\ref{algorithm1} reduces to a Kronecker product and the total PPI
covariance is
\begin{equation}
\ \widehat{\mb V}_{\beta,\PP}
=\widehat{\mb{V}}\otimes\mathcal A_N^{-1}
+\widehat{\mb V}_{\beta,R}.
\label{eq:Vtotal}
\end{equation}
Reported standard errors are square roots of its diagonal elements.
Proposition~\ref{prop:variance-consistency} in the appendix shows
$N\widehat{\mb V}_{\beta,\PP}\toP\Sigma_{\beta,\PP}$.

\subsubsection{Specification choices}
\label{sec:robustness}

For our main results, we report a variant of the method described above that replaces the plug-in ALR in
\eqref{eq:YA-YH} by a first-order Taylor-expansion  around $\mb q_i^{\AI}$. This linearized ALR is
affine in $\widehat{\mb q}_i^{\Human}$ and therefore inherits its conditional
unbiasedness exactly, at the cost of targeting a first-order approximation to
$\mb T(\mb q_i^{\Human})$. The second change we make is to power-tune the BART contribution at the message
stage, at the resident stage, or both, in the spirit of PPI++ \citep{angelopoulos2023ppi++}. Ordinary PPI removes the bias introduced by LLM measurements, but does not guarantee improvement in efficiency relative to classical inference. In particular, if the LLM measures are weakly correlated with human preferences, then using LLM measures may even inflate the variance of the estimator. The power-tuned PPPI guarantees asymptotic efficiency compared to classical peer effect inference with available human data alone. A single
message-stage tuning parameter is used across categories, which preserves the sum-to-one constraint of the calibrated
profile. Both specifications are defined in
Appendix~\ref{app:linearized}--\ref{app:ppipp}; in each case the peer lag is
recomputed from the resulting outcome and the estimator and covariance
correction of Sections~\ref{sec:resident-ppi}--\ref{sec:ppi-se} are applied
unchanged.

\section{Simulations}
In this section, we describe a simulation study we performed to evaluate the finite-sample performance of the proposed peer effect inference strategy in the presence of latent homophily and network endogeneity. Specifically, we consider two main scenarios for network and outcome generation: (i) a Random Dot Product Graph (RDPG) model and (ii) a latent space model with covariates. The simulated data are then used to compare several alternative adjustment strategies for latent confounding within the 2SLS estimation framework. We evaluate both nonparametric sieve-adjustment methods, which use tensor product polynomial basis expansions of degree 2, as well as linear functions of the estimated latent positions. We report the bias, variance, mean squared error, and 95\% confidence interval coverage rates calculated across simulation replications.

\begin{table}[!h]
\centering
\caption{Monte Carlo Simulations RDPG}
\label{simsfinal}
\small
\renewcommand{\arraystretch}{0.85} 
\setlength{\tabcolsep}{8pt}        
\begin{tabular}{p{2.0cm}p{2.0cm}p{2.0cm}p{2.0cm}p{2.0cm}p{2.0cm}}
  \hline\hline
Parameter & N & Bias & Variance & MSE & Coverage \\ \hline
  \multicolumn{6}{c}{\textbf{a. Linear}}\\
$D_{11}$ &   100 & 0.1457 & 0.2625 & 0.2824 & 0.9300 \\ 
 $D_{11}$ &   500 & 0.1067 & 0.3150 & 0.3248 & 0.9500 \\ 
  $D_{11}$ &  1000 & 0.0542 & 0.1934 & 0.1954 & 0.9500 \\ 

  $D_{12}$ &  100 & -0.0250 & 0.0437 & 0.0441 & 0.9550 \\ 
  $D_{12}$ &  500 & -0.0263 & 0.0518 & 0.0523 & 0.9600 \\ 
 $D_{12}$ &  1000 & -0.0108 & 0.0292 & 0.0292 & 0.9550 \\ 
    $D_{21}$ &  100 & 0.2343 & 0.4031 & 0.4560 & 0.9350 \\ 
 $D_{21}$ & 500 & 0.1036 & 0.3296 & 0.3387 & 0.9300 \\ 
  $D_{21}$ &  1000 & 0.0982 & 0.2576 & 0.2660 & 0.9400 \\ 

    $D_{22}$ &  100 & 0.0137 & 0.0676 & 0.0674 & 0.9150 \\ 
  $D_{22}$ &  500 & 0.0091 & 0.0552 & 0.0551 & 0.9500 \\ 
  $D_{22}$ &  1000 & -0.0043 & 0.0401 & 0.0399 & 0.9500 \\ 
  \multicolumn{6}{c}{\textbf{b. Tensor polynomial degree 2}}\\

    $D_{11}$ &  100 & 0.3325 & 0.5534 & 0.6612 & 0.9200 \\ 
  $D_{11}$ &  500 & 0.2052 & 0.3830 & 0.4232 & 0.9300 \\ 
  $D_{11}$ &  1000 & 0.1030 & 0.2328 & 0.2422 & 0.9750 \\ 
  $D_{12}$ &  100 & -0.1145 & 0.1439 & 0.1563 & 0.9600 \\ 
 $D_{12}$ &  500 & -0.0840 & 0.0989 & 0.1055 & 0.9400 \\ 
  $D_{12}$ &  1000 & -0.0407 & 0.0553 & 0.0567 & 0.9700 \\ 

 $D_{21}$ &  100 & 0.2219 & 0.5855 & 0.6318 & 0.9600 \\ 
   $D_{21}$ &  500 & 0.0491 & 0.3248 & 0.3256 & 0.9650 \\ 
  $D_{21}$ &  1000 & 0.0469 & 0.2828 & 0.2836 & 0.9600 \\ 

   $D_{22}$ &  100 & -0.0154 & 0.1552 & 0.1547 & 0.9600 \\ 
  $D_{22}$ &  500 & 0.0127 & 0.0839 & 0.0836 & 0.9700 \\ 
  $D_{22}$ &  1000 & 0.0017 & 0.0695 & 0.0691 & 0.9600 \\ 
\hline\hline
\end{tabular}

\begin{minipage}{0.90\textwidth}
\footnotesize
Notes: This table reports Monte Carlo bias, variance, mean squared error (MSE), and 95\% confidence interval coverage for the RDPG network model under the linear and tensor-polynomial specifications. Results are based on 200 simulations for network sizes $N=100,500,1000$.
\end{minipage}
\end{table}

Table \ref{simsfinal} documents the performance of the estimator when the network is generated from the RDPG model and the spectral embedding is used to estimate the latent variables. The dimension of latent variables in the outcome model is fixed at 2. There are two correlated outcomes where the data generation follows Equation \ref{MSARmodel} with endogeneity in network formation. The peer influence parameter matrix is set at $\begin{bmatrix}
0.3 &0.25  \\
0.2 & 0.6
\end{bmatrix}$, where the diagonal elements capture direct peer-influence spillovers and the off-diagonal elements capture indirect spillovers in the outcome. The dimension of the coefficients on $\mb h^Y(\mb U)$ varies across the three scenarios. For each scenario, we generate the data 200 times and estimate the parameters using the steps in Algorithm \ref{algorithm1}. The table shows that the bias, variance, and mean squared error steadily decline as the number of nodes in the network increases. This supports the fact that the estimator recovers the true parameter of interest well as the sample size increases. Moreover, the 95\% confidence interval coverage rates remain near nominal across all scenarios. This observation provides an empirical finite-sample validation of our standard error formulae.

Next, we change the data generation and latent-variable model to use latent space models that allow both multiplicative and additive unobserved variables, along with observed covariates, to generate the network. A Projected Gradient Descent (PGD) algorithm with Universal Singular Value Thresholding (USVT) initialization is used to estimate the MLE $\hat{U}$ \citep{ma2020universal}. Note that even though this solves a non-convex optimization problem, the solution has nice theoretical guarantees which we have used in the proofs above. The covariate that explains the network formation is not part of the data generation for the outcomes. The network formation consists of one additive latent factor, two multiplicative latent factors, and two observed covariates. The outcome data generation consists of $[\mb{AY},\mb X,\mb{AX},\mb h^Y(\mb U)]$ where $\mb{AY}$ is $N\times 2$ matrix, $\mb X$ is $N\times 5$, $\mb{AX}$ is $N\times 5$ and dimension of $\mb h^Y(\mb U)$ varies across the 3 scenarios. The peer influence matrix $D$ used in the data generation for the latent space model with covariates is fixed at $\begin{bmatrix}
0.8 &0.3  \\
0.2 & 0.6
\end{bmatrix}$. Table \ref{simsfinal_lsm} reports the bias, variance, MSE and 95\% coverage rates across 200 simulations for each scenario as the sample size increases from 100 to 300 to 500. The results show that bias, variance, and MSE consistently decline as sample size increases, while the coverage rates are near nominal levels for all four parameters of the peer influence $\mb D$ matrix across all three scenarios.

\begin{table}[!h]
\centering
\caption{Monte Carlo Simulations for Latent Space Model with Covariates}
\label{simsfinal_lsm}
\small
\renewcommand{\arraystretch}{0.85} 
\setlength{\tabcolsep}{8pt}   
\begin{tabular}{p{2.0cm}p{2.0cm}p{2.0cm}p{2.0cm}p{2.0cm}p{2.0cm}}
\hline\hline
Parameter & N & Bias & Variance & MSE & Coverage \\
\hline
\multicolumn{6}{c}{\textbf{a. Linear}}\\
$D_{11}$ & 100 & 0.0156 & 0.0420 & 0.0421 & 0.9000 \\
$D_{11}$ & 300 & 0.0000 & 0.0164 & 0.0163 & 0.9400 \\
$D_{11}$ & 500 & -0.0033 & 0.0153 & 0.0152 & 0.9350 \\
$D_{12}$ & 100 & -0.0059 & 0.0624 & 0.0622 & 0.9050 \\
$D_{12}$ & 300 & 0.0008 & 0.0222 & 0.0220 & 0.9400 \\
$D_{12}$& 500 & -0.0020 & 0.0153 & 0.0152 & 0.9400 \\
$D_{21}$ & 100 & 0.0288 & 0.0508 & 0.0514 & 0.9100 \\
$D_{21}$ & 300 & 0.0048 & 0.0164 & 0.0164 & 0.9500 \\
$D_{21}$ & 500 & -0.0039 & 0.0141 & 0.0141 & 0.9650 \\
$D_{22}$ & 100 & -0.0348 & 0.0637 & 0.0646 & 0.9100 \\
$D_{22}$ & 300 & -0.0111 & 0.0235 & 0.0235 & 0.9550 \\
$D_{22}$ & 500 & 0.0049 & 0.0149 & 0.0149 & 0.9700 \\
\multicolumn{6}{c}{\textbf{b. Tensor Polynomial Degree 2}}\\
$D_{11}$ & 100 & 0.0341 & 0.0578 & 0.0587 & 0.9250 \\
$D_{11}$ & 300 & 0.0008 & 0.0269 & 0.0267 & 0.9350 \\
$D_{11}$ & 500 & -0.0073 & 0.0199 & 0.0198 & 0.9250 \\
$D_{12}$ & 100 & -0.0117 & 0.0699 & 0.0696 & 0.9050 \\
$D_{12}$ & 300 & 0.0002 & 0.0293 & 0.0292 & 0.9500 \\
$D_{12}$ & 500 & -0.0031 & 0.0169 & 0.0169 & 0.9300 \\
$D_{21}$ & 100 & 0.0310 & 0.0728 & 0.0734 & 0.9200 \\
$D_{21}$ & 300 & -0.0009 & 0.0276 & 0.0274 & 0.9550 \\
$D_{21}$ & 500 & 0.0005 & 0.0141 & 0.0140 & 0.9750 \\
$D_{22}$ & 100 & -0.0254 & 0.0686 & 0.0689 & 0.9250 \\
$D_{22}$ & 300 & -0.0175 & 0.0287 & 0.0289 & 0.9500 \\
$D_{22}$ & 500 & -0.0018 & 0.0167 & 0.0166 & 0.9500 \\
\hline\hline
\end{tabular}
\begin{minipage}{15cm}\footnotesize
Notes: This table reports Monte Carlo bias, variance, mean squared error (MSE), and 95\% confidence interval coverage for the latent space model with covariates under the linear and tensor-polynomial specifications. Results are based on 200 simulations for network sizes $N=100,300,500$.
\end{minipage}
\end{table}

\section{Empirical Application}
\label{results}

\subsection{Predict Recidivism using Text Messages}

 We first investigate the ability of text exchanges to predict recidivism. As mentioned in Section \ref{BERTandBART}, we use two approaches 
 for recidivism prediction, which are using BERT embeddings in a high-dimensional penalized regression (LASSO) model, and fine-tuning using LoRA and adapter methods.

LASSO is implemented with 5-fold cross-fitting on the held-out data to do the entire prediction exercise out-of-sample.\footnote{Cross-fitting divides the sample into K-folds, and the prediction for the residents in the $k^{th}$ fold uses the model trained on the data from the remaining $ K-1 $ folds.} 
We choose the penalty parameter separately for each of the five folds using cross-validation on that fold's training data, without using the held-out test data for that fold.

Figure \ref{predictembeddings} provides the Receiver Operator Characteristic (ROC) curves and the Area Under the Curve (AUC) for two models, one with both pre-entry covariates ($X$) and embeddings ($T$), and an alternative model without the text embeddings (only $X$) as predictors. Panel (a) shows that the prediction AUC is 23\% higher for recidivism in the female unit when we incorporate the LLM-based text embeddings relative to the model with only pre-entry covariates $X$. The dimension reduction is substantial for the former model, as only about 11-29 predictors from the text embeddings (768) and pre-entry covariates remain across the five folds in the LASSO model fit. Panel (b) compares the prediction accuracy for the male units with and without the embeddings, where the embedding profiles are aggregated by sender. The improvement in prediction accuracy for males is 29\% with the embeddings. Panels (c) and (d) repeat the same analysis, but for the case where the text embeddings are aggregated by receiver ID, and find slightly higher improvements (28.6\% and 31\% respectively for female and male units).\footnote{There is a very small difference in the AUC for the female sender and female receiver without the embeddings. This slight difference can be attributed to the lambda selection inside each fold for this prediction analysis.} 

\begin{figure}[!h] 
  \caption{Predict Recidivism with and without text embeddings}
    \label{predictembeddings} 

 \begin{minipage}[b]{0.5\linewidth}
       \begin{center}

    \includegraphics[width=\linewidth]{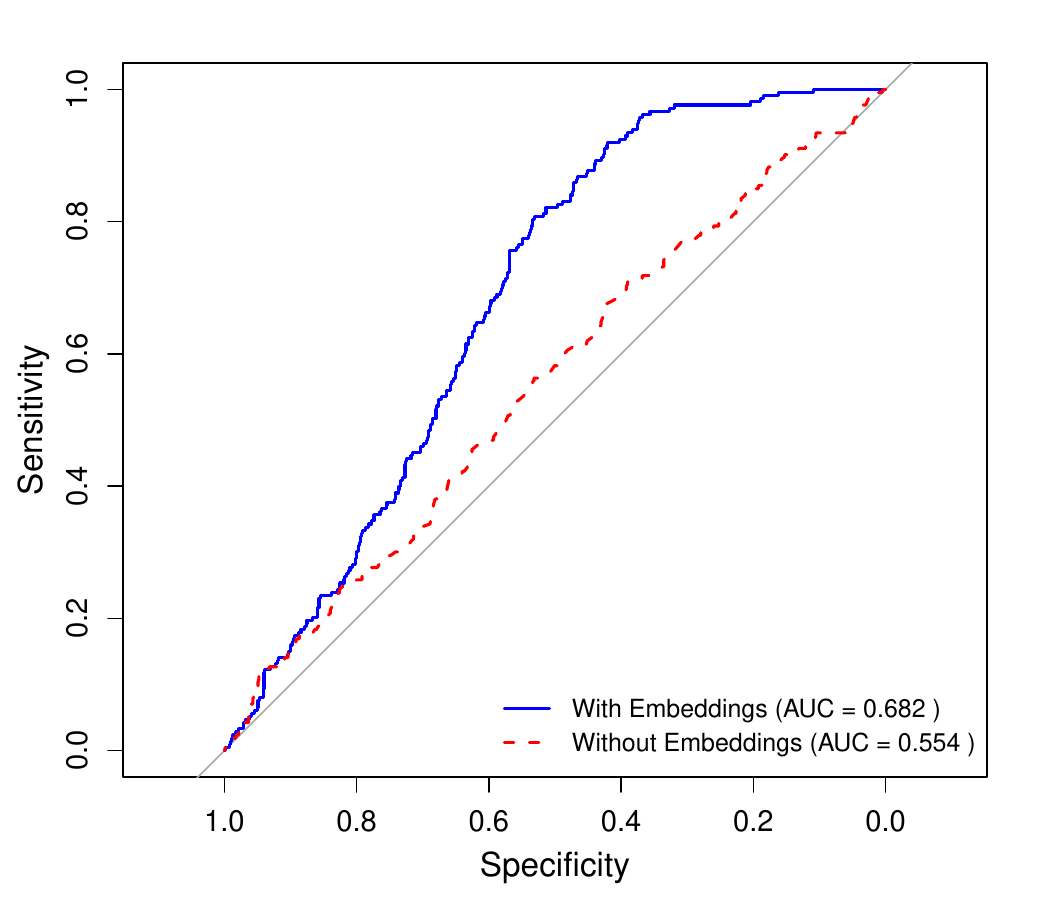}
    \caption*{a. Female, sender} 
        \end{center}

  \end{minipage}
  \begin{minipage}[b]{0.5\linewidth}
       \begin{center}

    \includegraphics[width=\linewidth]{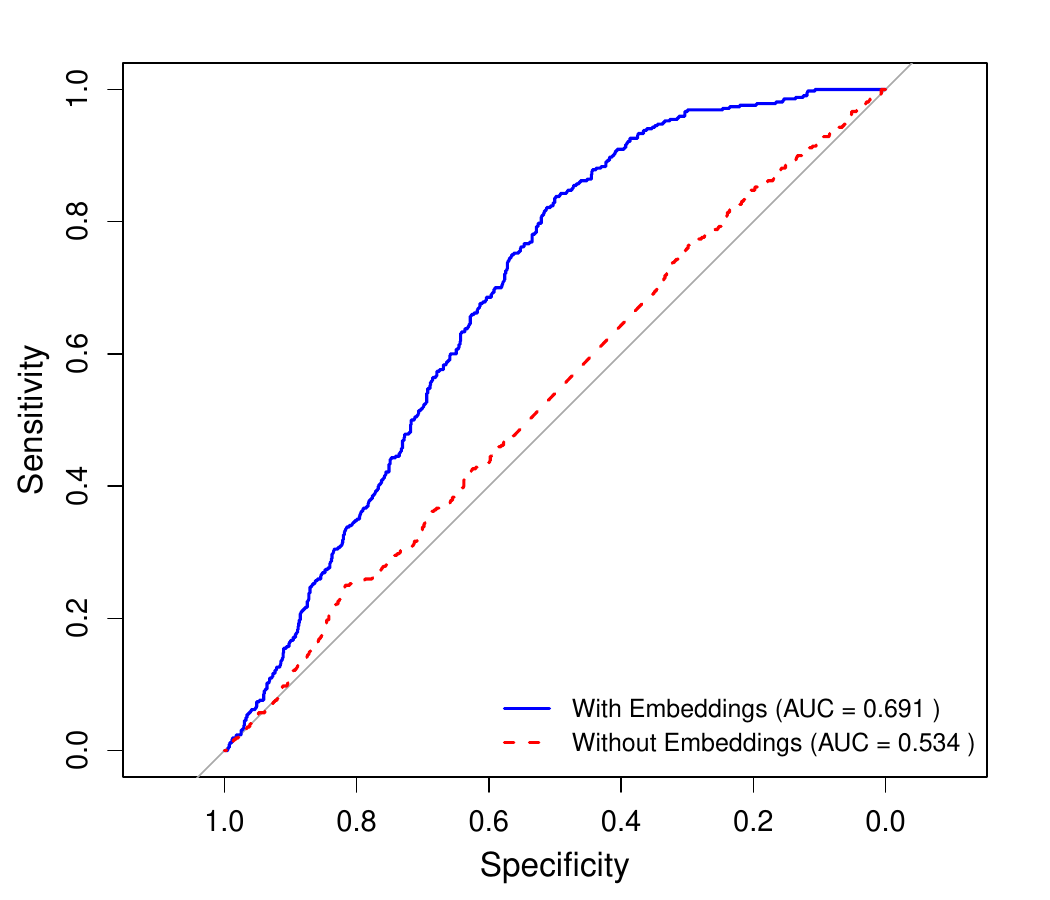}
    \caption*{b. Male, sender} 
        \end{center}

  \end{minipage}

   \begin{minipage}[b]{0.5\linewidth}
       \begin{center}

    \includegraphics[width=\linewidth]{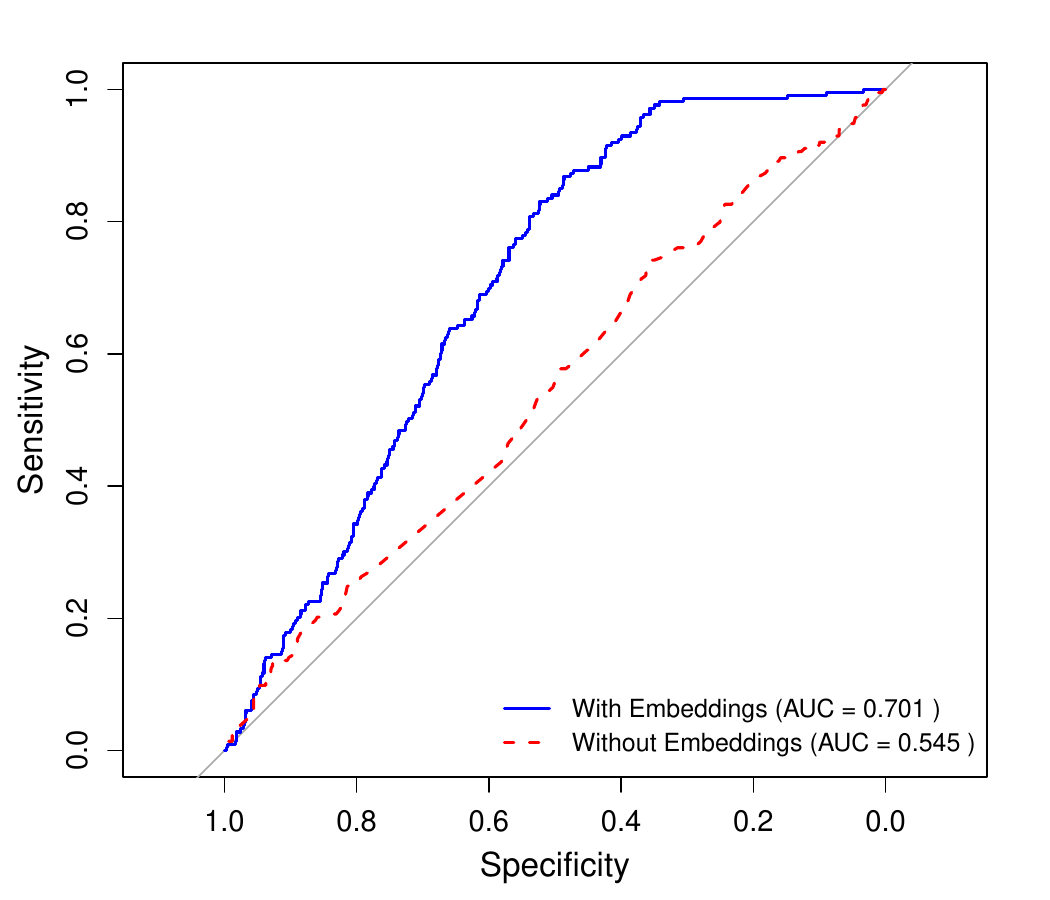}
    \caption*{c. Female, receiver} 
        \end{center}

  \end{minipage}
  \begin{minipage}[b]{0.5\linewidth}
       \begin{center}

    \includegraphics[width=\linewidth]{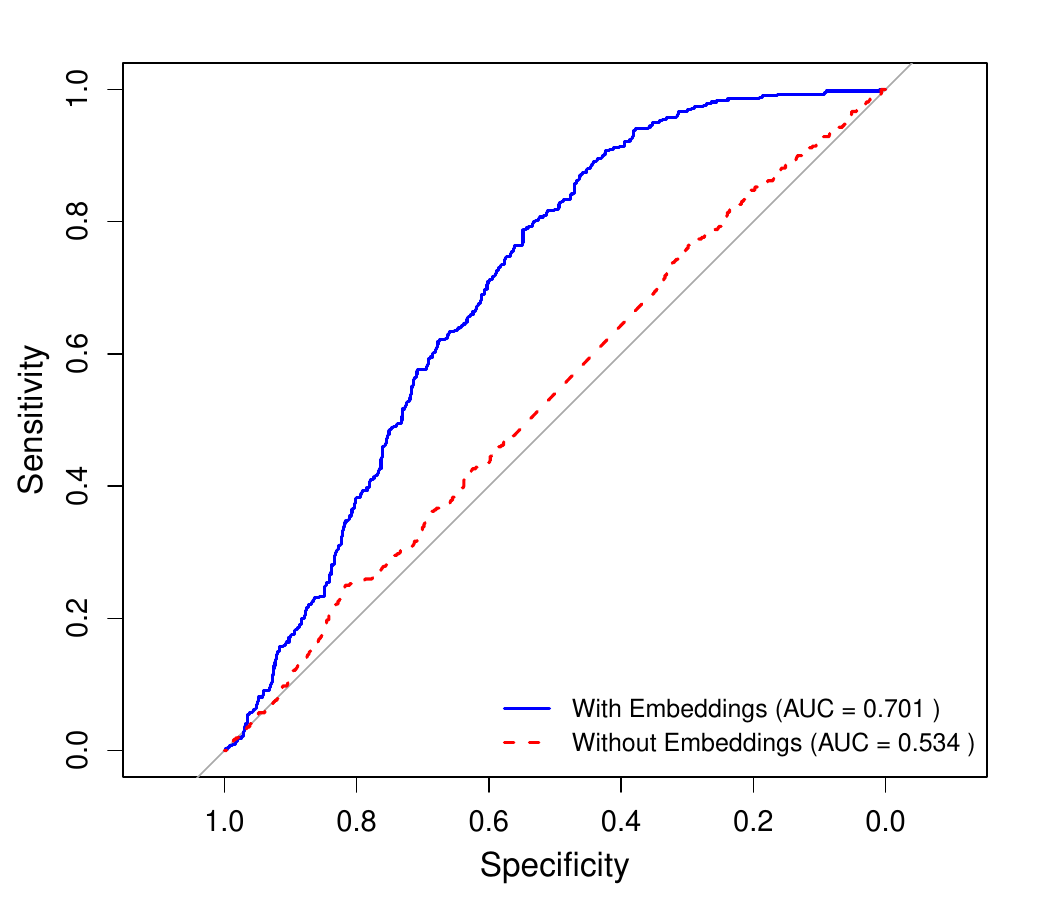}
    \caption*{d. Male, receiver} 
        \end{center}

  \end{minipage}
  \begin{minipage}{15.0 cm}{\footnotesize{Notes: Five-fold cross-fitting is used to estimate the out-of-sample AUC for all models. Logistic LASSO is used for the prediction models because the text embeddings are high-dimensional. Panels (a) and (b) aggregate the text embeddings by sender ID, and (c) and (d) aggregate by receiver ID.}}
\end{minipage} 
\end{figure}

\begin{figure}[!h]
    \centering
    \caption{Pooled out-of-fold test AUC by epoch for LoRA and adapter based parameter efficient fine tuning}
    \label{fig:sender_receiver_auc}
    \begin{subfigure}[t]{0.49\textwidth}
        \centering
        \includegraphics[width=\linewidth]{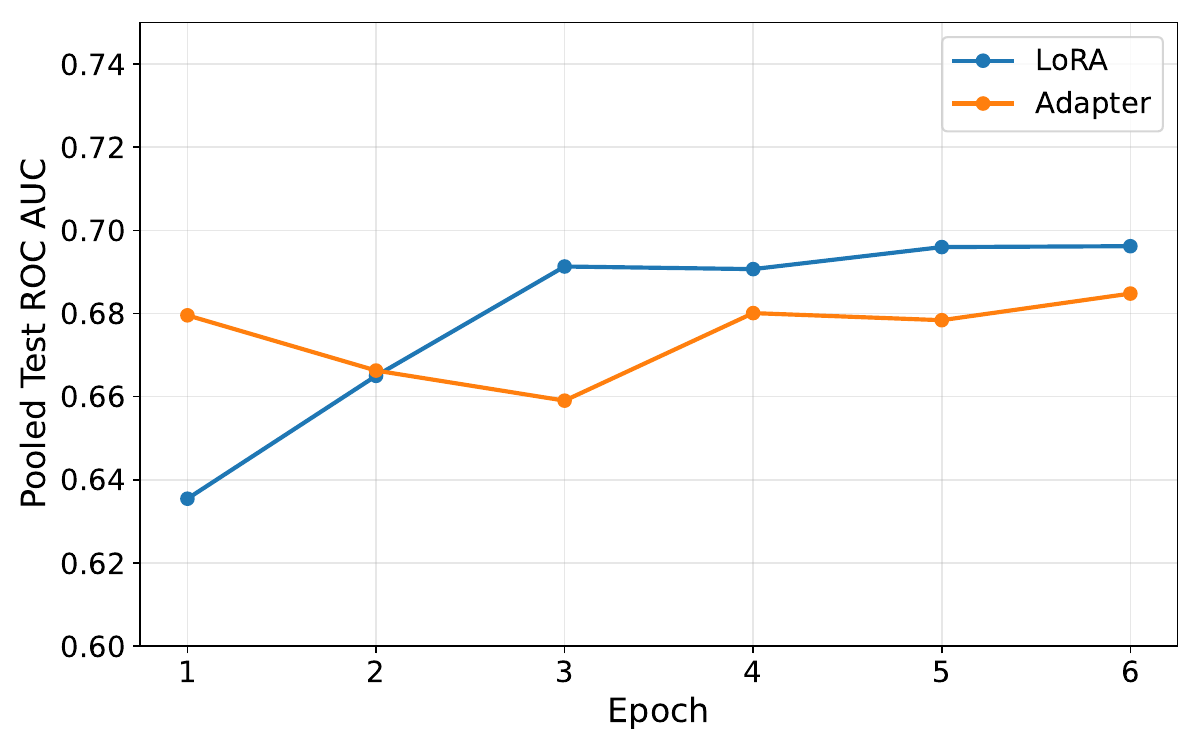}
        \caption{Male Sender}
        \label{fig:sender_auc}
    \end{subfigure}
    \hfill
    \begin{subfigure}[t]{0.49\textwidth}
        \centering
        \includegraphics[width=\linewidth]{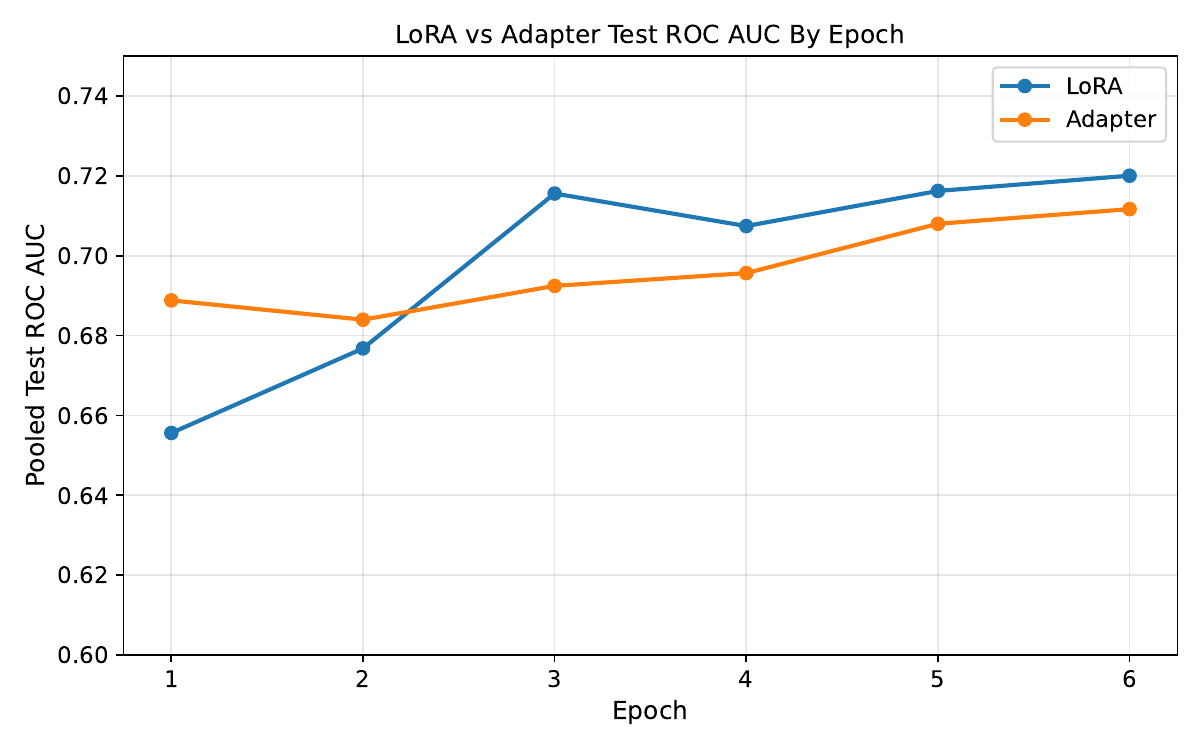}
        \caption{Male Receiver}
        \label{fig:receiver_auc}
    \end{subfigure}
       \begin{subfigure}[t]{0.49\textwidth}
        \centering
        \includegraphics[width=\linewidth]{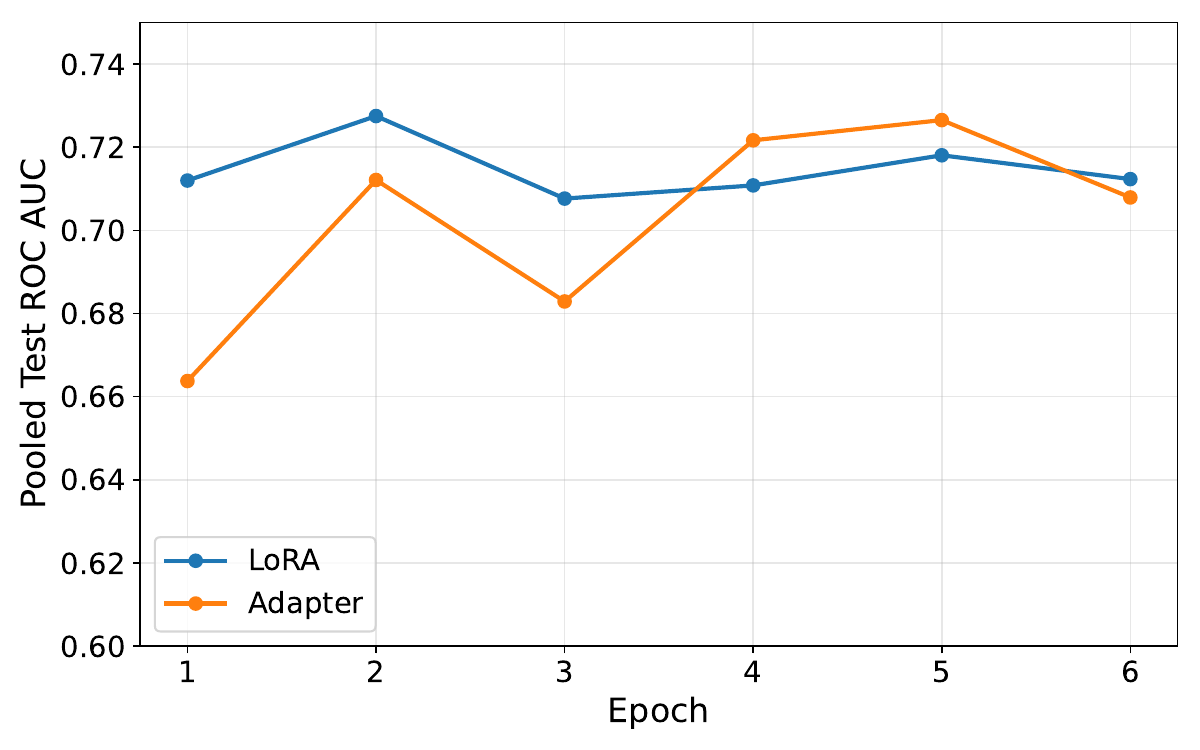}
        \caption{Female Sender}
        \label{fig:sender_auc}
    \end{subfigure}
    \hfill
    \begin{subfigure}[t]{0.49\textwidth}
        \centering
        \includegraphics[width=\linewidth]{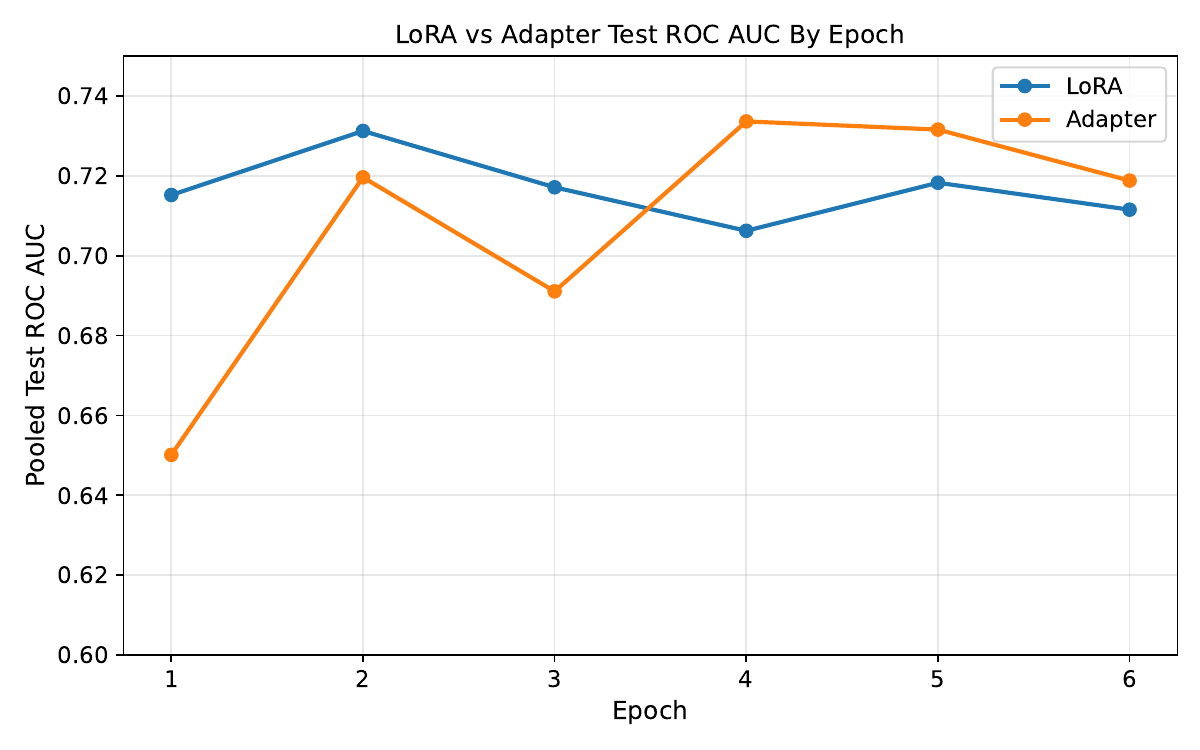}
        \caption{Female Receiver}
        \label{fig:receiver_auc}
    \end{subfigure}
\end{figure}

Next, we display the results of predicting recidivism with the fine-tuned models in Figure \ref{fig:sender_receiver_auc} for sender and receiver profiles in the male facility. We plot the out of sample pooled AUC (obtained through cross-fitting) with increasing training rounds or epochs. We see that generally the AUC values are slightly higher for both LoRA and Adapter fine-tuning than Lasso for both sender and receiver figures. Between these two fine-tuning methods, LoRA performs better in the Male facilities, while the methods are similar in the female facilities. In all cases prediction accuracy test AUC from fine-tuning either reaches 0.70 or exceeds it and reaches upto 0.72 pushing the out-of-sample AUC for the male receiver facility upto 35\% over including only covariates.

We further note that our pre-entry covariates include an aggregate measurement known as Level of Service Inventory-Revised (LSI-R), a widely used measure in the context of criminal psychology \citep{andrews2014psychology}. The LSI-R survey contains 54 questions pertaining to criminal history, education, employment, family, alcohol and drug problems, living conditions, and emotional health. Several questions that are part of the LSI-R are known to be predictors of recidivism \citep{dressel2018accuracy}. Yet, in our sample, the pre-entry covariates, including LSI-R, were only mildly predictive of recidivism. This could be because LSI-R is an aggregated measure, and we do not observe the component scores in our data. Nevertheless, this underscores the importance of AI-based text measurements as potential predictors of recidivism, especially in low-security correctional facilities, where observed covariates have limited ability to predict recidivism.

\subsection{Validation results for zero-shot output}

The first part of the pipeline in Figure \ref{fig:validation_pipeline} validates the input label names. The second part of the pipeline describes validation of the zero-shot output on held-out data, meaning data not used for label generation. We conduct several validation exercises to validate the zero-shot profiles. First, we use the zero-shot probabilities to predict corrections/affirmations at the message level. Note that every message is tagged as an affirmation or correction by the correctional facility staff. Consequently, we can use the four-label zero-shot output to predict affirmations and corrections using 5-fold cross-fitting. Figure \ref{predictaffirmationcorrection} shows that the probabilities on the final four labels can predict affirmations and corrections very well, with out-of-sample AUCs of 0.977 and 0.982 for the female unit and male unit, respectively. This validates that the zero-shot probabilities given to individual messages at least have a strong ability to differentiate between positive and negative emotions in the messages.

Next, we do a placebo analysis where we randomly draw a sample of 1000 messages from the held-out data, do the zero-shot with four arbitrary labels and separately with the final four chosen labels. The distribution of these probabilities is provided in Figure \ref{comparisonplacebo}. Panel (a) shows that placebo labels have their medians around 0.25 and there is no difference in the distribution by message type. Panel (b), on the contrary, which uses the final four labels, has a lot more variation that is directly correlated with message type, where positive labels have a much higher predicted probability for affirmations than corrections. We see the opposite for the negative labels. 

\begin{figure}[!h] 
  \caption{Boxplots of differences between BART and PPI++ profiles of residents selected for human validation relative to human-validated profiles}
    \label{ppi_human} 

 \begin{minipage}[b]{0.5\linewidth}
       \begin{center}

    \includegraphics[width=\linewidth]{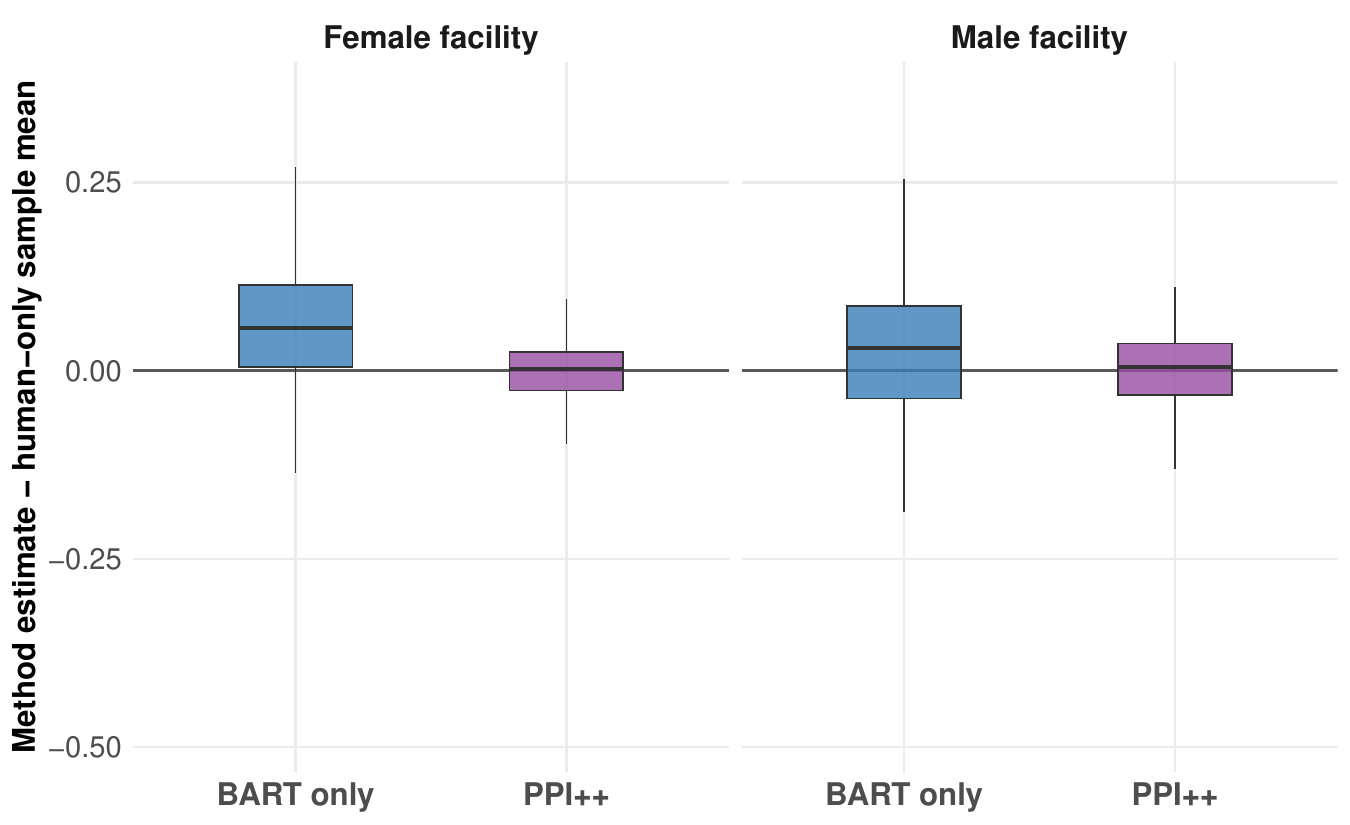}
    \caption*{a. Personal Growth} 
        \end{center}

  \end{minipage}
  \begin{minipage}[b]{0.5\linewidth}
       \begin{center}

    \includegraphics[width=\linewidth]{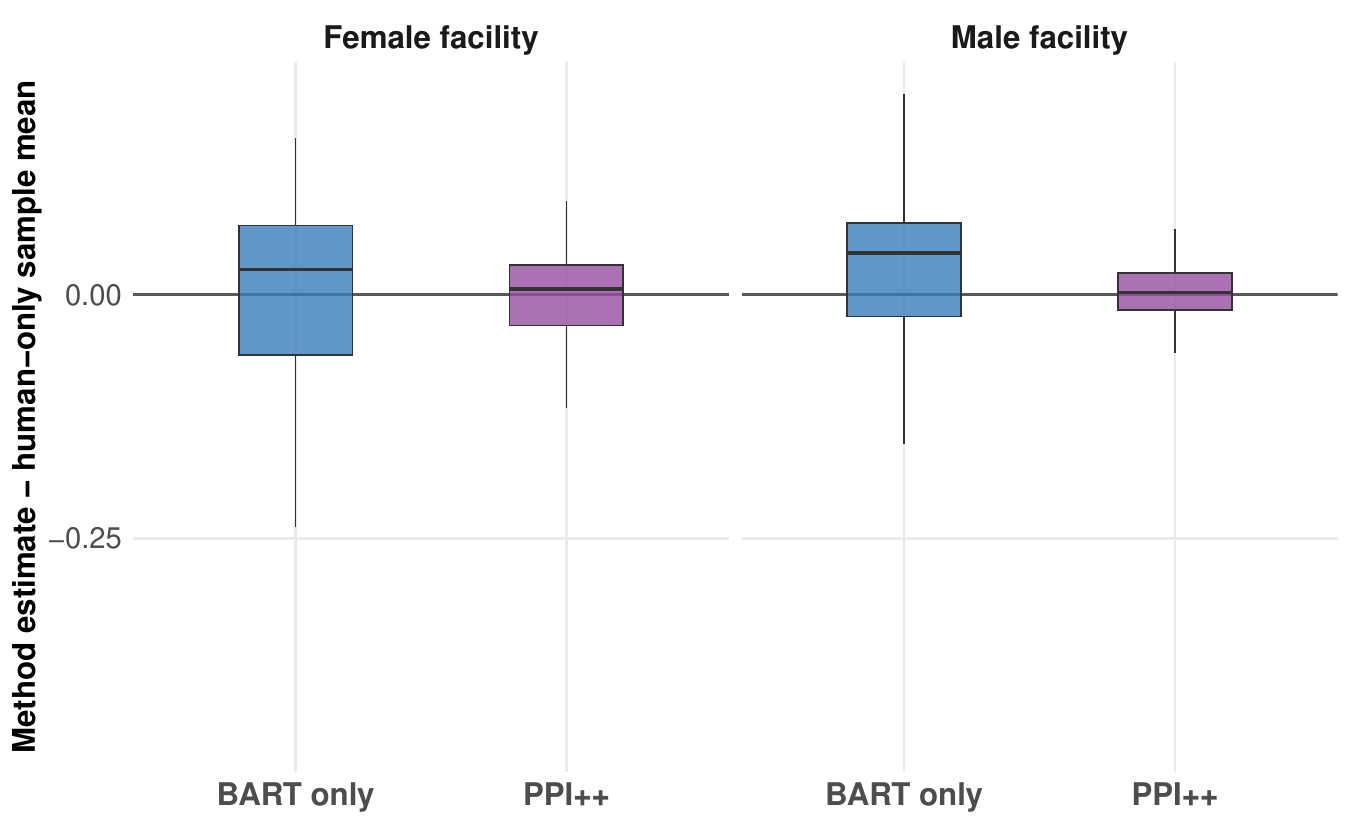}
    \caption*{b. Community Support} 
        \end{center}

  \end{minipage}

   \begin{minipage}[b]{0.5\linewidth}
       \begin{center}

    \includegraphics[width=\linewidth]{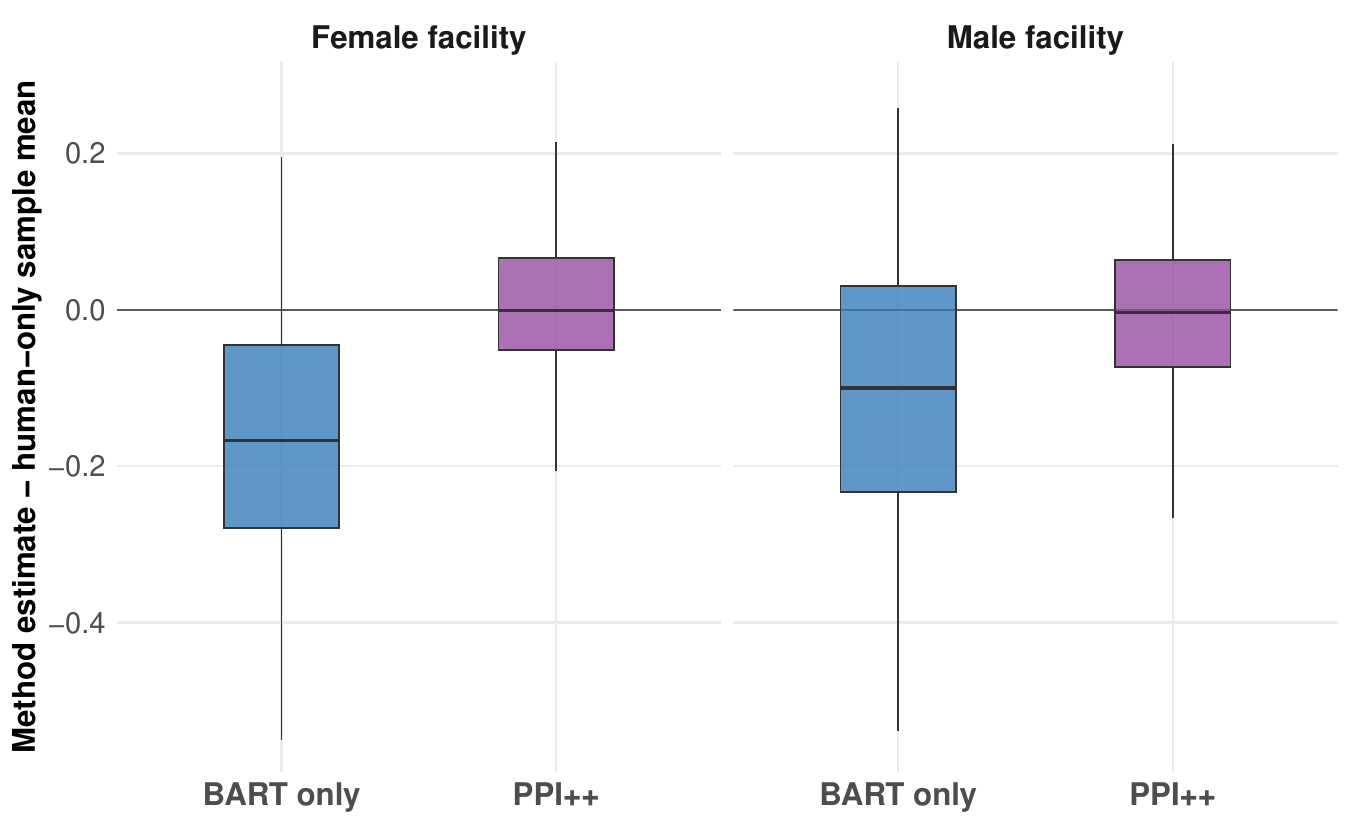}
    \caption*{c. Rule Violations} 
        \end{center}

  \end{minipage}
  \begin{minipage}[b]{0.5\linewidth}
       \begin{center}

    \includegraphics[width=\linewidth]{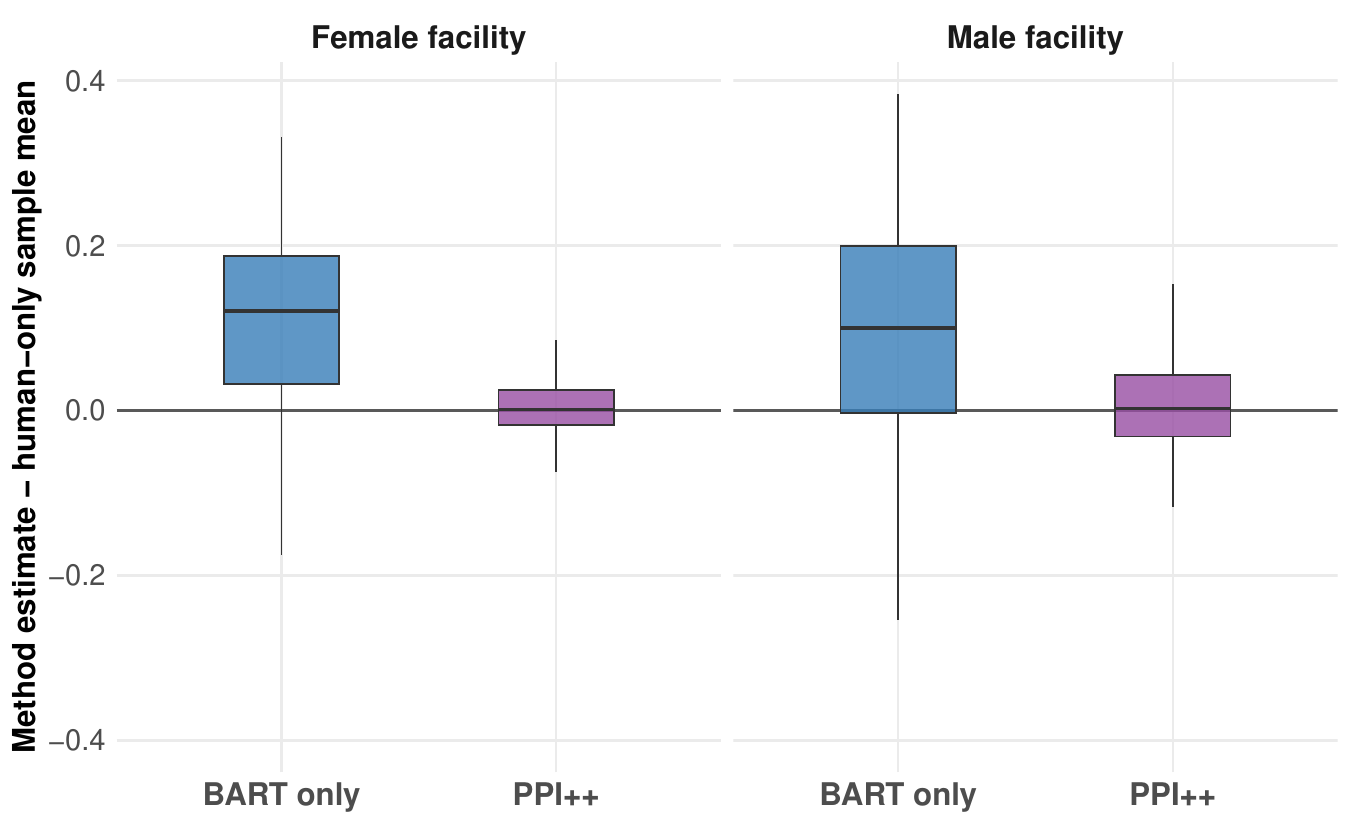}
    \caption*{d. Disruptive Conduct} 
        \end{center}

  \end{minipage}
\begin{minipage}{15.0 cm}{\footnotesize{Notes: The human only profiles are constructed by averaging multinomial classifications from Prolific raters over sampled messages for each resident. ``BART only'' uses zero-shot probabilities from BART aggregated over the all of the resident's messages. ``PPI++'' combines the human-labeled sample with BART predictions.}}
\end{minipage}
\end{figure}

Finally, we run two additional surveys on Prolific to validate the zero-shot output. We call these two surveys Human Survey 2 and Human Survey 3, respectively. Both surveys recruited 100 human participants in Prolific and gave them 50 classification task each, paying them at an hourly rate for their service. The messages for survey 2 are a subset of the held-out data, selected as follows. We select a random sample of 5 messages (by receiver profiles) for 100 random residents in the male unit and another 100 in the female unit. For each of these 1000 messages, we ask 5 human participants to rate it as a multinomial choice among the four final chosen labels. In this way, there are a total of 5000 classification tasks for the participants on Prolific. Having each message shown to 5 raters helps us estimate the empirical probability for each message by aggregating over the raters. Next by aggregating over each resident, we get a human-validated average of the behavioral profile. Below we present some comparisons between the human-validated classification and the zero-shot probabilities generated using BART. 

Using the results of the human validation survey, we can compare the behavioral profiles generated by BART and participants. We find that at the resident level, there is 70.5\% agreement between BART and humans for the top (highest probability) label. However, for our downstream analysis, the entire probability vector over the four labels is relevant. Therefore, we compare the correlations between human scores and zero-shot scores for each input label. The results are strongly aligned, with class-specific correlations of about 0.75 for Personal Growth, 0.86 for Community Support, 0.80 for Rule Violations, and 0.70 for Disruptive Conduct. Additionally, we compute the $\ell_1$ distance by adding up the absolute difference between the human validated behavioral profile and BART generated behavioral profile for each resident. The $\ell_1$ distance at the resident level is 0.42 for BART. For comparison, the $\ell_1$  distance of a profile obtained by assigning equal probability from the human consensus behavioral profile is 0.781. Thus, BART reduces full-profile prediction of human consensus by roughly 46\% relative to a uniform (equal-probability) benchmark.

In human survey 3, a similar design is followed except that the random sample of messages is constructed using sender profiles. For the sender profiles, we see a slightly better alignment between human labels and BART output. Once aggregated at the message level, we see that the match rate between humans and BART output is 71.5\% for the top label. Moreover, the correlation between human output and BART output is 0.74 for personal growth, 0.80 for community support, 0.78 for rule violations, and 0.73 for disruptive conduct. The $\ell_1$  distance for sender profiles of BART is 0.36 relative to the equal-probability assignment of 0.77.

Next, we obtain PPI++ behavioral profiles by combining limited and expensive human consensus data with the BART zero shot profile predictions  \citep{angelopoulos2023ppi++}. Figure \ref{ppi_human} compares the bias of PPI++ profiles and BART output receiver profiles relative to human-only profile  for the residents who were human-validated. Specifically, the boxplots show the distribution of differences in the BART and PPI++ profiles relative to human-consensus profiles. We can see from the figure that BART has relatively large biases in both positive and negative directions for different labels. This is not surprising given the profiles are compositional data. However, in each case, PPI++, which uses the limited human consensus to calibrate the BART profiles, shrinks the biases towards 0. Figure \ref{ppi_human_senderprofs} repeats the same comparison for the sender profiles. As with the receiver profiles, the bias similarly shrinks when we use PPI++ profiles. 

\subsection{Peer Effects in Text}
\label{peereffectlowsec}
\begin{table}[!htbp]
\centering
\caption{Peer Effects in PPPI-adjusted Class Probabilities}
\label{tab:peer_effects_ppiplus_all}
\begin{tabular}{p{6cm}p{2cm}p{2cm}p{2cm}}
\hline\hline
 & Personal Growth & Community Support & Rule Violations \\ 
\hline
\multicolumn{4}{l}{\textit{a. Female Unit: Sender Profiles}} \\
Peer Personal Growth ALR & 1.118 & 0.633 & 0.284 \\
 & (0.572) & (0.766) & (0.650) \\
Peer Community Support ALR & 0.102 & 0.850 & -0.112 \\
 & (0.490) & (0.660) & (0.567) \\
Peer Rule Violations ALR & -0.098 & 0.987 & 1.626 \\
 & (1.031) & (0.954) & (1.210) \\
\addlinespace
\multicolumn{4}{l}{\textit{b. Female Unit: Receiver Profiles}} \\
Peer Personal Growth ALR & 0.423 & -0.723 & -2.158 \\
 & (0.613) & (0.627) & (1.295) \\
Peer Community Support ALR & 0.529 & 1.409 & 1.501 \\
 & (0.473) & (0.579) & (1.060) \\
Peer Rule Violations ALR & 0.044 & -0.183 & 0.829 \\
 & (0.765) & (1.660) & (2.049) \\
\addlinespace
\multicolumn{4}{l}{\textit{c. Male Unit: Sender Profiles}} \\
Peer Personal Growth ALR & 0.240 & 1.183 & -0.082 \\
 & (0.868) & (1.246) & (0.732) \\
Peer Community Support ALR & 1.072 & 1.171 & -1.122 \\
 & (1.036) & (1.273) & (1.009) \\
Peer Rule Violations ALR & -1.928 & -2.534 & 2.031 \\
 & (2.083) & (2.598) & (2.009) \\
\addlinespace
\multicolumn{4}{l}{\textit{d. Male Unit: Receiver Profiles}} \\
Peer Personal Growth ALR & 0.866 & -0.189 & -0.237 \\
 & (0.620) & (0.981) & (0.540) \\
Peer Community Support ALR & 0.366 & 2.552 & 0.430 \\
 & (0.765) & (0.888) & (0.773) \\
Peer Rule Violations ALR & 0.798 & 0.156 & 1.930 \\
 & (1.027) & (1.534) & (1.096) \\
\hline\hline
\end{tabular}
\begin{minipage}{15.0cm}
\footnotesize{Notes: Entries report latent-adjusted IV-2SLS estimates using the PPPI-adjusted LLM output. The three outcomes and endogenous peer regressors are additive log-ratios relative to Disruptive Conduct. Standard errors are provided in parentheses.}
\end{minipage}
\end{table}

This section presents our results for the peer effects model, in which the outcome variables are the ALRs of class probabilities with Disruptive Conduct as the baseline. Consequently, the outcome is a vector of three ALRs: Personal Growth ALR, Community Support ALR, and Rule Violations ALR. Our preferred specification estimates the network formation model using an additive and multiplicative effect latent space model with observed covariates (\text{lsm}-\text{cov}). Section \ref{comparisonofnet} discusses this choice of network model and shows that the latent space model with covariates best replicates key structural properties of the observed male and female unit networks.

We use the time difference (in days) between the residents' entry dates as a dyadic (node-pair) level covariate in the network formation model. We assume residents' entry dates are random and not explained by their observed or unobserved characteristics, since they depend on a complex function of court decisions and bed availability in TCs. However, the proximity of entry dates is an important determinant of whether residents exchange messages, since overlap in their time in the TC is a necessary condition for message exchange. We do not use proximity in exit dates or time in TC as a predictor of network formation, as unobserved resident variables likely affect these variables and are not predetermined to network formation.

First, we assess the strength of the instruments in explaining variability in the three endogenous regressors using the Cragg-Donald F statistic (CD statistic) and compare it with the Stock-Yogo critical values \citep{stock2002testing}. Table \ref{tab:weak_instruments} shows the CD statistic is between 12.42 and 14.84, while the critical value at the 10\% maximal IV relative bias is 10.25. This indicates the instruments are moderately strong. We note that maximal size-based Stock Yogo critical values are not available when the number of endogenous regressors is greater than 2 (three in our setup, one for each ALR), and hence we compare with only the relative bias critical values.

Tables \ref{rhotablewithadjust} and \ref{rhotablewithadjustmale} provide the results for IV 2SLS in the three BART derived ALRs (Personal Growth, Community Support and Rule Violations) after adjusting for the unobserved heterogeneity in the outcome model (Algorithm \ref{algorithm1}) for females and males, respectively. We choose the dimension of the multiplicative factor in the latent space model using data-driven cross-validation methods for networks in \cite{li2020network}. The number of data-driven dimensions for the multiplicative latent effects is 14 for the female and 19 for the male units, respectively. In addition, the network formation model includes one additive nodal latent variable.

The diagonal elements in Table \ref{rhotablewithadjust}, inform us about the spillovers of peer ALRs on female residents for the same class. There is evidence of a positive peer effect in all three ALRs relative to disruptive conduct at the 5\% significance level. Moreover, there is also evidence of peer spillovers from personal growth on the rule violations category. For comparison, Table \ref{rhotablewithoutadjust} presents the results when the unobservables (latent effects) are not adjusted for in the outcome model. Next, we repeat the analysis for the female unit, aggregating by receiver profiles (panel (b) of Table \ref{rhotablewithadjust}). When aggregated by the receiver profiles, we find evidence of direct spillover effects only in personal growth ALR and rule violations ALR at the 5\% significance level.

Next, we estimate peer effects for residents in the male units (Table \ref{rhotablewithadjustmale}). The estimates, when aggregated by sender profiles, suggest statistically significant positive direct spillovers of community support and rule violations ALR at the 95$\%$ confidence level. For all of the indirect spillovers, the 95\% confidence interval includes 0, amd are therefore insignificant. In panel (b) of Table \ref{rhotablewithadjustmale}, estimates are provided for receiver profile aggregation. The direct spillovers are all positive, but only the ones for personal growth and rule violations are statistically significant at the 5\% level. The indirect spillovers are again small relative to the standard errors. \footnote{For comparison, Table \ref{rhotablewithoutadjustmale} provides the estimates without adjusting for latent variables.}

Our results also show that same-category peer influence spillovers generally have a larger impact than across-category spillovers for both male and female units. Finally, we note that even though some of the values in our peer effect parameter matrix are greater than 1, the matrix $(I-D \otimes G)$  is invertible, which makes the structural multivariate peer effect model in Equation \ref{MSARmodel} stable (by Lemma 1 in \cite{zhu2020multivariate}).

\paragraph{PPPI results:}
Since the peer effects results presented so far use the LLM output ALRs, they may be biased. Consequently, we estimate peer effects with the PPPI estimator and the associated standard error estimators developed in this paper. Table \ref{tab:peer_effects_ppiplus_all} provides the results for PPPI IV 2SLS estimates in text class ALRs after adjusting for the unobserved factors in the outcome model for females and males, respectively.

First, we discuss panel (a) in Table \ref{tab:peer_effects_ppiplus_all} which is for female unit with sender profiles. The diagonal elements inform us about the spillovers of peer ALRs on female residents for the same class. There is evidence of a positive peer effect in all three ALRs relative to disruptive conduct, but only the peer effect on the Personal Growth ALR is significant at the conventional levels. As for the cross-spillovers 95\% confidence interval for every estimate includes 0 and hence is not significant at the conventional level.

\begin{table}[!h]
\centering
\caption{Network-Shuffling Tests for PPPI-Adjusted Peer Effects}
\label{tab:permutation_test}
\resizebox{0.85\textwidth}{!}{%
\begin{tabular}{p{6cm}p{1.5cm}p{1.5cm}p{1.5cm}p{1.5cm}p{1.5cm}p{1.5cm}}
\hline\hline
& \multicolumn{3}{c}{Conditional Model-SE Test} & \multicolumn{3}{c}{Nested Total-SE Test} \\
\cline{2-4} \cline{5-7}
Peer behavioral profile & Personal Growth & Community Support & Rule Violations & Personal Growth & Community Support & Rule Violations \\
\hline
\multicolumn{7}{l}{\textit{a. Female Unit: Sender Profiles}} \\
Peer Personal Growth ALR & 0.000 & 0.320 & 0.660 & 0.087 & 0.450 & 0.609 \\
Peer Community Support ALR & 0.880 & 0.080 & 0.880 & 0.788 & 0.197 & 0.836 \\
Peer Rule Violations ALR & 0.940 & 0.280 & 0.100 & 0.926 & 0.252 & 0.169 \\
\addlinespace
\multicolumn{7}{l}{\textit{b. Female Unit: Receiver Profiles}} \\
Peer Personal Growth ALR & 0.180 & 0.160 & 0.020 & 0.393 & 0.198 & 0.103 \\
Peer Community Support ALR & 0.120 & 0.020 & 0.020 & 0.250 & 0.054 & 0.135 \\
Peer Rule Violations ALR & 0.900 & 0.860 & 0.180 & 0.947 & 0.899 & 0.600 \\
\addlinespace
\multicolumn{7}{l}{\textit{c. Male Unit: Sender Profiles}} \\
Peer Personal Growth ALR & 0.740 & 0.000 & 0.960 & 0.768 & 0.275 & 0.889 \\
Peer Community Support ALR & 0.020 & 0.060 & 0.100 & 0.223 & 0.314 & 0.262 \\
Peer Rule Violations ALR & 0.060 & 0.020 & 0.120 & 0.316 & 0.275 & 0.331 \\
\addlinespace
\multicolumn{7}{l}{\textit{d. Male Unit: Receiver Profiles}} \\
Peer Personal Growth ALR & 0.000 & 0.540 & 0.320 & 0.160 & 0.801 & 0.565 \\
Peer Community Support ALR & 0.640 & 0.000 & 0.340 & 0.560 & 0.026 & 0.499 \\
Peer Rule Violations ALR & 0.080 & 0.860 & 0.000 & 0.325 & 0.887 & 0.085 \\
\hline\hline

\end{tabular}%
}
\begin{minipage}{15.0cm}
\footnotesize{Notes: Entries are two-sided empirical permutation $p$-values for the latent-adjusted IV-2SLS estimates using PPPI-adjusted behavioral profiles. The nested total-SE test combines 50 network permutations with 100 Rao--Wu resident-validation bootstrap draws per permutation and uses standard errors incorporating both model and resident-validation sampling variation. The conditional model-SE test holds the validated resident sample fixed while rewiring the network and uses model-based standard errors. 
}
\end{minipage}
\end{table}

An example of how we can interpret this result is as follows. An individual is more likely to send messages of personal growth as opposed to disruptive behavior (higher personal growth ALR) if their peers are also sending more messages of personal growth as opposed to disruptive behavior (higher personal growth ALR). For a diagonal element, such as $D_{\mathrm{PG},\mathrm{PG}}$, a one-unit increase in peers' average log ratio of Personal Growth to Disruptive Conduct changes the resident's own log ratio of Personal Growth to Disruptive Conduct by $D_{\mathrm{PG},\mathrm{PG}}$ units. Since the outcome and $GY$ are log ratios, a 1\% increase in the peer's Personal Growth-to-Disruptive Conduct ratio is associated with an approximately $1.12\%$ increase in the
resident's corresponding ratio.

Next, the same analysis is repeated for the female unit, but aggregation is done by receiver profiles (panel (b) of Table \ref{tab:peer_effects_ppiplus_all}). When aggregated by the receiver profiles, we find evidence of direct spillover effects for all three ALRs but only for community support ALR 95\% confidence interval does not include 0. As for indirect spillovers, similar to the sender profiles for all estimates 95\% confidence intervals include 0.

The peer effects estimation is next done for the residents in the male units with sender and receiver profiles (panels (c) and (d) in Table \ref{tab:peer_effects_ppiplus_all}). The estimates, when aggregated by sender profiles, suggest that there are no statistically significant direct or indirect spillovers at the 95$\%$ confidence level. In panel (d) of Table \ref{tab:peer_effects_ppiplus_all}, estimates are provided for receiver profile aggregation. The direct spillovers are all positive, but only the ones for Community Support ALR are statistically significant at the 5\% level. The indirect spillovers are again small relative to the standard errors. 
\begin{figure}[!htbp]
  \centering
  \caption{Comparison of Structural Properties of Networks}
\label{comparisonnetwork}
  \begin{subfigure}[t]{0.45\textwidth}
    \centering
    \includegraphics[width=\linewidth]{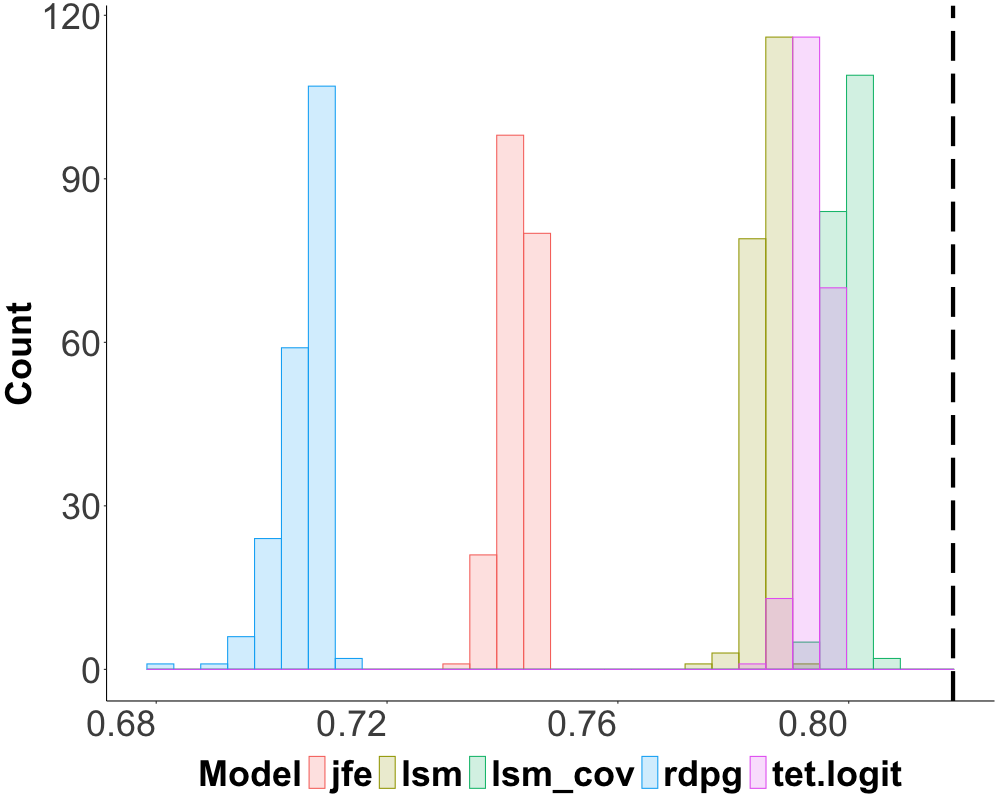}
    \caption{Modularity value (Male)}
  \end{subfigure}
  \begin{subfigure}[t]{0.45\textwidth}
    \centering
    \includegraphics[width=\linewidth]{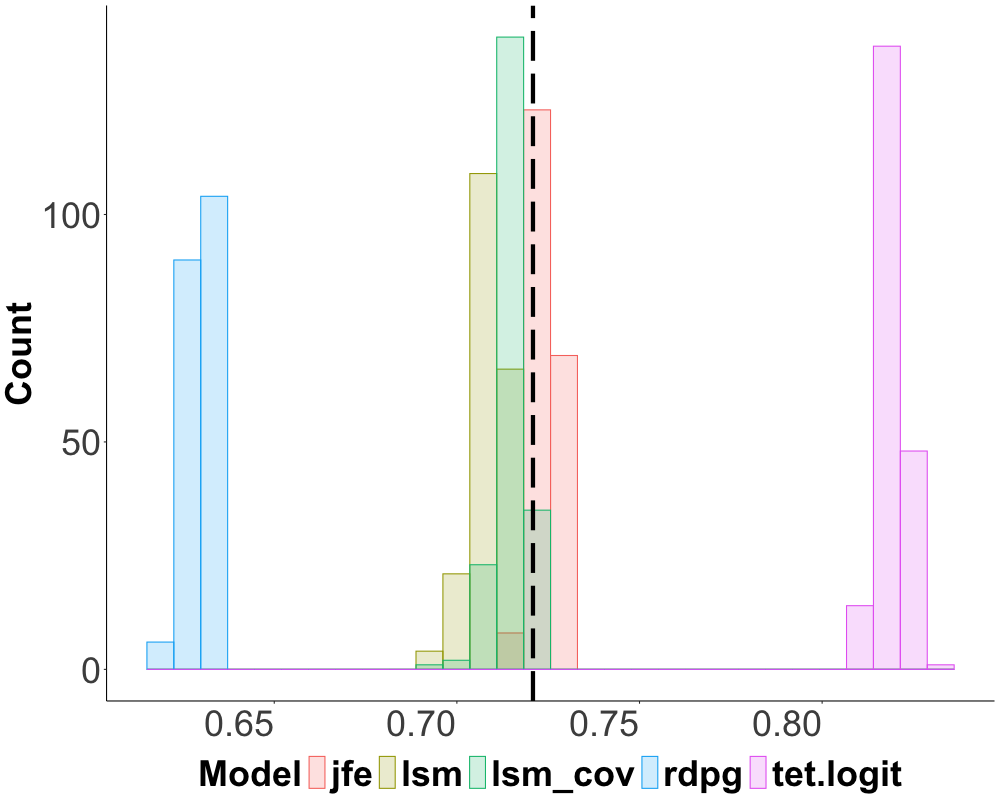}
    \caption{Modularity value (Female)}
  \end{subfigure}

  \vspace{1em}

  \begin{subfigure}[t]{0.45\textwidth}
    \centering
    \includegraphics[width=\linewidth]{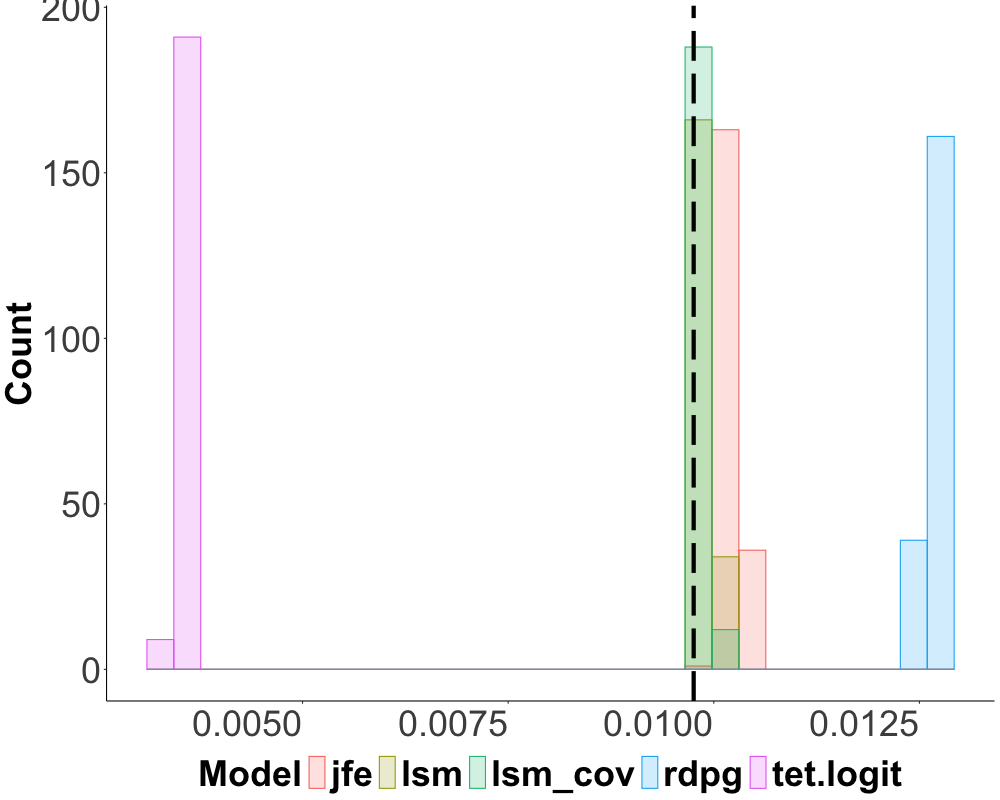}
    \caption{Std. deviation of row means (Male)}
  \end{subfigure}
  \begin{subfigure}[t]{0.45\textwidth}
    \centering
    \includegraphics[width=\linewidth]{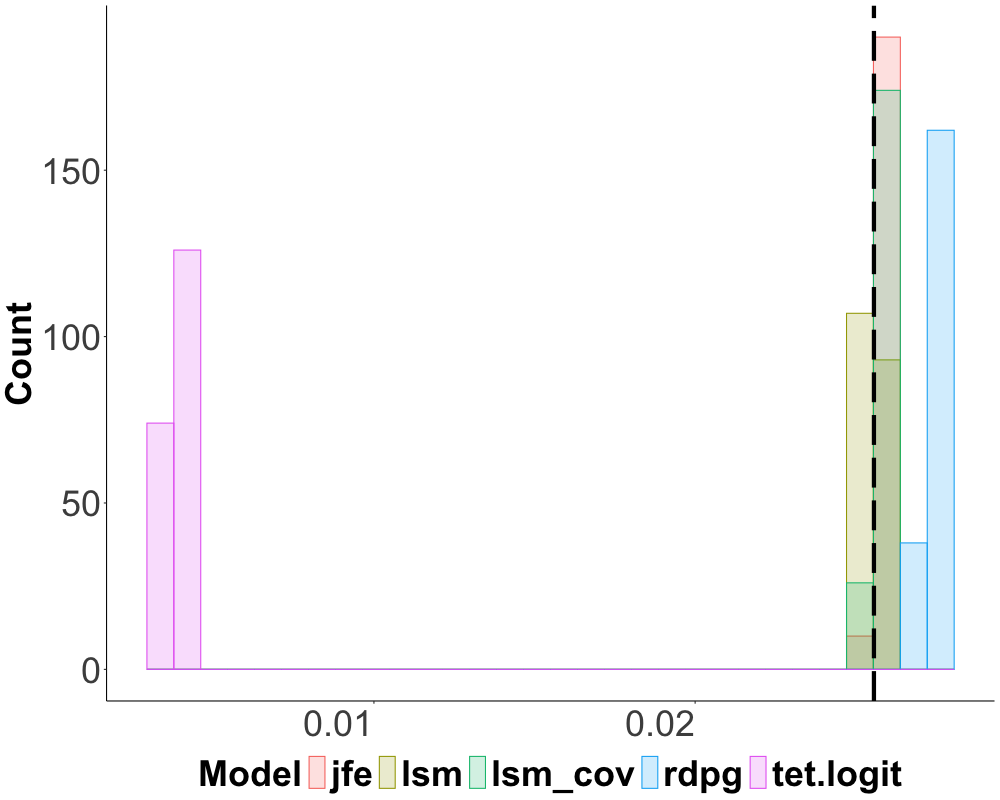}
    \caption{Std. deviation of row means (Female)}
  \end{subfigure}

  \vspace{1em}

  \begin{subfigure}[t]{0.450\textwidth}
    \centering
    \includegraphics[width=\linewidth]{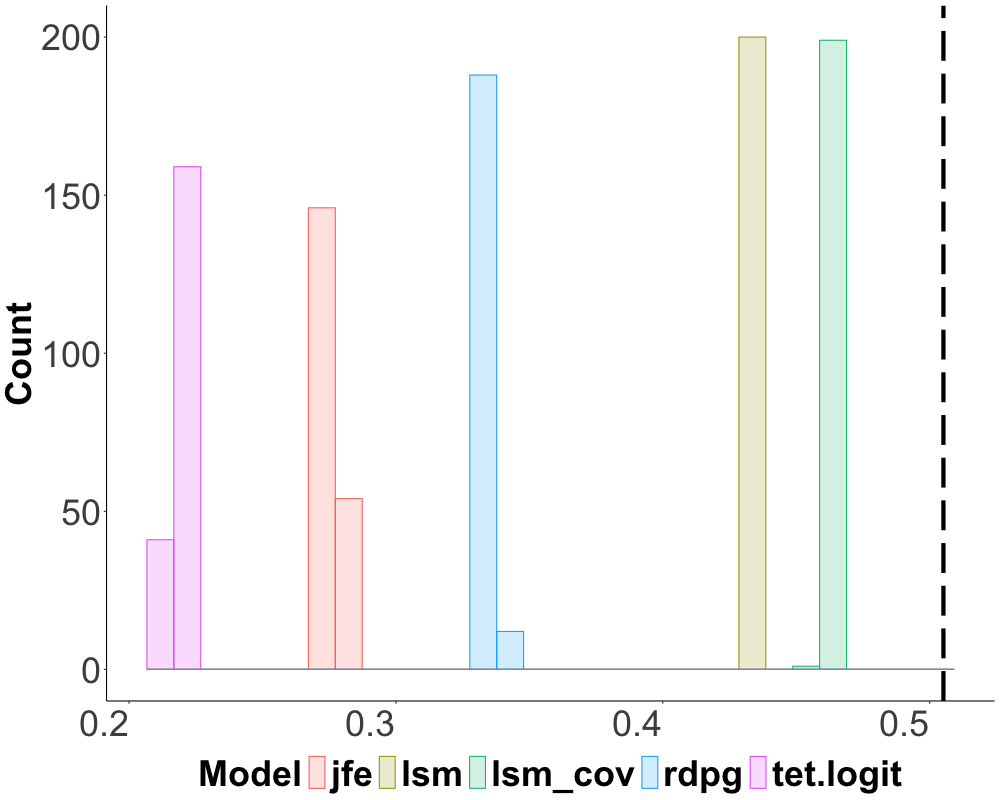}
    \caption{Transitivity of clustering coef. (Male)}
  \end{subfigure}
  \begin{subfigure}[t]{0.450\textwidth}
    \centering
    \includegraphics[width=\linewidth]{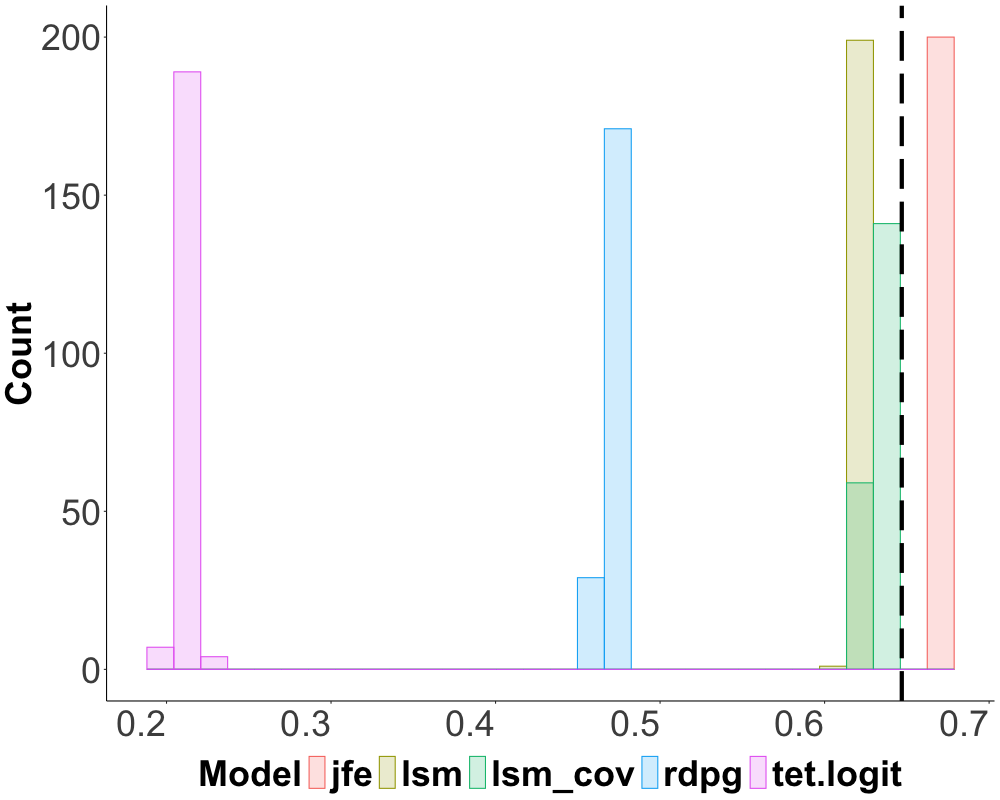}
    \caption{Transitivity of clustering coef. (Female)}
  \end{subfigure}
\begin{minipage}{15.0 cm}{\footnotesize{Notes: The dotted black vertical line displays the true network, and the distribution is obtained using 200 networks simulated from the fitted model, which varies for each of the estimators.}}
\end{minipage} 
\end{figure}
Overall, we find evidence of peer influence in the ALR of text class probabilities after recalibrating with PPPI methods developed in this paper. While the direction of the effects did not change much, especially for the direct spillovers, the associated standard errors increased substantially, resulting in fewer significant peer effects at conventional levels of statistical significance. This is due to the additional uncertainty introduced by having gold-standard human
labels for only a small fraction of residents owing to high cost of acquiring such labels. In particular, the PPPI standard
errors contain both the usual model-based uncertainty and a nonnegative
resident-sampling component arising from the human-validation design. We only had human labeled data for 10\% samples in the female unit and 6\% samples in the male unit. Therefore, this second component can be substantial and leads to noticeably wider
confidence intervals. Indeed, many of the estimated effects retain the same
direction as their LLM-only counterparts, but are estimated less precisely.
Thus, the fact that several same-category peer effects remain detectable despite
the limited amount of gold-standard data provides evidence that PPPI can recover
meaningful peer-effect signals while appropriately accounting for the additional
uncertainty due to human validation.

\paragraph{Robustness of Peer Effect Estimates:}

We validate the peer effect estimates obtained in Section \ref{peereffectlowsec} by comparing them against the distribution of peer effects obtained by randomly permuting the network of peers. Intuitively, the ``significant'' peer effects at conventional levels should be much higher in standardized absolute magnitude than this distribution since the peer effects with random peer assignment should be null. We implement this by randomly shuffling the network connections and repeating the peer-effect estimation on each simulated random network. Since the network changes with each simulation, we re-estimate the latent variables using our preferred latent-space model with covariates for each permutation.

We consider two versions of this permutation exercise. In the first, we condition on the realized human-validation sample and the resulting PPI outcomes and vary only the network. Since resident-validation uncertainty is held fixed in this exercise, the peer-effect estimates are standardized using only the model-based component of the PPPI standard error. This provides a conditional network-rewiring test that isolates whether the observed network generates unusually large peer-effect estimates relative to random peer assignment. In the second version of this exercise, for each randomly shuffled network, we additionally use a Rao-Wu rescaling bootstrap of the validated residents to approximate the uncertainty induced by resident validation. This does not require human labels for additional residents, as the bootstrap reweighs the observed human-AI differences among the validated residents. We standardize the resulting estimates using the total PPPI standard error, which includes both the model-based and resident-validation components. We compute the final empirical p-values as the proportion of simulated standardized statistics whose absolute values exceed the corresponding observed statistic. The first set of p-values therefore provides a network-shuffling robustness check conditional on the realized validation sample, while the second additionally incorporates resident-validation uncertainty and is more directly comparable to the PPPI results that are reported in Section \ref{peereffectlowsec}.
\begin{table}[!htbp]
    \centering
    \caption{Network Properties}
    \label{tab:network_summary}
    \begin{tabular}{lcc}
\hline\hline
Network property & Female & Male \\
\hline
Node size & 982 & 1649\\
Average degree & 61.442 & 37.124 \\
Density & 0.0626 & 0.0225 \\
Standard deviation of row means & 0.026 & 0.010 \\
Transitivity & 0.647 & 0.505 \\
Modularity & 0.722 & 0.819 \\
Average path length & 6.915 & 8.165 \\
\hline\hline
\multicolumn{3}{c}{ \begin{minipage}{12 cm}{\footnotesize{Notes: If $L$ denotes the total number of edges and $n$ denotes the number of nodes, then average degree is $\frac{2L}{n}$, while density is $\frac{2L}{n(n-1)}$.}}
\end{minipage}}
\end{tabular}
\end{table}

The comparison between the two network-shuffling tests further helps explain
the reduction in statistical significance after PPPI adjustment. When we hold
the realized human-validation sample fixed and isolate only network shuffling, several same-category peer effects remain high relative to randomly assigned peers. Once we add resident-validation uncertainty through the nested Rao-Wu procedure, the empirical $p$-values generally increase, in line with the larger total PPPI standard errors we observe in our results. Nevertheless, the strongest patterns remain the same across both. For example, the Community Support
same-category effect for male receiver profiles remains significant in the nested test, while the corresponding effect for female receiver profiles and
the Personal Growth effect for female sender profiles remain close to or slightly above conventional significance levels. The contrast between the conditional and nested tests therefore suggests that the weaker significance of some PPPI estimates is driven importantly by the limited human-validation sample rather
than by a lack of peer-effect signal in the observed network. Taken together, the PPPI results and the network-reshuffling exercises provide complementary evidence that some of the direct peer-effect patterns are robust, while
also reflecting the uncertainty introduced by the limited availability of gold-standard human labels.

\subsection{Choice of the Network Model}
\label{comparisonofnet}

Before discussing the choice of the network models, we provide Table \ref{tab:network_summary} summarizing the network properties for both male and female units in the correctional facility data. The average degree for the female facility is 61 with 982 nodes and it is 37 for male unit with 1649 nodes. These summary statistics are consistent with the sparsity regimes accommodated by our theory.

We compare five estimators for four network models in terms of their ability to replicate observed salient features of the TC networks. These estimators/models are (1) tetra logit estimator (tet.logit) and the (2) joint fixed effect maximum likelihood estimator (jfe) in \cite{graham2017econometric} that allow for single latent variable (additive) and covariates, 3) Spectral embedding in RDPG model \citep{athreya2017statistical} which allows for multivariate latent factors but no covariates (rdpg), 4) MLE in latent space model that allows for both multiplicative and additive latent factors but no observed covariates \cite{hoff2002latent} (lsm) and finally 5) MLE in latent space model from \cite{ma2020universal} that allows for both multiplicative and additive latent factors and observed covariates (\text{lsm}-\text{cov}).

We compare the estimators in terms of their ability to replicate three widely observed structural network properties: modularity \citep{girvan2002community}, standard deviation of row means, and transitivity or the clustering coefficient \citep{newman2018networks}. The steps deployed for comparison are as follows. Let us consider the network data for the male unit. We implement each estimator on the observed male unit network to estimate the respective fitted models. Next, we simulate 200 networks using these fitted models (200 for each estimator listed above). The structural properties over these simulated networks are then compared with the true network to assess how well the simulated data can replicate the truth. Figure \ref{comparisonnetwork} displays the performance. We observe that the latent space model that allows for both multiplicative and additive latent factors, along with covariates (proximity in entry dates), outperforms all other models for every property in the male unit. For the female unit, it performs the best for the standard deviation of row means and transitivity of the clustering coefficient, and at least as good as the other estimators for modularity.

\section{Conclusion}

Our findings have important policy implications for the criminal justice system. Recidivism imposes substantial costs. Reconviction rates range from 18-55\% across countries \citep{yukhnenko2023criminal}, with U.S. incarceration costs exceeding one trillion dollars annually \citep{pettus2016economic} as estimated in 2016. Recidivism prediction has been a central focus in criminal justice policy, with widely-used risk assessment tools such as COMPAS incorporating hundreds of pre-entry features \citep{dressel2018accuracy}. In contrast, our low-security correctional facility setting offers few covariates with limited predictive signal for recidivism. Our results demonstrate that language exchanges contain rich behavioral signals in correctional settings. AI-derived text measurements can be a powerful predictor in such settings. We also found significant peer effects in the behavioral profiles of the residents in these low-security correctional facilities. These results show a potential mechanism of peer encouragement and monitoring in these facilities.

In conclusion, we develop a framework that combines AI-based measurements of unstructured text with econometric methods to investigate peer effects in language. This framework can be applied more generally, beyond our application to recidivism, to study peer effects in contexts such as education and worker productivity, where researchers have access to written or spoken language exchange data. We further show how AI measurements can be combined with limited expensive gold standard human or experimental data to debias the measures. Finally, our results also highlight the predictive ability of AI measurements to complement traditional covariates for predicting economic outcomes.

\bibliographystyle{apalike}
\setlength{\bibsep}{-1.2pt} 

\bibliography{typ}
\clearpage

\appendix
\renewcommand{\thesection}{A\arabic{section}}
\renewcommand{\thesubsection}{A\arabic{section}.\arabic{subsection}}

\begin{center}
{\Huge\bfseries Supplementary Online Appendix}
\end{center}

\section{Mathematical Appendix}
\label{mathproofs}
\setcounter{figure}{0}
\setcounter{table}{0}

\paragraph{Notations:} For a matrix $\mb A$ we denote its spectral norm as $\|\mb A\|$ and the Frobenius norm as $\|\mb A\|_F$.  We use the notation $\mb A^\top$ to represent the transpose of a matrix $\mb A$. We use the stochastic order notation $o_\bbP(1)$ to mean that if $x_N = o_\bbP(1)$, then $\operatorname{plim}_{N\to\infty}x_N=0$. i.e., for any $\epsilon>0$ we have $\lim_{N\to\infty}\bbP(|x_N|>\epsilon)=0$.

\subsection{Proof of Proposition \ref{identify}}
\begin{proof}
We start from the moment condition:
\begin{align*}
    & \E [(\mb{K}_i- \E[\mb{K}_i|\mb{U}_i])^\top(\mb{Y}_i - \E [ \mb{Y}_i|\mb U_i] -(\mb{Z}_i- \E[\mb{Z}_i|\mb{U}_i])\beta)] \\
    & =   \E [(\mb{K}_i- \E[\mb{K}_i|\mb{U}_i])^\top((\mb{Z}_i- \E[\mb{Z}_i|\mb{U}_i])(\beta_0-\beta) + (\mathbf E_i-\E[\mb{e}_i|\mathbf U_i]))]\\
    & = \E [(\mb{K}_i- \E[\mb{K}_i|\mb{U}_i])^\top(\mb{Z}_i- \E[\mb{Z}_i|\mb{U}_i])(\beta_0-\beta)] + \E [(\mb{K}_i- \E[\mb{K}_i|\mb{U}_i])^\top(\mathbf E_i-\E[\mb{e}_i|\mathbf U_i])]
\end{align*}
 We can easily see that the conditional expectation of the second part, conditioning on $\mb{U}_i$ is (adapting the arguments from \cite{johnsson2021estimation}) 
 \begin{align*}
      \E [(\mb{K}_i- \E[\mb{K}_i|\mb{U}_i])^\top(\mathbf E_i-\E[\mb{e}_i|\mathbf U_i])| \mb{U}_i] & = \E [\mb{K}_i^\top \mb{e}_i |\mb{U}_i] - \E[ \mb{K}_i|\mb{U}_i]^\top \E[\mb{e}_i| \mb{U}_i]\\
      & = \E [ \E [\mb{K}_i^\top \mb{e}_i |\mb{U}_i, \mb{X},\mb{A}|\mb{U}_i]] - \E[ \mb{K}_i|\mb{U}_i]^\top \E[\mb{e}_i| \mb{U}_i]\\
      & = \E [ \mb{K}_i^\top \E [ \mb{e}_i |\mb{U}_i, \mb{X},\mb{A}|\mb{U}_i]] - \E[ \mb{K}_i|\mb{U}_i]^\top \E[\mb{e}_i| \mb{U}_i]\\
      & = \E [ \mb{K}_i^\top \E [ \mb{e}_i |\mb{U}_i]|\mb{U}_i]] - \E[ \mb{K}_i|\mb{U}_i]^\top \E[\mb{e}_i| \mb{U}_i]=0.
 \end{align*}
 Note that in the above, in line 2, we have used the fact that $\mb{K}_i$ is a function of $\mb{A}$ and $\mb{X}$, and therefore, conditional on them, we can take $\mb{K}_i$ outside the expectation. Further, in line 4, we have used the fact that under the assumptions on the model $\E [\mb{e}_i |\mb{U}_i, \mb{X},\mb{A}] = \E [\mb{e}_i|\mb{U}_i]$. 
 Therefore, by the law of iterated expectations,
 \[
 \E [(\mb{K}_i- \E[\mb{K}_i|\mb{U}_i])^\top(\mathbf E_i-\E[\mb{e}_i|\mathbf U_i])] =0.
 \]
 Therefore, the moment condition implies, $\bm\beta=\bm\beta_0$, as long as the $3p \times (m+2p)$ matrix $\E [(\mb{K}_i- \E[\mb{K}_i|\mb{U}_i])^\top((\mb{Z}_i- \E[\mb{Z}_i|\mb{U}_i])]$ is invertible.
\end{proof}

\subsection{Proof of Theorem \ref{thm:clt}}
We derive the asymptotic distribution of the 2SLS estimator \( \hat{\bm\beta}_{v,\text{IV}} \) in three steps similar to  the structure in \cite{johnsson2021estimation}.

We start from the IV-2SLS estimator described in \ref{eq-2sls}. Substituting $\mb{Y}_N=\mb{Z}_N\bm\beta_0 + \mb{E}_N$ in Equation \ref{eq-2sls}, we have 
\begin{align*}
   \hat{\bm\beta}_{2SLS}-\bm\beta_0 =& \left(\mb{Z}_N^\top \mb{M}_{\hat{\bm\Phi}_N} \mb{K}_N \left(\mb{K}_N^\top \mb{M}_{\hat{\bm\Phi}_N} \mb{K}_N\right)^{-1} \mb{K}_N^\top \mb{M}_{\hat{\bm\Phi}_N} \mb{Z}_N\right)^{-1}\nonumber \\
    &\times  \mb{Z}_N^\top \mb{M}_{\hat{\bm\Phi}_N} \mb{K}_N \left(\mb{K}_N^\top \mb{M}_{\hat{\bm\Phi}_N} \mb{K}_N\right)^{-1}\mb{K}_N^\top \mb{M}_{\hat{\bm\Phi}_N} \left( \tilde{\mb{e}}_N+ \mb{h}^E(\mb{U}_N) -  \hat{\bm\Phi}_N \bm \alpha_{L_N}^{E} \right).
\end{align*}
This is because $\tilde{\mb{e}}_N = \mb{e}_N - \mb h^E(\mb{U}_N)$, and $\mb{M}_{\hat{\bm\Phi}_N} \hat{\bm\Phi}_N=\bm 0$. 
The next lemma is our first step in developing an asymptotic normality result.
\begin{Lemma}
Assume the assumptions in section \ref{sec:assmp}. Then, we have
\begin{enumerate}
    \item $\frac{1}{N} (\mb K_N^\top \mb P_{\hat{\bm{\Phi}}_N} \mb Z_N - \mb K_N^\top \mb P_{\bm{\Phi}_N} \mb Z_N) = o_\bbP(1)$.
    \item $\frac{1}{N} (\mb K_N^\top \mb P_{\hat{\bm{\Phi}}_N} \mb K_N - \mb K_N^\top \mb P_{\bm{\Phi}_N} \mb K_N) = o_\bbP(1)$.
    \item $\frac{1}{\sqrt{N}} (\mb K_N^\top \mb P_{\hat{\bm{\Phi}}_N} \tilde{\mb{e}}_N - \mb K_N^\top \mb P_{\bm{\Phi}_N} \tilde{\mb{e}}_N) = o_\bbP(1)$.
    \item $\frac{1}{\sqrt{N}} \{\mb K_N^\top \mb M_{\hat{\bm{\Phi}}_N} (\mb{h}^E(\mb{U}_N) - \hat{\bm{\Phi}}_N \bm{\alpha}_{L_N}^E)\} = o_\bbP(1)$.
\end{enumerate}
Consequently, 
\begin{align*}
    \sqrt{N} \left( \hat{\bm\beta}_{2SLS} - \bm\beta_0 \right)
    &= \left( \frac{1}{N} \mb Z_N^\top \mb M_{\bm{\Phi}_N} \mb K_N 
        \left( \frac{1}{N}\mb K_N^\top \mb M_{\bm{\Phi}_N} \mb K_N \right)^{-1}\frac{1}{N} \mb K_N^\top \mb M_{\bm{\Phi}_N} \mb Z_N \right)^{-1} \\
    & \times \frac{1}{N} \mb Z_N^\top \mb M_{\bm{\Phi}_N} \mb K_N 
        \left( \frac{1}{N}\mb K_N^\top \mb M_{\bm{\Phi}_N} \mb K_N \right)^{-1} 
        \frac{1}{\sqrt{N}} \mb K_N^\top \mb M_{\bm{\Phi}_N} \tilde{\mb{e}}_N + o_\bbP(1).
\end{align*}
\label{lem:samplingerror}
\end{Lemma}

The next lemma further shows that under our assumptions, the sieve estimators effectively approximate the various conditional mean functions.
This is done in the next lemma.
\begin{Lemma}
Under the assumptions in section \ref{sec:assmp}, we have the following results
\begin{enumerate}
    \item $\frac{1}{N} \mb K_N^\top \mb{M}_{\bm{\hat\Phi}_N} \mb Z_N = \frac{1}{N}\sum_{i\le N}[\mb{K}_i-\mb{h}^K(\mb{U}_i)][\bm z_i-\mb{h}^Z(\mb{U}_i)]^\top+o_\bbP(1)$.
    \item $\frac{1}{N} \mb K_N^\top \mb{M}_{\bm{\hat\Phi}_N} \mb K_N = \frac{1}{N}\sum_{i\le N}[\mb{K}_i-\mb{h}^K(\mb{U}_i)][\mb{K}_i-\mb{h}^K(\mb{U}_i)]^\top+o_\bbP(1)$.
    \item $\frac{1}{\sqrt{N}}\mb K_N^\top \mb{M}_{\bm{\hat\Phi}_N} \tilde{\mb{e}}_N=\frac1{\sqrt{N}}\sum_{i\le N}[\mb{K}_i-\mb{h}^K(\mb{U}_i)]{\tilde{\mb{e}}_i}^\top+o_\bbP(1)$.
\end{enumerate}
As a consequence of these results,
$\sqrt{N} \left( \hat{\bm\beta}_{2SLS} - \bm\beta_0 \right)$ can be written as
\begin{align}\label{eq:clt-obj}
    \left( \frac{1}{N} \tilde{\mb{Z}}_N^\top\tilde{\mb{K}}_N \left( \tilde{\mb{K}}_N^\top\tilde{\mb{K}}_N \right)^{-1}
        \tilde{\mb{K}}_N^\top\tilde{\mb{Z}}_N \right)^{-1}\times \frac{1}{\sqrt{N}}\tilde{\mb{Z}}_N^\top\tilde{\mb{K}}_N \left( \tilde{\mb{K}}_N^\top\tilde{\mb{K}}_N \right)^{-1}
        \tilde{\mb{K}}_N^\top\tilde{\mb{e}}_N + o_\bbP(1)
\end{align}
\label{lem:sieveapprox}
\end{Lemma}

Now it remains to be shown that the first term in \eqref{eq:clt-obj} is asymptotically normal.
For this, we show that the vectorized estimator $\sqrt{N}\left( \hat{\bm\beta}_{v,2SLS} - \bm\beta_v^0 \right)\in\mathbb{R}^{m(m+2p)}$ is asymptotically normal.
First, by vectorizing  \eqref{eq:clt-obj}, we note that,
\begin{align}
    &\sqrt{N}\left( \hat{\bm\beta}_{v,2SLS} - \bm\beta_v^0 \right)\nonumber \\
    &= \frac{1}{N}\left[\mb{I}_m\otimes\left\{\tilde{\mb{Z}}_N^\top\tilde{\mb{K}}_N \left( \tilde{\mb{K}}_N^\top\tilde{\mb{K}}_N \right)^{-1}\tilde{\mb{K}}_N^\top\tilde{\mb{Z}}_N\right\}^{-1}\tilde{\mb{Z}}_N^\top\tilde{\mb{K}}_N \left( \tilde{\mb{K}}_N^\top\tilde{\mb{K}}_N \right)^{-1}\right]\nonumber \\
    &\qquad\times\frac{1}{\sqrt{N}}\operatorname{vec}(\tilde{\mb{K}}_N^\top{\tilde{\mb{e}}_N})+o_\bbP(1).\label{eq:clt1}
\end{align}
Let $\mb Q\in\mathbb{R}^{3p\times m}$ be an arbitrary matrix with $\|\mb Q\|_F^2=1$ and $\mb q:=\operatorname{vec}(\mb Q)$.
Writing the $i$-th row of $\tilde{\mb{e}}_N$ by $\tilde{\mb{e}}_i$ and the $i$-th row of $\tilde{\mb{K}}_N$ by $\tilde{\mb{K}}_i$,
\[\operatorname{vec}(\tilde{\mb{K}}_N^\top{\tilde{\mb{e}}_N})^\top\mb q=\tr({\tilde{\mb{e}}_N}^\top\tilde{\mb{K}}_N\mb Q)=\sum_{i\le N}\tilde{\mb{K}}_i^\top\mb Q\tilde{\mb{e}}_i.\]
Now, let ${\c F}_{Ni}:=\sigma\left({\c F}_N\cup\sigma(\{\mb{e}_s;\,s\le i\})\right)$ for ${\c F}_N:=\sigma(\mb{X}_N,\mb{A}, \mb{U}_N)$ where $\sigma(X)$ denotes a sigma algebra generated by a random variable $X$.
Then, $T_{Ni}:=\frac{1}{\sqrt{N}}\tilde{\mb{K}}_i^\top\mb Q\tilde{\mb{e}}_i$ is a martingale difference adapted to ${\c F}_{Ni}$.
For simplicity, let $\bm\iota_i^E:=\mb Q\tilde{\mb{e}}_i$, which are i.i.d from $(0_{3p},\mb Q\Sigma_{\tilde E}\mb Q^\top)$ for $\Sigma_{\tilde E}:=\Var(\tilde{\mb{e}}_1)$ ($=\mb{V}$ by the first and second moment conditions on $\mb{e}_i$'s).
We check the sufficient conditions for Corollary 3.1 of \cite{hall2014martingale} to apply the Martingale central limit theorem.
\begin{itemize}
    \item[(i).] For any $\epsilon>0$, we have
    \[\sum_i\E[T_{Ni}^21_{\{|T_{Ni}|>\epsilon\}}|{\c F}_{N,i-1}]=o_\bbP(1).\]
    \begin{proof}
        We have
        \begin{align*}
            \sum_i\E[T_{Ni}^21_{\{|T_{Ni}|>\epsilon\}}|{\c F}_{N,i-1}]&\le
            \frac{1}{\epsilon^2}\sum_i\E[T_{Ni}^4|{\c F}_{N,i-1}] \\
            &\qquad=\frac{1}{N^2\epsilon^2}\sum_i\E[({\bm\iota_i^E}^\top\tilde{\mb{K}}_i\tilde{\mb{K}}_i^\top\bm\iota_i^E)^2|{\c F}_{N,i-1}].
        \end{align*}
        By \citet[equation 45 in Lemma A.8]{chen2023community},
        \begin{align*}
            \E[({\bm\iota_i^E}^\top\tilde{\mb{K}}_i\tilde{\mb{K}}_i^\top\bm\iota_i^E)^2|{\c F}_{N,i-1}]=\V({\bm\iota_i^E}^\top\tilde{\mb{K}}_i\tilde{\mb{K}}_i^\top\bm\iota_i^E|{\c F}_{N,i-1})+\tr(\tilde{\mb{K}}_i\tilde{\mb{K}}_i^\top\mb Q\Sigma_{\tilde E}\mb Q^\top)^2\\
            \le c\tr[(\tilde{\mb{K}}_i\tilde{\mb{K}}_i^\top)^2]+\tr(\tilde{\mb{K}}_i\tilde{\mb{K}}_i^\top\mb Q\Sigma_{\tilde E}\mb Q^\top)^2\\
            \le \tr[(\tilde{\mb{K}}_i\tilde{\mb{K}}_i^\top)^2]\left\{c+3p\tr(\mb Q\Sigma_{\tilde E}\mb Q^\top)^2\right\},
        \end{align*}
        so it suffices to verify $\frac{1}{N^2}\sum_i\tr[(\tilde{\mb{K}}_i\tilde{\mb{K}}_i^\top)^2]=o_\bbP(1)$.
        Note that
        \[\frac{1}{N^2}\sum_i(\tilde{\mb{K}}_i^\top\tilde{\mb{K}}_i)^2=\frac{1}{N^2}\sum_i\left(\sum_{j\le3p}\tilde k_{ij}^2\right)^2\le\frac{3p}{N^2}\sum_{i,j}\tilde k_{ij}^4=o_\bbP(1)\]
        which can be shown similarly as Lemma \ref{lem:cov-cons}.
        For this, we show that
        \begin{align}\label{eq:ggx}
            \frac{1}{N^2}\sum_{i\le N,j\le p}(\mb g_i^\top\mb{G}\mb{X}_j)^4=o_\bbP(1)
        \end{align}
        for example. As $\frac{1}{N^2}\sum_j\|\mb{X}_j\|^4=\frac{1}{N^2}\sum_j(\sum_{i\le N}X_{ij}^2)^2=O_\bbP(1)$, it remains to show that $\sum_i\|\mb g_i^\top\mb{G}\|^4=o_\bbP(1)$.
        Observe that
        \[\sum_i\|\mb g_i^\top\mb{G}\|^4\le \frac{1}{\dmin^8}\sum_{i=1}^N\left\{\sum_{j=1}^N\left(\sum_{k=1}^Na_{ik}a_{kj}\right)^2\right\}^2=\frac{W}{\dmin^8}\]
        for which we defined an 8-length path subgraph sum as
        $$W = \sum_{i, j_1, j_2} \sum_{k_1, \dots, k_4} a_{ik_1} a_{k_1 j_1} a_{ik_2} a_{k_2 j_1} a_{ik_3} a_{k_3 j_2} a_{ik_4} a_{k_4 j_2}.$$
        On an event $\{\dmin\ge d_N\}$, we have $W/\dmin^8\le W/d_N^8$, and it suffices to show $\E(W)/d_N^8=o(1)$ and use Markov inequality to obtain \eqref{eq:ggx}.
        
        We track the number of distinct vertices $v$ and the resulting number of distinct edges $e$. 
        When $v=7$, the $e=8$ corresponding edges are mutually independent. There are ${\c O}(N^7)$ such combinations, and the expectation of each product is bounded by $P_{\max}^8$. Thus, this configuration contributes ${\c O}(N^7 P_{\max}^8)$.

        Next, suppose one pair of intermediate indices overlaps (e.g., $k_1 = k_2$). The number of free vertices reduces to $v=6$, yielding ${\c O}(N^6)$ combinations. Crucially, the edge product $a_{ik_1} a_{k_1 j_1} a_{ik_2} a_{k_2 j_1}$ collapses to $a_{ik_1}^2 a_{k_1 j_1}^2 = a_{ik_1} a_{k_1 j_1}$, reducing the total number of distinct edges to $e=6$. This configuration contributes ${\c O}(N^6 P_{\max}^6)$.

        Additional overlaps further decrease $v$ and $e$. Since $P_{\max} \in (0,1)$, as long as the network average degree diverges ($N P_{\max} \to \infty$), all higher-order collision terms are strictly dominated by the $v=6$ term. 

        Consequently, integrating over the latent distribution, the unconditional expectation is strictly bounded by the structural sum:
        \[\E[W] = {\c O}(N^7 P_{\max}^8+N^6 P_{\max}^6).\]
        Therefore, we require $N^7P_{\max}^8+N^6P_{\max}^6=o(d_N^8)$.
    \end{proof}
    \item[(ii).] 
    \[\sum_i\E[T_{Ni}^2|{\c F}_{N,i-1}]=\mb q^\top(\Sigma_{\tilde E}\otimes\Sigma_{\tilde K})\mb q+o_\bbP(1).\]
    \begin{proof}
        By Lemma \ref{lem:cov-cons}, we have
        \[\frac{1}{N}\sum_i\E({\bm\iota_i^E}^\top\tilde{\mb{K}}_i\tilde{\mb{K}}_i^\top\bm\iota_i^E|{\c F}_{N,i-1})=\tr\left[\frac{1}{N}\sum_i\tilde{\mb{K}}_i\tilde{\mb{K}}_i^\top\E({\bm\iota_1^E}{\bm\iota_1^E}^\top)\right]=\tr(\Sigma_{\tilde K}\mb Q\mb{V}\mb Q^\top)+o_\bbP(1).\]
    \end{proof}
\end{itemize}
This shows that $\frac{1}{\sqrt{N}}\operatorname{vec}(\tilde{\mb{K}}_N^\top{\tilde{\mb{e}}_N})$ has the limiting distribution ${\c N}(0,\Sigma_{\tilde E}\otimes\Sigma_{\tilde K})$.
Lemma \ref{lem:cov-cons} implies that
\begin{align*}
    \frac{1}{N}\mb{I}_m\otimes\left[\left\{\tilde{\mb{Z}}_N^\top\tilde{\mb{K}}_N \left( \tilde{\mb{K}}_N^\top\tilde{\mb{K}}_N \right)^{-1}\tilde{\mb{K}}_N^\top\tilde{\mb{Z}}_N\right\}^{-1}\tilde{\mb{Z}}_N^\top\tilde{\mb{K}}_N \left( \tilde{\mb{K}}_N^\top\tilde{\mb{K}}_N \right)^{-1}\right] \\
    =\mb{I}_m\otimes\left\{\left(\Sigma_{\tilde Z\tilde K}\Sigma_{\tilde K}^{-1}\Sigma_{\tilde K\tilde Z}\right)^{-1}\Sigma_{\tilde Z\tilde K}\Sigma_{\tilde K}^{-1}\right\}+o_\bbP(1),
\end{align*}
so the conclusion of Theorem \ref{thm:clt} follows by applying Slutsky's theorem to \eqref{eq:clt1}. For the asymptotic covariance form, we can use the relationship
\[(\mb{I}\otimes\mb{B})(\mb{V}_1\otimes\mb{V}_2)(\mb{I}\otimes\mb{B})^\top=\mb{V}_1\otimes(\mb{B}\mb{V}_2\mb{B}^\top)\]
for $\mb{B},\mb{V}_1,\mb{V}_2$ with appropriate dimensions.
\subsection{Proof of additional lemmas}

\subsubsection{Proof of Lemma \ref{lem:samplingerror}}

This part of the proof follows similar logic as in \cite{johnsson2021estimation}; however, since our terms are matrix-valued instead of vector-valued, our bounds involve the matrix spectral and Frobenius norms, as well as the matrix trace. We first prove the four results claimed in the lemma.

\begin{proof}
\begin{enumerate}
    \item \label{lem1:prf1} We start with the following expansion for part 1 of the lemma.
    \begin{align*}
        & \frac{1}{N} (\mb K_N^\top \mb P_{\hat{\bm{\Phi}}_N} \mb Z_N - \mb K_N^\top \mb P_{\bm{\Phi}_N} \mb Z_N) \\
        & = \mb K_N^\top \left\{ \frac{\hat{\bm{\Phi}}_N}{N} \left( \frac{\hat{\bm{\Phi}}_N^\top\hat{\bm{\Phi}}_N}{N} \right)^{-1} \frac{\hat{\bm{\Phi}}_N^\top}{N} - \frac{\bm{\Phi}_N}{N} \left( \frac{\bm{\Phi}_N^\top\bm{\Phi}_N}{N}\right)^{-1} \frac{\bm{\Phi}_N^\top}{N} \right\}\mb Z_N \\
        &= \mb K_N^\top \left\{ \frac{\hat{\bm{\Phi}}_N - \bm{\Phi}_N}{N} \left( \frac{\hat{\bm{\Phi}}_N^\top\hat{\bm{\Phi}}_N}{N} \right)^{-1} \frac{\hat{\bm{\Phi}}_N^\top}{N} + \frac{\bm{\Phi}_N}{N} \left( \left( \frac{\hat{\bm{\Phi}}_N^\top\hat{\bm{\Phi}}_N}{N} \right)^{-1} - \left( \frac{\bm{\Phi}_N^\top\bm{\Phi}_N}{N} \right)^{-1} \right) \frac{\bm{\Phi}_N^\top}{N} \right. \\
        &\qquad\qquad \left. + \frac{\bm{\Phi}_N}{N} \left( \frac{\hat{\bm{\Phi}}_N^\top\hat{\bm{\Phi}}_N}{N} \right)^{-1} \frac{\hat{\bm{\Phi}}_N^\top - \bm{\Phi}_N^\top}{N} \right\}\mb Z_N \\
        & = \mb K_N^\top\left\{ \frac{\hat{\bm{\Phi}}_N - \bm{\Phi}_N}{N} \left( \frac{\hat{\bm{\Phi}}_N^\top\hat{\bm{\Phi}}_N}{N} \right)^{-1} \frac{\hat{\bm{\Phi}}_N - \bm{\Phi}_N}{N} + \frac{\hat{\bm{\Phi}}_N^\top - \bm{\Phi}_N^\top}{N} \left( \frac{\hat{\bm{\Phi}}_N^\top\hat{\bm{\Phi}}_N}{N} \right)^{-1} \frac{\bm{\Phi}_N^\top}{N} \right. \\
        &\qquad\qquad + \frac{\bm{\Phi}_N}{N} \left( \left( \frac{\hat{\bm{\Phi}}_N^\top\hat{\bm{\Phi}}_N}{N} \right)^{-1} - \left( \frac{\bm{\Phi}_N^\top\bm{\Phi}_N}{N} \right)^{-1} \right) \frac{\bm{\Phi}_N^\top}{N} \\
        &\qquad\qquad \left. + \frac{\bm{\Phi}_N}{N} \left( \frac{\hat{\bm{\Phi}}_N^\top\hat{\bm{\Phi}}_N}{N} \right)^{-1} \frac{\hat{\bm{\Phi}}_N^\top - \bm{\Phi}_N^\top}{N} \right\}\mb Z_N \\
        &= I_1 + I_2 + I_3 + I_4, \text{ say.}
    \end{align*}
    We show that all these terms $\|I_1\|, \|I_2\|, \|I_3\|, \|I_4\| = o_\bbP(1)$ which collectively imply
    \[
    \left\| \frac{1}{N} (\mb K_N^\top \mb P_{\hat{\bm{\Phi}}_N} \mb Z_N - \mb K_N^\top \mb P_{\bm{\Phi}_N} \mb Z_N) \right\| = o_\bbP(1).
    \]
    As we assume that $m,p$ are fixed, the spectral-norm being $o_\bbP(1)$, is sufficient for all entries to be $o_\bbP(1)$. First note that,
    \begin{align*}
        \| I_1 \| &\leq \left\| \frac{\mb K_N}{\sqrt{N}} \right\|_F \left\| \frac{\hat{\bm{\Phi}}_N - \bm{\Phi}_N}{\sqrt{N}} \right\|_F^2 \left\| \left( \frac{\hat{\bm{\Phi}}_N^\top \hat{\bm{\Phi}}_N}{N} \right)^{-1} \right\| \left\| \frac{\mb Z_N}{\sqrt{N}} \right\|_F \\
        &\qquad= O_\bbP(1) \underbrace{{\c O}_{\t\bbP}\left(\sum_{k\le L_N}\zeta_1(k)^2\frac{\log^{2c}N}{N}\right)}_{\rm Lemma \ \ref{lmm:sieve:est}} \underbrace{O_\bbP(1)}_{\rm Lemma\ \ref{lmm:sieve:rank}} {\c O}(1) = o_\bbP(1)
    \end{align*}
    under Assumption \ref{assmp:lipsch} using the results of Lemmas \ref{lmm:sieve:est} and \ref{lmm:sieve:rank}. Next,
    \begin{align*}
    \| I_2 \| &\leq \left\| \frac{\mb K_N}{\sqrt{N}} \right\|_F \left\| \frac{\hat{\bm{\Phi}}_N - \bm{\Phi}_N}{\sqrt{N}} \right\|_F \left\| \left( \frac{\hat{\bm{\Phi}}_N^\top \hat{\bm{\Phi}}_N}{N} \right)^{-1} \right\| \left\| \frac{\bm{\Phi}_N}{\sqrt{N}} \right\|_F \left\| \frac{\mb Z_N}{\sqrt{N}} \right\|_F\\
    &\qquad \le O_\bbP(1){\c O}_{\t\bbP}\left(\sqrt{\zeta_0(L_N)^2\sum_{k\le L_N}\zeta_1(k)^2\frac{\log^{2c}N}{N}}\right)=o_\bbP(1)
    \end{align*}
    by $\|\bm\Phi_N\|_F^2=\sum_{i\le N}\|\bm{\phi}^{L_N}(\mb{U}_i)\|^2\le N\zeta_0(L_N)^2$. $\|I_4\|=o_\bbP(1)$ follows in a similar way.
    Finally noting that,
    \begin{align*}
        I_3 &= \frac{\mb K_N^\top \bm{\Phi}_N}{N} \left( \frac{\hat{\bm{\Phi}}_N^\top \hat{\bm{\Phi}}_N}{N} \right)^{-1}
        \left\{ \left( \frac{{\bm{\Phi}}_N^\top{\bm{\Phi}}_N}{N} \right) - \left(\frac{\hat{\bm{\Phi}}_N^\top\hat{\bm{\Phi}}_N}{N} \right) \right\}
        \left( \frac{\bm{\Phi}_N^\top \bm{\Phi}_N}{N} \right)^{-1} \frac{\bm{\Phi}_N^\top \mb Z_N}{N} \\
        &= \frac{\mb K_N^\top \bm{\Phi}_N}{N} \left( \frac{\hat{\bm{\Phi}}_N^\top \hat{\bm{\Phi}}_N}{N} \right)^{-1}
        \left( \frac{{\bm{\Phi}}_N^\top ({\bm{\Phi}}_N - \hat{\bm{\Phi}}_N)}{N} \right)
        \left( \frac{\bm{\Phi}_N^\top \bm{\Phi}_N}{N} \right)^{-1} \frac{\bm{\Phi}_N^\top \mb Z_N}{N} \\
        &\quad + \frac{\mb K_N^\top \bm{\Phi}_N}{N} \left( \frac{\hat{\bm{\Phi}}_N^\top \hat{\bm{\Phi}}_N}{N} \right)^{-1}
        \left( \frac{({\bm{\Phi}}_N - \hat{\bm{\Phi}}_N)^\top \hat{\bm{\Phi}}_N}{N} \right)
        \left( \frac{\bm{\Phi}_N^\top \bm{\Phi}_N}{N} \right)^{-1} \frac{\bm{\Phi}_N^\top \mb Z_N}{N}
    \end{align*}
    and we get $\|I_3\|=o_\bbP(1)$. 
    
    \item This can be shown in the same way as \ref{lem1:prf1}.

    \item Using the similar expansion as \ref{lem1:prf1}, we have
    \begin{align*}
        &\frac{1}{\sqrt{N}} (\mb{K}_N^\top \mb P_{\hat{\bm\Phi}_N}\tilde{\mb{e}}_N - \mb{K}_N^\top \mb P_{\bm\Phi_N}\tilde{\mb{e}}_N) \\
        &= \mb K_N^\top \left\{ \frac{\hat{\bm{\Phi}}_N - \bm{\Phi}_N}{N} \left( \frac{\hat{\bm{\Phi}}_N^\top\hat{\bm{\Phi}}_N}{N} \right)^{-1} \frac{\hat{\bm{\Phi}}_N - \bm{\Phi}_N}{N} + \frac{\hat{\bm{\Phi}}_N - \bm{\Phi}_N}{N} \left( \frac{\hat{\bm{\Phi}}_N^\top\hat{\bm{\Phi}}_N}{N} \right)^{-1} \frac{\bm{\Phi}_N}{N} \right. \\
        &\qquad\qquad + \frac{\bm{\Phi}_N}{N} \left( \left( \frac{\hat{\bm{\Phi}}_N^\top\hat{\bm{\Phi}}_N}{N} \right)^{-1} - \left( \frac{\bm{\Phi}_N^\top\bm{\Phi}_N}{N} \right)^{-1} \right) \frac{\bm{\Phi}_N^\top}{N} \\
        &\qquad\qquad \left. + \frac{\bm{\Phi}_N}{N} \left( \frac{\hat{\bm{\Phi}}_N^\top\hat{\bm{\Phi}}_N}{N} \right)^{-1} \frac{\hat{\bm{\Phi}}_N - \bm{\Phi}_N}{N} \right\}\tilde{\mb{e}}_N \\
        &= J_1+J_2+J_3+J_4, \text{ say.}
    \end{align*}
    First,
    \begin{align*}
        \| J_1 \| &\leq \left\| \frac{\mb K_N}{\sqrt{N}} \right\|_F \left\| \frac{\hat{\bm{\Phi}}_N - \bm{\Phi}_N}{\sqrt{N}} \right\|_F \left\| \left( \frac{\hat{\bm{\Phi}}_N^\top \hat{\bm{\Phi}}_N}{N} \right)^{-1} \right\| \left\| \frac{(\hat{\bm{\Phi}}_N - \bm{\Phi}_N)^\top\tilde{\mb{e}}_N}{\sqrt{N}} \right\|_F \\
        &\qquad = O_\bbP(1) {\c O}_{\t\bbP}\left(\sqrt{\sum_{k\le L_N}\zeta_1(k)^2\frac{\log^{2c}N}{N}}\right) O_\bbP(1) \left\| \frac{(\hat{\bm{\Phi}}_N - \bm{\Phi}_N)^\top\tilde{\mb{e}}_N}{\sqrt{N}} \right\|_F.
    \end{align*}
    We have
    \begin{align*}
        \E&\left[\left\| \frac{(\hat{\bm{\Phi}}_N - \bm{\Phi}_N)^\top\tilde{\mb{e}}_N}{\sqrt{N}} \right\|_F^2 \Big| \mb{X},\mb{A},\mb{U}\right]\\
        &=\frac{1}{N}\tr[(\hat{\bm{\Phi}}_N - \bm{\Phi}_N)(\hat{\bm{\Phi}}_N - \bm{\Phi}_N)^\top\E(\tilde{\mb{e}}_N{\tilde{\mb{e}}_N}^\top|\mb{X}, \mb{A},\mb{U})]
        = \frac{\|\hat{\bm{\Phi}}_N - \bm{\Phi}_N\|_F^2}{N}\tr(\Sigma_{\tilde E})
    \end{align*}
    where $\Sigma_{\tilde E}=\mb{V}$ which in turn verifies $\|J_1\|=o_\bbP(1)$. 
    We omit the rest of the proof as it works in a similar way.
    
    \item Using $\mb{M}_{\hat{\bm\Phi}_N}\hat{\bm\Phi}_N=(0)$ and $\mb{M}_{{\bm\Phi}_N}{\bm\Phi}_N=(0)$,
    \begin{align*}
        \frac{1}{\sqrt{N}} & \mb{K}_N^\top \mb{M}_{\hat{\bm\Phi}_N} (\mb{h}^E(\mb{U}_N) - \hat{\bm\Phi}_N \bm\alpha_{L_N}^{E}) \\
        &= \frac{1}{\sqrt{N}} \mb{K}_N^\top \left( \mb{M}_{\hat{\bm\Phi}_N}-\mb{M}_{{\bm\Phi}_N} \right) \mb{h}^E(\mb{U}_N) + \frac{1}{\sqrt{N}} \mb{K}_N^\top \mb{M}_{{\bm\Phi}_N} \left( \mb{h}^E(\mb{U}_N) - \bm\Phi_N \bm\alpha_{L_N}^{E} \right) \\
        &= I_1'+I_2', \text{ say}.
    \end{align*}
    $\|I_1'\|=o_\bbP(1)$ follows in the similar way as \ref{lem1:prf1}. We finish the proof by observing that
    \begin{align*}
        \| I_2' \| &\leq \frac{\| \mb K_N \|_F}{\sqrt{N}} \| \mb{h}^E(\mb{U}_N) - \bm\Phi_N \bm\alpha_{L_N}^{E} \|_F \\
        &\qquad= \frac{\| \mb{K}_N \|_F}{\sqrt{N}} \sqrt{\sum_{j\le m}\| \mb{h}_j^E(\mb{U}_N) - \bm\Phi_N\bm\alpha_{L_N}^{j} \|^2} \\
        &\qquad\le O_\bbP(1) \sup_{i\le N,j\le m}\sqrt{mN\{\mb{h}_j^E(\mb{U}_i) - \bm{\phi}^{L_N}(\mb{U}_i)^\top\bm\alpha_{L_N}^{j}\}^2}= O_\bbP({L_N}^{-\kappa}\sqrt{N}) = o_\bbP(1)
    \end{align*}
    by the Assumption \ref{assmp:sieve}-(iii).
\end{enumerate}
\end{proof}

\subsubsection{Proof of Lemma \ref{lem:sieveapprox}}

\begin{proof}
    We only need to show that
    \begin{enumerate}
        \item $\frac{1}{N}\sum_{i\le N}[\hat{\mb{h}}^K(\mb{U}_i)-\mb{h}^K(\mb{U}_i)][\hat{\mb{h}}^Z(\mb{U}_i)-\mb{h}^Z(\mb{U}_i)]^\top=o_\bbP(1)$
        \item $\frac{1}{N}\sum_{i\le N}[\hat{\mb{h}}^K(\mb{U}_i)-\mb{h}^K(\mb{U}_i)][\hat{\mb{h}}^K(\mb{U}_i)-\mb{h}^K(\mb{U}_i)]^\top=o_\bbP(1)$
        \item $\frac{1}{\sqrt{N}}\sum_{i\le N}[\hat{\mb{h}}^K(\mb{U}_i)-\mb{h}^K(\mb{U}_i)]{\tilde{\mb{e}}_i}^\top=o_\bbP(1)$.
    \end{enumerate}
    As 3 follows from 2, it suffices to address 1 and 2, and these follow by Lemma \ref{lmm:sieve:apprx}.
\end{proof}

\subsection{Additional Technical Lemmas}
First, we realize that
\[
\mb{Y}_v = (\mb{D}^\top \otimes \mb{G})\mb{Y}_v + (\mb{B}_1^\top \otimes \mb{I})\, \mb{X}_v + (\mb{B}_2^\top\otimes\mb{G})\, \mb{X}_v + \mb{e}_v
\]
which has a solution of the form
\begin{align*}
    \begin{bmatrix}
        \mb{Y}_1 \\ \vdots \\ \mb{Y}_m
    \end{bmatrix} = \left\{\mb{I}_{Nm} - (\mb{D}^\top \otimes \mb{G})\right\}^{-1} \left[\left\{(\mb{B}_1^\top \otimes \mb{I}_N) + (\mb{B}_2^\top \otimes \mb{G})\right\} \mb{X}_v + \mb{e}_v\right]
    \\ = \sum_{l=0}^\infty(\mb{D}^\top \otimes \mb{G})^l \left[\left\{(\mb{B}_1^\top \otimes \mb{I}_N) + (\mb{B}_2^\top \otimes \mb{G})\right\} \mb{X}_v + \mb{e}_v\right]
    \\ = \sum_{l=0}^\infty\{(\mb{D}^l)^\top \otimes \mb{G}^l\} \left[\left\{(\mb{B}_1^\top \otimes \mb{I}_N) + (\mb{B}_2^\top \otimes \mb{G})\right\} \mb{X}_v + \mb{e}_v\right]
\end{align*}
provided that $\rho(\mb{D}^\top\otimes \mb{G})<1$, i.e., $\rho(\mb{D})<1$.
\begin{Lemma}\label{lem:cov-cons}
    Assume that
    \[\max_{i_1,\dots,i_4\le N}\E\left(\frac{\|\mb{a}_{i_1}\|^2\|\mb{a}_{i_2}\|^2\|\mb{a}_{i_3}\|\|\mb{a}_{i_4}\|\|\mb{A}\|_{2,\infty}^2}{\delta_{i_1}^2\delta_{i_2}^2\dmin^4}\right)=o(1/N^2).\]
    Write $\Sigma_{\tilde K\tilde Z}:=\lim_{N\to\infty}\frac{1}{N}\E(\tilde{\mb{K}}^\top\tilde{\mb{Z}})$, $\Sigma_{\tilde Z\tilde K}:=\lim_{N\to\infty}\frac{1}{N}\E(\tilde{\mb{Z}}^\top\tilde{\mb{K}})$, and $\Sigma_{\tilde K}:=\lim_{N\to\infty}\frac{1}{N}\E(\tilde{\mb{K}}^\top\tilde{\mb{K}})$.
    Then, $\frac{1}{N}\tilde{\mb{K}}_N^\top\tilde{\mb{Z}}_N=\Sigma_{\tilde K\tilde Z}+o_\bbP(1)$, $\frac{1}{N}\tilde{\mb{Z}}_N^\top\tilde{\mb{K}}_N=\Sigma_{\tilde Z\tilde K}+o_\bbP(1)$, and $\frac{1}{N}\tilde{\mb{K}}_N^\top\tilde{\mb{K}}_N=\Sigma_{\tilde K}+o_\bbP(1)$.
\end{Lemma}

\begin{proof}
    We only show
    \begin{align*}
        \frac{1}{N}\tilde{\mb{K}}^\top\tilde{\mb{Z}}=\frac{1}{N}\begin{bmatrix}
        \mb{X}^\top\tilde{(\mb{G}\mb{Y})} & 
        \mb{X}^\top\mb{X} & \mb{X}^\top\tilde{(\mb{G}\mb{X})}\\
        \tilde{(\mb{G}\mb{X})}^\top\tilde{(\mb{G}\mb{Y})} & \tilde{(\mb{G}\mb{X})}^\top\mb{X} & \tilde{(\mb{G}\mb{X})}^\top\tilde{(\mb{G}\mb{X})} \\
        \tilde{(\mb{G}^2\mb{X})}^\top\tilde{(\mb{G}\mb{Y})} & \tilde{(\mb{G}^2\mb{X})}^\top\mb{X} & \tilde{(\mb{G}^2\mb{X})}^\top\tilde{(\mb{G}\mb{X})}
    \end{bmatrix}\\
    =\begin{bmatrix}
        \Sigma_{X,\tilde{(GY)}} & 
        \Sigma_{X} & \Sigma_{X,\tilde{(GX)}}\\
        \Sigma_{\tilde{(GX)},\tilde{(GY)}} & \Sigma_{\tilde{(GX)},X} & \Sigma_{\tilde{(GX)}} \\
        \Sigma_{\tilde{(G^2X)},\tilde{(GY)}} & \Sigma_{\tilde{(G^2X)},X} & \Sigma_{\tilde{(G^2X)},\tilde{(GX)}}
        \end{bmatrix}+o_\bbP(1)
    \end{align*}
    where $\tilde{(\mb{G}\mb{Y})}:=\mb{G}\mb{Y}-\E(\mb{G}\mb{Y}|\mb{U})$ (and the rest terms are defined analogously) as the convergence of $\frac{1}{N}\tilde{\mb{Z}}^\top\tilde{\mb{K}}$, $\frac{1}{N}\tilde{\mb{K}}^\top\tilde{\mb{K}}$ can be proved in the same way.
    The block entries are defined as $\Sigma_{X,\tilde{(GY)}}:=\frac{1}{N}\E[\mb{X}^\top\tilde{(\mb{G}\mb{Y})}]$, $\Sigma_{X}:=\E(\mb{X}_1\mb{X}_1^\top)$, and the rest are defined analogously.
    We only show that $\frac{1}{N}\tilde{(\mb{G}^2\mb{X})}^\top\tilde{(\mb{G}\mb{Y})}-\Sigma_{\tilde{(G^2X)},\tilde{(GY)}}=o_\bbP(1)$ as the rest can be shown in the same way. i.e.,
    \[\left|\frac{1}{N}\tilde{(\mb{G}^2\mb{X}_j)}^\top\tilde{(\mb{G}\mb{Y}_i)}-\sigma_{ji,\tilde{(G^2X)},\tilde{(GY)}}\right|=o_\bbP(1)\]
    for all $i\le m$ and $j\le p$.
    We only need to show that $\sigma_{ji,\tilde{(G^2X)},\tilde{(GY)}}=\lim_{N\to\infty}\frac{1}{N}\E[\tilde{(\mb{G}^2\mb{X}_j)}^\top\tilde{(\mb{G}\mb{Y}_i)}]$ exists and $\V\left[\frac{1}{N}\tilde{(\mb{G}^2\mb{X}_j)}^\top\tilde{(\mb{G}\mb{Y}_i)}\right]=o(1)$.
    Observe that $\tilde{(\mb{G}^2\mb{X})}=\mb{G}^2\mb{X}$ and
    \[\frac{1}{N}\E\left[\tilde{(\mb{G}^2\mb{X}_j)}^\top\tilde{(\mb{G}\mb{Y}_i)}\right]=\frac{1}{N}\sum_{k\le N}\E[(\mb{g}_k^\top\mb{G}\mb{X}_j)\tilde{(\mb{g}_k^\top\mb{Y}_i)}]\]
    as $(\mb{X}_i, \mb{U}_i, \mb{e}_i)$ are i.i.d. and $X_{ij}$ is uncorrelated with $\mb{U}_i$.
    Second, it suffices to show that
    \begin{align*}
        \frac{1}{N^2}\E&\left[(\mb{G}^2\mb{X}_j)^\top\tilde{(\mb{G}\mb{Y}_i)}\tilde{(\mb{G}\mb{Y}_i)}^\top\mb{G}^2\mb{X}_j\right]\\
        &=\frac{1}{N^2}\sum_{k,h\le N}\E\left[(\mb{g}_k^\top\mb{G}\mb{X}_j)(\mb{g}_h^\top\mb{G}\mb{X}_j)\tilde{(\mb{g}_k^\top\mb{Y}_i)}\tilde{(\mb{g}_h^\top\mb{Y}_i)}\}\right]=o(1).
    \end{align*}
    Using the block selector $\bm L_i:=\mb{e}_i^\top\otimes\mb{I}_N$, by $\bm L_i\mb{Y}_v=\mb{Y}_i$ we get
    \begin{align*}
        \mb{g}_k^\top\mb{Y}_i&=\sum_{l=0}^\infty\mb{g}_k^\top\{(\mb{D}^l\mb{e}_i)^\top \otimes \mb{G}^l\} \left[\left\{(\mb{B}_1^\top \otimes \mb{I}_N) + (\mb{B}_2^\top \otimes \mb{G})\right\} \mb{X}_v + \mb{e}_v\right]\\
        &=\sum_{l=0}^\infty\{(\mb{B}_1\mb{D}^l\mb{e}_i)^\top \otimes (\mb{g}_k^\top\mb{G}^l)+(\mb{B}_2\mb{D}^l\mb{e}_i)^\top \otimes (\mb{g}_k^\top\mb{G}^{l+1})\}\mb{X}_v\\
        &\qquad+\sum_{l=0}^\infty\{(\mb{D}^l\mb{e}_i)^\top \otimes \bm (\mb{g}_k^\top\mb{G}^l)\}\mb{e}_v\\
        &=\sum_{l=0}^\infty \sum_{o=1}^p \sum_{s=1}^N 
        \left[ 
        \left( \mb{e}_o^\top\mb{B}_1 \mb{D}^l \mb{e}_i \right) \left( \mb{e}_k^\top \mb{G}^{l+1} \mb{e}_s\right) 
        + 
        \left( \mb{e}_o^\top\mb{B}_2 \mb{D}^l \mb{e}_i \right) \left( \mb{e}_k^\top \mb{G}^{l+2} \mb{e}_s\right)
        \right] X_{so}\\
        &\qquad+\sum_{l=0}^\infty \sum_{o=1}^p \sum_{s=1}^N 
        \left( \mb{e}_o^\top\mb{D}^l \mb{e}_i \right) \left( \mb{e}_k^\top \mb{G}^{l+1} \mb{e}_s\right) E_{so}\\
        &=\sum_{o=1}^p \sum_{s=1}^N \underbrace{\sum_{l=0}^\infty
        \left(\mb{B}_1\mb{D}^l\mb{e}_i\mb{e}_k^\top\mb{G}^{l+1} + 
        \mb{B}_2\mb{D}^l\mb{e}_i\mb{e}_k^\top\mb{G}^{l+2}\right)_{os}}_{=:({\c B}_{ik})_{os}} X_{so}\\
        &\qquad+\sum_{o=1}^p \sum_{s=1}^N\underbrace{\sum_{l=0}^\infty 
        (\mb{D}^l\mb{e}_i\mb{e}_k^\top\mb{G}^{l+1})_{os}}_{=:({\c C}_{ik})_{os}} E_{so}=\sum_{o=1}^p \sum_{s=1}^N\{({\c B}_{ik})_{os}X_{so}+({\c C}_{ik})_{os}E_{so}\}
    \end{align*}
    where the penultimate line follows by $(\mb{B}^\top\otimes\mb{a}^\top)\operatorname{vec}(\mb{X})=\mb{a^\top Xb}=\sum_{s,o}a_sb_oX_{so}$. So,
    \begin{align*}
        \E(\mb{g}_k^\top\mb{Y}_i|\mb{U})&=\sum_{o=1}^p \sum_{s=1}^N\E\{({\c B}_{ik})_{os}X_{so}+({\c C}_{ik})_{os}E_{so}|\mb{U}\}=\sum_{o=1}^p \sum_{s=1}^N\E\{({\c C}_{ik})_{os}E_{so}|\mb{U}\}
    \end{align*}
    hence
    \begin{align*}
        \frac{1}{N}\sum_{k\le N}\E[&(\mb{g}_k^\top\mb{G}\mb{X}_j)\tilde{(\mb{g}_k^\top\mb{Y}_i)}]
        =\frac{1}{N}\sum_{k,o,s}\E\left[(\mb{g}_k^\top\mb{G}\mb{X}_j)\{({\c B}_{ik})_{os}X_{so}+({\c C}_{ik})_{os}E_{so}-\E(({\c C}_{ik})_{os}E_{so}|\mb{U})\}\right]\\
        &=\frac{1}{N}\sum_{k,o,s}\E\left[ ({\c B}_{ik})_{os}X_{so}\mb{X}_j^\top\mb{G}^\top\mb{g}_k\right]
        =\frac{1}{N}\sum_{k,o,s}\E(X_{so}\mb{X}_j^\top)\E[({\c B}_{ik})_{os}\mb{G}^\top\mb{g}_k]\\
        &=\frac{1}{N}\sum_{k,o,s}\sigma_{oj,X}\mb{e}_s^\top\E[({\c B}_{ik})_{os}\mb{G}^\top\mb{g}_k]=\frac{1}{N}\sum_{k,o,s}\sigma_{oj,X}\E[({\c B}_{ik})_{os}\mb{G}_s^\top\mb{g}_k]
    \end{align*}
    as $\mb{X}$ and $\mb{U}$ are uncorrelated.
    We have $\sum_{s,k}({\c B}_{ik})_{os}(\mb{G}^\top\mb{G}^\top)_{sk}=\sum_k({\c B}_{ik}\mb{G}^\top\mb{G}^\top)_{ok}=(\bm 1\otimes\mb{e}_o)^\top\operatorname{vec}({\c B}_{ik}\mb{G}^\top\mb{G}^\top)$ and
    \begin{align*}
        {\c B}_{ik}(\mb{G}^\top)^2&=\sum_{l=0}^\infty(\mb{B}_1\mb{D}^{l}\mb{e}_i\mb{e}_k^\top\mb{G}^{l+1}+\mb{B}_2\mb{D}^{l}\mb{e}_i\mb{e}_k^\top\mb{G}^{l+2})(\mb{G}^\top)^2\\
        &=\sum_{l=0}^\infty\mb{B}_1(\mb{D}^{l}\mb{e}_i\mb{e}_k^\top\mb{G}^{l})\mb{G}(\mb{G}^\top)^2+\sum_{l=0}^\infty\mb{B}_2(\mb{D}^{l}\mb{e}_i\mb{e}_k^\top\mb{G}^{l})\mb{G}^2(\mb{G}^\top)^2,\\
        \operatorname{vec}({\c B}_{ik}\mb{G}^\top\mb{G}^\top)&=\{((\mb{G}^2\mb{G}^\top)\otimes\mb{B}_1)+((\mb{G}^2\mb{G}^\top\mb{G}^\top)\otimes\mb{B}_2)\}\sum_{l}((\mb{G}^\top)\otimes\mb{D})^l\operatorname{vec}(\mb{e}_i\mb{e}_k^\top)\\
        &=\{((\mb{G}^2\mb{G}^\top)\otimes\mb{B}_1)+((\mb{G}^2\mb{G}^\top\mb{G}^\top)\otimes\mb{B}_2)\}(\mb{I}_{Nm}-(\mb{G}^\top)\otimes\mb{D})^{-1}(\mb{e}_k\otimes\mb{e}_i).
    \end{align*}
    As $\rho(\mb{G}),\rho(\mb{D})<1$, and by the dominated convergence theorem, $\sigma_{ji,\tilde{(G^2X)},\tilde{(GY)}}$ exists.

    Next, writing
    \begin{align*}
        \eta_{os,k}&:=({\c C}_{ik})_{os}E_{so}-\E(({\c C}_{ik})_{os}E_{so}|\mb{U})\\
        &=\sum_{l=0}^\infty \{(\mb{D}^l\mb{e}_i\mb{e}_k^\top\mb{G}^{l+1})_{os} E_{so}-\E((\mb{D}^l\mb{e}_i\mb{e}_k^\top\mb{G}^{l+1})_{os} E_{so}|\mb{U})\},
    \end{align*}
    we have
    \begin{align*}
        &\frac{1}{N^2}\sum_{k,h\le N}\E\left[(\mb{g}_k^\top\mb{G}\mb{X}_j)(\mb{g}_h^\top\mb{G}\mb{X}_j)\tilde{(\mb{g}_k^\top\mb{Y}_i)}\tilde{(\mb{g}_h^\top\mb{Y}_i)}\}\right]\\
        &=\frac{1}{N^2}\sum_{k,h,s_1,s_2,o_1,o_2}\E\left[ (\mb{g}_k^\top\mb{G}\mb{X}_j)(\mb{g}_h^\top\mb{G}\mb{X}_j)\{({\c B}_{ik})_{o_1s_1}X_{s_1o_1}+\eta_{o_1s_1,k}\}\{({\c B}_{ih})_{o_2s_2}X_{s_2o_2}+\eta_{o_2s_2,h}\}\right].
    \end{align*}
    Let $\mb{g}_k^\top\mb{G}\mb{X}_j = \sum_{m=1}^N (\mb{g}_k^\top\mb{G})_m X_{mj}$. Note that
    \begin{align*}
        &\E\left[ ({\c B}_{ik})_{o_1s_1}({\c B}_{ih})_{o_2s_2}X_{s_1o_1}X_{s_2o_2}(\mb{g}_k^\top\mb{G}\mb{X}_j)(\mb{g}_h^\top\mb{G}\mb{X}_j) \mid \mb{A}, \mb{U} \right] \\
        &= \sum_{m_1, m_2=1}^N ({\c B}_{ik})_{o_1s_1}({\c B}_{ih})_{o_2s_2} (\mb{g}_k^\top\mb{G})_{m_1} (\mb{g}_h^\top\mb{G})_{m_2} \E_X\left[ X_{s_1o_1}X_{s_2o_2} X_{m_1j}X_{m_2j} \right]
    \end{align*}
    so the last expectation term is non-zero only when the indices are matched in pairs. 
    The dominant non-zero leading terms occur when $m_1 = m_2 = m$, $s_1 = s_2 = s$, and $o_1 = o_2 = o$. 
    For this leading case, $\E\left[ X_{so}^2 X_{mj}^2 \right] \le \max_{i,j}\E[X_{ij}^4]$. Thus, the following summation simplifies:
    \begin{align*}
        \frac{1}{N^2}&\sum_{k,h,s_1,s_2,o_1,o_2}\E \left[ ({\c B}_{ik})_{o_1s_1}({\c B}_{ih})_{o_2s_2}X_{s_1o_1}X_{s_2o_2}(\mb{g}_k^\top\mb{G}\mb{X}_j)(\mb{g}_h^\top\mb{G}\mb{X}_j) \right] \\
        &\lesssim \frac{1}{N^2}\sum_{k,h,m}\E\left[ (\mb{G}^2)_{km} (\mb{G}^2)_{hm}\sum_{o,s}({\c B}_{ik})_{os}({\c B}_{ih})_{os}\right].
    \end{align*}
    Recalling that
    $$({\c B}_{ik})_{os} = \sum_{l=0}^\infty \left[ (\mb{B}_1 \mb{D}^l)_{oi} (\mb{g}_k^\top \mb{G}^l)_s + (\mb{B}_2 \mb{D}^l)_{oi} (\mb{g}_k^\top \mb{G}^{l+1})_s\right]$$
    and $\rho(\mb{D})<1$, there exists $\lambda\in(0,1)$ such that $\|\mb{D}^l\|_{\infty}\lesssim\lambda^l$.
    Since the covariate parameter matrices $\mb{B}_1, \mb{B}_2$ have finite bounded norms, we can couple the geometric decay $\lambda^l$ directly with the path lengths:
    \[ |({\c B}_{ik})_{os}| \lesssim\sum_{l=0}^\infty \lambda^l (\mb{G}^{l+1})_{ks}.\]
    Returning to the cross-product of the coefficients, we can bound the summation over $o$ by expanding the series:
    \begin{align*}
        \sum_{o} ({\c B}_{ik})_{os}({\c B}_{ih})_{os} &\lesssim \sum_{l_1=0}^\infty \sum_{l_2=0}^\infty \lambda^{l_1+l_2} (\mb{G}^{l_1+1})_{ks} (\mb{G}^{l_2+1})_{hs}.
    \end{align*}
    Then,
    \begin{align*}
        \frac{1}{N^2}&\sum_{k,h,m}\E\left[ (\mb{G}^2)_{km} (\mb{G}^2)_{hm}\sum_{o,s}({\c B}_{ik})_{os}({\c B}_{ih})_{os}\right] \\
        &\lesssim\frac{1}{N^2}\sum_{l_1=0}^\infty \sum_{l_2=0}^\infty \lambda^{l_1+l_2}\underbrace{\sum_{k,h,m,s} \E \left[ (\mb{G}^2)_{km} (\mb{G}^2)_{hm} (\mb{G}^{l_1+1})_{ks} (\mb{G}^{l_2+1})_{hs} \right]}_{=: W(l_1, l_2)}.
    \end{align*}
    We directly use the row-stochastic property of $\mb{G}$. By definition, the row sums of $\mb{G}$ are one and $G_{ij}\le1/\dmin$. It follows that for any $l\ge2$,
    \[(\mb{G}^l)_{ij}=\sum_q(\mb{G}^{l-1})_{iq}G_{qj}\le\frac{1}{\dmin}\sum_q(\mb{G}^{l-1})_{iq}=\frac{1}{\dmin}.\]
    Therefore,
    \begin{align*}
        W(l_1,l_2)=\E\left[\sum_{k,h,s}(\mb{G}^{l_1+1})_{ks} (\mb{G}^{l_2+1})_{hs}\left\{\sum_m(\mb{G}^2)_{km} (\mb{G}^2)_{hm}\right\}\right] \\
        \le\E\left[\sum_{k,h}\frac{1}{\dmin}\left\{\frac{1}{\dmin}\sum_m(\mb{G}^2)_{km}\right\}\right]=\frac{N^2}{\dmin^2}
    \end{align*}
    hence
    \[\frac{1}{N^2}\sum_{k,h,m}\E\left[ (\mb{G}^2)_{km} (\mb{G}^2)_{hm}\sum_{o,s}({\c B}_{ik})_{os}({\c B}_{ih})_{os}\right]\lesssim\E\left[\frac{1}{(1-\lambda)^2\dmin^2}\right]=o(1).\]

    Next, noting that $\E[({\c C}_{ik})_{os}E_{so}|\mb{U}]=({\c C}_{ik})_{os}\E[E_{so}|\mb{U}]$ and
    \[\sum_{l=0}^\infty\sum_{k,s,o}(\mb{D}^l\mb{e}_i\mb{e}_k^\top\mb{G}^{l+1})_{os}\tilde E_{so}\le\sum_l\sum_{k,s,o}C\mb{g}_k^\top\mb{G}^{l-1}\mb{G}_{s}\tilde E_{so},\]
    we have
    \begin{align*}
        \frac{1}{N^2}&\sum_{k,h,s_1,s_2,o_1,o_2}\E(\eta_{o_1s_1,k}\eta_{o_2s_2,h})=\frac{1}{N^2}\sum_{k,h,s_1,s_2,o_1,o_2}\E[({\c C}_{ik})_{o_1s_1}({\c C}_{ih})_{o_2s_2}\tilde E_{s_1o_1}\tilde E_{s_2o_2}]\\
        &=\frac{1}{N^2}\sum_{k,h,s_1,s_2,o_1,o_2}\E(\tilde E_{s_1o_1}\tilde E_{s_2o_2})\E[({\c C}_{ik})_{o_1s_1}({\c C}_{ih})_{o_2s_1}]=o(1)
    \end{align*}
    by similar arguments.
    This completes the proof.
\end{proof}

\begin{Lemma}\label{lmm:sieve:est2}
    Let Assumption \ref{assmp:lipsch} hold.
    Under the conditions on RDPG model in \cite{rubin2022statistical} and sparsity parameter $\rho_N=\omega(\frac{\log^4N}{N})$, we have
    \[\frac{1}{N}\|\hat{\bm\Phi}_N-\bm\Phi_N\|_F^2={\c O}_{\t\bbP}\left(\frac{\log^{2c}N}{N}\sum_{k\le L_N}\zeta_1(k)^2\right)\]
    and under the conditions on the sparse LSM model with sparsity parameter $\omega_N=\omega(N^{-1/2})$ in \cite{li2023statistical}, we have
    \[\frac{1}{N}\|\hat{\bm\Phi}_N-\bm\Phi_N\|_F^2={\c O}_{\t\bbP}\left(\frac{1}{N\omega_N}\sum_{k\le L_N}\zeta_1(k)^2\right).\]
\end{Lemma}
\begin{proof}
    \begin{align*}
        \frac{1}{N}\|\hat{\bm\Phi}_N-\bm\Phi_N\|_F^2=\frac{1}{N}\sum_{i\le N,k\le L_N}|\phi_k(\hat {\mb{U}}_i)-\phi_k(\mb{h}^\top\mb{U}_i)|^2\le\frac{1}{N}\sum_{i\le N,k\le L_N}\zeta_1(k)^2\|\hat {\mb{U}}_i-\mb{h}^\top\mb{U}_i\|^2
    \end{align*}
    By the results on Adjacency Spectral Embedding in RDPG model in \cite{rubin2022statistical}, the last term is less than
    \[\inf_{\mb H\in {\cal O}^K}\|\hat {\mb{U}}-\mb{U}\mb{H}\|_{2,\infty}^2\sum_{k\le L_N}\zeta_1(k)^2={\c O}_{\t\bbP}\left(\frac{\log^{2c}N}{N}\sum_{k\le L_N}\zeta_1(k)^2\right)\]
    if $\mb{A}$ follows the RDPG. If we instead assume that $\mb{A}$ is generated from LSM with sparsity factor $\omega_n$, the last term is bounded by
    \[\frac{1}{N}\|\hat {\mb{U}}-\mb{U}\|_F^2\sum_{k\le L_N}\zeta_1(k)^2=O_\bbP\left(\frac{1}{N\omega_n}\sum_{k\le L_N}\zeta_1(k)^2\right)\]
    by the result of the concentration of the MLE in LSM \cite{li2023statistical}.
\end{proof}

Clearly, Lemma \ref{lmm:sieve:est2} provides a result on the concentration of the sieve design matrix when the true latent positions $\mb{U}$ are replaced with their estimated counterparts $\hat{\mb{U}}$. Both the RDPG and the LSM models can accommodate sparse networks through the sparsity parameters $\rho_N$ and $\omega_N$, respectively. However, our result with the RDPG model has a better concentration of the $\hat{\bm\Phi}_N$ to $\bm\Phi_N$, for whenever $\omega_N$ is smaller than $\frac{1}{\log^{2c}N}$, i.e., if the network is even slightly sparse. However, we emphasize that our results for both models accommodate sparse networks with density requirements for RDPG being $\frac{\log^4N}{N}$ and LSM being $\frac{N^{1/2}}{N}$.

\begin{Lemma}\label{lmm:sieve:rank}
    Given Assumptions \ref{assmp:sieve}-(ii),(iii), and \ref{assmp:lipsch},
    \[\underset{N\to\infty}{\operatorname{plim}}\lambda_{\min}\left(\frac{\bm\Phi_N^\prime \bm\Phi_N}{N}\right)=\underset{N\to\infty}{\operatorname{plim}}\lambda_{\min}\left(\frac{\hat{\bm\Phi}_N^\prime \hat{\bm\Phi}_N}{N}\right)=\lim_{N\to\infty}\lambda_{\min}\left(\E[\bm{\phi}^{L_N}(\mb{U}_1)\bm{\phi}^{L_N}(\mb{U}_1)^\top]\right).\]
\end{Lemma}
\begin{proof}
    We first show
    \[
    \left| \lambda_{\min}\left(\frac{\bm\Phi_N^\prime \bm\Phi_N}{N}\right) - \lambda_{\min}\left(\E[\bm{\phi}^{L_N}(\mb{U}_1)\bm{\phi}^{L_N}(\mb{U}_1)^\prime]\right) \right| = o_\bbP(1).
    \]
    Temporarily write $\bm W^{(i)}:=\bm{\phi}^{L_N}(\mb{U}_i)\bm{\phi}^{L_N}(\mb{U}_i)^\prime$.
    By Weyl's inequality, we have
    \begin{align*}
        &\left| \lambda_{\min}\left(\frac{\bm\Phi_N^\prime \bm\Phi_N}{N}\right) - \lambda_{\min}\left(\E\bm W^{(1)}\right) \right| \leq \left\lVert \frac{\bm\Phi_N^\prime \bm\Phi_N}{N} - \E(\bm W^{(1)}) \right\rVert_F \\
        &= \left\lVert \frac{1}{N} \sum_{i=1}^N \left\{\bm W^{(i)} - \E(\bm W^{(i)})\right\} \right\rVert_F.
    \end{align*}
    Then for $\tilde{\bm W}^{(i)}:=\bm W^{(i)} - \E(\bm W^{(i)})=(\tilde w_{kl}^{(i)})_{k,l\le L_N}$, by assumption, we have
    \begin{align*}
        \E&\left\lVert \frac{1}{N} \sum_{i=1}^N \tilde{\bm W}^{(i)} \right\rVert_F^2 = \E\sum_{1\le k,l\le L_N}\left( \frac{1}{N} \sum_{i=1}^N \tilde w_{kl}^{(i)} \right)^2 = \sum_{1\le k,l\le L_N}\frac{N}{N^2}\E(\tilde w_{kl}^{(1)2}) \\
        & \leq \frac{1}{N}\sum_{1\le k,l\le L_N}\E(w_{kl}^{(1)2}) = \frac{1}{N}\sum_{1\le k,l\le L_N}\E[\phi_k(\mb{U}_1)^2\phi_l(\mb{U}_1)^2]=\frac1N\E[\|\bm{\phi}^{L_N}(\mb{U}_1)\|^4] \\
        &\le \frac1N\E[\|\bm\phi^{L_N}(\mb{u})\|^2]\sup_{\mb{u}\in{\c U}}\|\bm\phi^{L_N}(\mb{u})\|^2={\c O}\Bigl(\frac{\zeta_0(L_N)^2}{N}L_N\Bigr).
    \end{align*}
    Next,
    \begin{align*}
        &\left| \lambda_{\min}\left( \frac{\hat{\bm\Phi}_N^\prime \hat{\bm\Phi}_N}{N} \right) - \lambda_{\min}\left( \frac{\bm\Phi_N^\prime \bm\Phi_N}{N} \right) \right|
        \leq \left\lVert \frac{\hat{\bm\Phi}_N^\prime \hat{\bm\Phi}_N}{N} - \frac{\bm\Phi_N^\prime \bm\Phi_N}{N} \right\rVert_F \\
        &\leq \left\lVert \frac{\hat{\bm\Phi}_N - \bm\Phi_N}{\sqrt{N}} \right\rVert_F^2
        + \left\lVert \frac{\bm\Phi_N^\prime (\hat{\bm\Phi}_N - \bm\Phi_N)}{N}\right\rVert_F + \left\lVert\frac{(\hat{\bm\Phi}_N - \bm\Phi_N)^\prime \bm\Phi_N}{N} \right\rVert_F.
    \end{align*}
    Since $\|\bm\Phi_N\|_F^2=\sum_{i\le N}\|\bm{\phi}^{L_N}(\mb{U}_i)\|^2\le N\zeta_0(L_N)^2$,
    \[\left\lVert\frac{(\hat{\bm\Phi}_N - \bm\Phi_N)^\prime \bm\Phi_N}{N} \right\rVert_F={\c O}_{\t\bbP}\left(\sqrt{\zeta_0(L_N)^2\sum_{k\le L_N}\zeta_1(k)^2\frac{\log^{2c}N}{N}}\right)\] 
    and by the assumptions, we have the conclusion.
\end{proof}

In the following lemma, we write
\[\mb Q_{L_N}:=\E[\bm{\phi}^{L_N}(\mb{U}_1)\bm{\phi}^{L_N}(\mb{U}_1)^\top].\]
\begin{Lemma}\label{lmm:linracine}
    Under Assumption \ref{assmp:sieve} in section \ref{sec:assmp}, we have $\lambda_{\min}(\frac{1}{N}\bm\Phi_N^\top\bm\Phi_N)=\lambda_{\min}(\mb Q_{L_N})+o_\bbP(1)$.
\end{Lemma}
\begin{proof}
    Denote the $j$-th column of $\bm\Phi_N$ by $\bm\Phi_j(\mb{U})$.
    Let $\hat{\mb Q}_{L_N}:=\frac{1}{N}\bm\Phi_N^\top\bm\Phi_N$. 
    Let $q_{ij}:=\E[\phi_i(\mb{U}_1)^\top\phi_j(\mb{U}_1)]$. Then, for $z_{ij}(\mb{U}_k):=\phi_i(\mb{U}_k)\phi_j(\mb{U}_k)-q_{ij}$ which has zero mean,
    \begin{align*}
        \E\|\hat{\mb Q}_{L_N}&-\mb Q_{L_N}\|_F^2
        =\sum_{1\le i,j\le L_N}\E\left[\left(\frac{1}{N}\bm\Phi_i(\mb{U})^\top\bm\Phi_j(\mb{U})-q_{ij}\right)^2\right]\\
        &=\sum_{1\le i,j\le L_N}\E\left[\left(\frac{1}{N}\sum_{k\le N}z_{ij}(\mb{U}_k)\right)^2\right]=\frac{1}{N}\sum_{1\le i,j\le L_N}\E\left[z_{ij}(\mb{U}_1)^2\right]\\
        &\qquad\le\frac{1}{N}\sum_{1\le i,j\le L_N}\E\left[\phi_i(\mb{U}_1)^2\phi_j(\mb{U}_1)^2\right]=\frac{1}{N}\E[\|\bm{\phi}^{L_N}(\mb{U}_1)\|^4] \\
        &\qquad\qquad\le\frac{\zeta_0(L_N)^2}{N}\tr(\mb Q_{L_N})={\c O}\left(\frac{\zeta_0(L_N)^2L_N}{N}\right).
    \end{align*}
    So, the conclusion follows as $\|\hat{\mb Q}_{L_N}-\mb Q_{L_N}\|_F=o_\bbP(1)$.
\end{proof}
\begin{Lemma}\label{lmm:sieve:apprx}
    Let $\hat{\bm\alpha}_j^f:=\left(\bm\Phi_N^\top\bm\Phi_N\right)^{-1}\bm\Phi_N^\top\mb{h}_j^f(\mb{U}_N)$ for some $j\in{\c I}(f)$ and $\hat{\mb{h}}_j^f:=\hat{\mb{h}}_j^f(\mb{U}_N)=\bm\Phi_N\hat{\bm\alpha}_j^f$. 
    Then, Under assumption on sieve approximation in section \ref{sec:assmp},
    \[\|\hat{\bm\alpha}_j^f-\bm\alpha^j\|=O_\bbP({L_N}^{-\kappa}),\quad \frac{1}{N}\|\mb{h}_j^f(\mb{U})-\hat{\mb{h}}^f_j(\mb{U})\|^2=O_\bbP({L_N}^{-2\kappa}).\]
\end{Lemma}
\begin{proof}
    Let $B_N:=\{\lambda_{\min}(\frac{1}{N}\bm\Phi_N^\top\bm\Phi_N)\ge1/2\}$. 
    Then, $B_N=\{1/\lambda_{\min}(\frac{1}{N}\bm\Phi_N^\top\bm\Phi_N)\le2\}$ implies
    \begin{align*}
        \|\hat{\bm\alpha}_j^f-\bm\alpha^j\|=\left\|\left(\frac{1}{N}\bm\Phi_N^\top\bm\Phi_N\right)^{-1}\frac{1}{N}\bm\Phi_N^\top\{\mb{h}_j^f(\mb{U})-\bm\Phi_N\bm\alpha^j\}\right\|\le2\left\|\frac{1}{N}\bm\Phi_N^\top\{\mb{h}_j^f(\mb{U})-\bm\Phi_N\bm\alpha^j\}\right\|\\
        \le\frac{2\sqrt{2}}{\sqrt{N}}\|\mb{h}_j^f(\mb{U})-\bm\Phi_N\bm\alpha^j\|\lesssim\|\mb{h}_j^f(\mb{U})-\bm\Phi_N\bm\alpha^j\|_\infty={\c O}({L_N}^{-\kappa}).
    \end{align*}
    As $\Pr(B_N)\to1$,
    \[\|\hat{\bm\alpha}_j^f-\bm\alpha^j\|=O_\bbP({L_N}^{-\kappa}).\]
    Next,
    \begin{align*}
        \frac{1}{N}\|\mb{h}_j^f(\mb{U})-\hat{\mb{h}}_j^f(\mb{U})\|^2\le\frac{1}{N}\|\mb{h}_j^f(\mb{U})-\bm\Phi_N\bm\alpha^j\|^2+\frac{1}{N}\|\bm\Phi_N(\bm\alpha^j-\hat{\bm\alpha}_j^f)\|^2 \\
        ={\c O}({L_N}^{-2\kappa})+\frac{1}{N}\|\bm\Phi_N^\top\bm\Phi_N\|O_\bbP({L_N}^{-2\kappa})=O_\bbP({L_N}^{-2\kappa}).
    \end{align*}
\end{proof}

\setcounter{figure}{0}
\renewcommand{\thefigure}{B\arabic{figure}}
\setcounter{table}{0}
\renewcommand{\thetable}{B\arabic{table}}

\subsection{Algorithm}
\begin{algorithm}[H]
\caption{Vectorized 2SLS with Network Endogeneity in MSAR}
\label{algorithm1}
\begin{algorithmic}[1]
    \REQUIRE $\mb Y \in \mathbb{R}^{N\times m}$ (Outcomes),  $\mb X \in \mathbb{R}^{N\times p}$ (Covariates),  $\mb A \in \mathbb{R}^{N\times N}$ (Adjacency matrix), latent dimension $d$
    \ENSURE Vectorized coefficient matrix $\hat{\bm\beta}_v$, Variance-Covariance matrix of coefficients $\hat{\Sigma}_{\beta}$.
    \STATE \textbf{Step 1: Spectral Embedding (RDPG)}
    \STATE $\hat{\mb U} \gets \texttt{SpectralEmbed}(\mb A,d)$ \COMMENT{Top $d$ eigenvectors and absolute eigenvalues}
    
    \STATE \textbf{Step 2: Construct Peer Lag Matrices}
    \STATE $\mb G\mb Y \gets \mb G\mb Y$, $\mb G\mb X \gets \mb G\mb X$, $\mb G^2\mb X \gets \mb G(\mb G\mb X)$
    \STATE $\mb Z \gets [\mb G\mb Y,  \mb X,  \mb G\mb X]$, $\mb K \gets [\mb X , \mb G\mb X,  \mb G^2\mb X]$
    
    \STATE \textbf{Step 3: Latent Adjustment}
    \STATE Form sieve design matrix $\hat{\bm\Phi}_N = \left[\bm\phi^{L_N}(\hat{\mb u}_1), \ldots, \bm\phi^{L_N}(\hat{\mb u}_N) \right] \in \mathbb{R}^{N \times L_N}.$
    \STATE Compute sieve projection $\mb P{\hat{\bm\Phi}_N} := \hat{\bm\Phi}_N (\hat{\bm\Phi}_N^\top \hat{\bm\Phi}_N)^{-1} \hat{\bm\Phi}_N^\top$
    \STATE $\tilde{\mb Y}^* \gets \mb Y - \mb P{\hat{\bm\Phi}_N} \mb Y$, $\tilde{\mb Z}^* \gets \mb Z - \mb P{\hat{\bm\Phi}_N} \mb Z$, $\tilde{\mb K}^* \gets \mb K - \mb P{\hat{\bm\Phi}_N}\mb K$
    
    \STATE \textbf{Step 4: Vectorization}
    \STATE $\tilde{\mb Y}_v^* \gets \operatorname{vec}(\tilde{\mb Y}^*)$ \COMMENT{$\in \mathbb{R}^{Nm}$}
    \STATE $\tilde{\mb Z}_v^* \gets \mb I_m \otimes \tilde{\mb Z}^*$ \COMMENT{$\in \mathbb{R}^{mN\times m(m+2p)}$}
    \STATE $\tilde{\mb K}_v^* \gets \mb I_m \otimes \tilde{\mb K}^*$ \COMMENT{$\in \mathbb{R}^{mN\times 3mp}$}
    
    \STATE \textbf{Step 5: IV Estimation and Covariance Calculation}
    \STATE $\mb P{\tilde{\mb K}_v^*} \gets \tilde{\mb K}_v^* (\tilde{\mb K}_v^*{}^\top \tilde{\mb K}_v^*)^{-1}\tilde{\mb K}_v^*{}^\top$ \COMMENT{$\in \mathbb{R}^{mN\times mN}$}
    \STATE $\hat{\bm\beta}_v \gets (\tilde{\mb Z}_v^*{}^\top \mb P{\tilde{\mb K}_v^*}\tilde{\mb Z}_v^*)^{-1}\tilde{\mb Z}_v^*{}^\top \mb P{\tilde{\mb K}_v^*}\tilde{\mb Y}_v^*$ \COMMENT{$\in \mathbb{R}^{m(m+2p)}$}
    \STATE $\hat{\mb E}_v' \gets \tilde{\mb Y}_v^* - \tilde{\mb Z}_v^*\hat{\bm\beta}_v$ \COMMENT{$\in \mathbb{R}^{Nm}$}
    \STATE $\hat{\mb E}' \gets \texttt{reshape}(\hat{\mb E}'_v,N,m)$ \COMMENT{Reshape vector to an $N \times m$ matrix}
    \STATE $\hat{\mb V}' \gets \frac{1}{N}{\hat{\mb E}'^\top} \hat{\mb E}'$
    \STATE $\hat{\Sigma}_\beta \gets (\tilde{\mb Z}_v^*{}^\top \mb P{\tilde{\mb K}_v^*{}}\tilde{\mb Z}_v^*{})^{-1}\tilde{\mb Z}_v^*{}^\top \mb P{\tilde{\mb K}_v^*}\Bigl(  \hat{\mb V}' \otimes \mb I_N\Bigr)\mb P{\tilde{\mb K}_v^*}\tilde{\mb Z}_v^*(\tilde{\mb Z}_v^*{}^\top \mb P{\tilde{\mb K}_v^*}\tilde{\mb Z}_v^*)^{-1}$
\end{algorithmic}
\end{algorithm}

\clearpage

\section{Proofs for prediction-powered peer-effect inference}
\label{app:ppi-proofs}

\subsection{Proof of Proposition~\ref{prop:moment-unbiased}}

Taking expectations in \eqref{eq:ppi-moment-rectified} over the resident
sampling design and using $\E_R(R_i/\pi_i)=1$ gives
$\E_R\{\widehat{\mb g}_N^{\PP}(\bm\beta)\}
=\mb g_N^{\AI}(\bm\beta)+N^{-1}\sum_i\Psi_i(\bm\beta)=\mb g_N^{\Human}(\bm\beta)$.
For the second statement, \eqref{eq:Yppi} gives
$\widehat{\mb Y}^{\PP}=\mb Y^{\AI}+\sum_{i}(R_i/\pi_i)\mb e_i\Delta_i^\top$. Substituting
into \eqref{eq:generic-moment} and using linearity of the residual in $\mb Y$ at
fixed $\bm\beta$ yields
$\mb g_N(\bm\beta;\widehat{\mb Y}^{\PP})
=\mb g_N^{\AI}(\bm\beta)+N^{-1}\sum_i(R_i/\pi_i)\mb K^\top \mb M_{\widehat{\mb \Phi}_N}
\Lambda_i(\mb D)=\widehat{\mb g}_N^{\PP}(\bm\beta)$. \qed

\subsection{Proof of Theorem~\ref{thm:ppi-asymptotic}}
We write  $a_i=R_i/f_N-1$, so that $\E_R(a_i\mid\mathcal F_N)=0$ and write $N^{-1}\sum_i a_i\mb v_i=\bar{\mb  v}_S-\bar{\mb v}_N$ for any fixed finite population $\{\mb v_i\}$, where $\bar{\mb  v}_S$ is the mean over the validated residents. SRSWOR variance is in the form of
\begin{equation}
    \Cov_R\!\left(\frac{1}{\sqrt N}\sum_{i=1}^N a_i \mb v_i
    \,\middle|\,\mathcal F_N\right)
    =\frac{1-f_N}{f_N}\mb S_{v,N},
    \label{eq:ppi-proof-srs-cov}
\end{equation}
where
$\mb S_{v,N}=(N-1)^{-1}\sum_i(\mb v_i-\bar{\mb  v}_N)(\mb v_i-\bar{\mb  v}_N)^\top$ and
$\bar{\mb  v}_N=N^{-1}\sum_i\mb v_i$. Thus, by conditional Chebyshev's inequality, if
$N^{-1}\sum_i\|\mb v_i\|^2=O_p(1)$, then
$N^{-1}\sum_i a_i\mb v_i=O_p(N^{-1/2})$.

Let $\mb W_{R,N}=\sum_{i=1}^N a_i \mb e_i\Delta_i^\top$, then by the definition of the PPPI pseudo-outcome,
\begin{equation}
    \widehat{\mb Y}^{\PP}=\mb Y^{\Human}+\mb W_{R,N},
    \qquad
    \widehat{\mb Z}^{\PP}
    =\mb Z^{\Human}+[\mb G\mb W_{R,N},0,0].
    \label{eq:ppi-proof-data-perturbation}
\end{equation}
Conditional on $\mathcal F_N$, the complete data, discrepancy vectors, and both
sieve projections are fixed; only the simple random sample of validated
residents is random. In particular, $\E_R(a_i\mid\mathcal F_N)=0$.
The final asymptotic distribution is under the joint law of the oracle data
and resident sampling. The truncation bounds $\max_i\|\Delta_i\|_2$ by a
constant $c_\Delta$. Also, for the undirected row-normalized network,
\[
\|\mb G\|^2\le \|\mb G\|_1\|\mb G\|_\infty
\le c_{G,N}:=\delta_{\max}/\delta_{\min}=O_p(1)
\]
by assumptions.

\paragraph{Consistency.}
Since $\|\mb M_{\widehat{\mb \Phi}_N}\|\le1$ and
$N^{-1}\|K\|_F^2=O_p(1)$ by
Assumption~\ref{ass:validation-regularity},
\begin{align*}
    \frac1N\sum_i\|\mb K^\top \mb M_{\widehat{\mb \Phi}_N}\mb e_i\Delta_i^\top\|_F^2
    &\le\frac{c_\Delta^2}{N}\|\mb M_{\widehat{\mb \Phi}_N}\mb K\|_F^2=O_p(1),\\
    \frac1N\sum_i\|\mb K^\top \mb M_{\widehat{\mb \Phi}_N}\mb G\mb e_i\Delta_i^\top\|_F^2
    &\le\frac{c_\Delta^2}{N}\|\mb G^\top \mb M_{\widehat{\mb \Phi}_N}\mb K\|_F^2
    \le\frac{c_\Delta^2\|\mb G\|^2}{N}\|\mb K\|_F^2=O_p(1).
\end{align*}
Consequently,
\begin{equation}
\frac1N \mb K^\top \mb M_{\widehat{\mb \Phi}_N}\mb W_{R,N}=O_p(N^{-1/2}),
\qquad
\frac1N \mb K^\top \mb M_{\widehat{\mb \Phi}_N}\mb G\mb W_{R,N}=O_p(N^{-1/2}).
\label{eq:ppi-proof-cross-moment-errors}
\end{equation}
Hence the normalized PPPI instrument--outcome and instrument--regressor
cross-moments differ from their complete-human counterparts by
$O_p(N^{-1/2})$. The oracle cross-moment results used in the proof of
Theorem~\ref{thm:clt} give their population limits, and
Assumption~\ref{ass:iv-relevance} keeps the limiting 2SLS matrices nonsingular.
The normalized human instrument--outcome cross-moment is $O_p(1)$, by the
bounds on $\mb K$ and $\mb Y^{\Human}$.

From \eqref{eq:ppi-moment-rectified},
\begin{equation}
\widehat{\mb g}_N^{\PP}(\bm\beta)-\mb g_N^{\Human}(\bm\beta)
=\frac{1}{N}\sum_{i=1}^N a_i\Psi_i(\bm\beta).
\label{eq:moment-difference}
\end{equation}
Under Assumption~\ref{ass:validation-regularity} the array
$\{\Psi_i(\bm\beta)\}$ is uniformly bounded, so for every fixed $\bm\beta$ the
SRSWOR variance of the right-hand side is
$\frac{1-f_N}{n_N}\mb S_{\Psi,N}=O(N^{-1})$. This proves the pointwise convergence by Chebyshev's inequality. Since $\Psi_i(\bm\beta)$ is affine in
$D$ and the parameter set is compact hence $\sup\|\bm\beta\|<\infty$,
$\sup_{\bm\beta\in\mathcal B}\norm{\widehat{\mb g}_N^{\PP}(\bm\beta)
-\mb g_N^{\Human}(\bm\beta)}_F=o_p(1)$ by the uniform law of large number.

It follows that
\[
\widehat{\bm\beta}_{\PP}-\widehat{\bm\beta}_{\Human}=O_p(N^{-1/2}),
\]
where $\widehat{\bm\beta}_{\Human}$ denotes the infeasible estimator obtained by
applying the same latent-adjusted 2SLS procedure to $\mb Y^{\Human}$ for all
residents. Since Theorem~\ref{thm:clt} gives
$\widehat{\bm\beta}_{\Human}\toP\bm\beta_0$, it follows that
$\widehat{\bm\beta}_{\PP}\toP\bm\beta_0$.
The rank condition of Proposition~\ref{identify} identifies $\bm\beta_0$ uniquely, and the continuous mapping theorem applied to the closed form \eqref{eq:ppi-2sls}
gives $\widehat{\bm\beta}_{\PP}\toP\bm\beta_0$.

\paragraph{Asymptotic normality.}
Evaluating the structural residual at the true parameter and using
\eqref{eq:Yppi} and \eqref{eq:Lambda-Psi},
\begin{equation}
\widehat{\mb Y}^{\PP}-\mb G\widehat{\mb Y}^{\PP}\mb D_0-\mb X\mb B_{10}-\mb G\mb X\mb B_{20}
=\mb E+\sum_{i=1}^N a_i\Lambda_i(\mb D_0).
\label{eq:ppi-residual-decomp}
\end{equation}
To compare the PPPI and complete-human estimators, define the normalized 2SLS
matrices
\[
\mb J_N=\frac1N \mb K^\top \mb M_{\widehat{\mb \Phi}_N}\mb K,
\qquad
\mb C_N^{\Human}=\frac1N \mb K^\top \mb M_{\widehat{\mb \Phi}_N}\mb Z^{\Human},
\qquad
\widehat{\mb  C}_N=\frac1N \mb K^\top \mb M_{\widehat{\mb \Phi}_N}\widehat{\mb Z}^{\PP},
\]
and, for any conformable $\mb C,\mb J$,
\[
\mathcal H(\mb C,\mb J)=(\mb C^\top \mb J^{-1}\mb C)^{-1}\mb C^\top \mb J^{-1}.
\]
Write
\[
\mb H_N^{\Human}=\mathcal H(\mb C_N^{\Human},\mb J_N),
\qquad
\widehat{\mb H}_N=\mathcal H(\widehat{\mb C}_N,\mb J_N).
\]
The oracle cross-moment results give
$\mb J_N\toP\Sigma_{\widetilde K}$ and
$\mb C_N^{\Human}\toP\Sigma_{\widetilde K\widetilde Z}$.
Thus \eqref{eq:ppi-proof-cross-moment-errors}, the strong-relevance
condition, and local smoothness of $\mathcal H$ give
\begin{equation}
\widehat{\mb H}_N-\mb H_N^{\Human}=O_p(N^{-1/2}),
\qquad
\widehat{\mb H}_N,\mb H_N^{\Human}\toP
\mb H_0:=\mathcal Q_0^{-1}\Sigma_{\widetilde Z\widetilde K}
\Sigma_{\widetilde K}^{-1}.
\label{eq:ppi-proof-map-convergence}
\end{equation}
Also, the oracle score and sieve-approximation arguments used to prove
Theorem~\ref{thm:clt} give the moment-consistency result
\[
\mb r_N:=\frac1N \mb K^\top \mb M_{\widehat{\mb \Phi}_N}\mb E=o_p(1).
\]
Using \eqref{eq:ppi-residual-decomp} and the identities
$\widehat{\mb H}_N\widehat{\mb C}_N=\mb I$ and
$\mb H_N^{\Human}\mb C_N^{\Human}=\mb I$, we obtain the exact decomposition
\begin{align}
\widehat{\bm\beta}_{\PP}-\widehat{\bm\beta}_{\Human}
={}&\widehat{\mb H}_N\frac1N\sum_{i=1}^N a_i
\mb K^\top \mb M_{\widehat{\mb \Phi}_N}\Lambda_i(\mb D_0)
+(\widehat{\mb H}_N-\mb H_N^{\Human})\mb r_N.
\label{eq:ppi-proof-beta-difference}
\end{align}
The second term is $o_p(N^{-1/2})$ by
\eqref{eq:ppi-proof-map-convergence} and $\mb r_N=o_p(1)$.

It remains to replace the estimated sieve in the first term by the oracle
sieve. Let
\[
\widehat\xi_{Ni}=\vecop\{\mb K^\top \mb M_{\widehat{\mb \Phi}_N}\Lambda_i(\mb D_0)\},
\qquad
\xi_{Ni}=\vecop\{\mb K^\top \mb M_{\mb\Phi_N}\Lambda_i(\mb D_0)\}.
\]
The estimated-sieve results used to prove Theorem~\ref{thm:clt} give
$\|\mb M_{\widehat{\mb \Phi}_N}-\mb M_{\mb\Phi_N}\|=o_p(1)$. Put $\mb F_N=\mb M_{\widehat{\mb \Phi}_N}-\mb M_{\mb\Phi_N}$.
Expanding $\mb K^\top\mb F_N\Lambda_i(\mb D_0)$ and summing squared norms gives
\begin{align*}
    \frac1N\sum_i\|\widehat\xi_{Ni}-\xi_{Ni}\|_2^2
    &\le\frac{2c_\Delta^2}{N}
    \left\{\|\mb F_N\mb K\|_F^2+\|\mb D_0\|^2\|\mb G^\top \mb F_N\mb K\|_F^2\right\}\\
    &\le 2c_\Delta^2(1+\|\mb D_0\|^2\|\mb G\|^2)
    \|\mb F_N\|^2\frac{\|\mb K\|_F^2}{N}=o_p(1).
\end{align*}
Applying the SRSWOR variance identity \eqref{eq:ppi-proof-srs-cov} once more,
\begin{equation}
\frac1{\sqrt N}\sum_{i=1}^N
 a_i(\widehat\xi_{Ni}-\xi_{Ni})=o_p(1).
\label{eq:ppi-proof-validation-replacement}
\end{equation}
Therefore, with
\[
\mb B_N:=\frac1{\sqrt N}\sum_{i=1}^N a_i\xi_{Ni},
\qquad
\Gamma_0=\mb I_m\otimes\mb H_0,
\]
the covariance identity and $\mb S_{\xi,N}\toP\Omega_\xi$ imply
$B_N=O_p(1)$. Vectorizing \eqref{eq:ppi-proof-beta-difference}, using
\eqref{eq:ppi-proof-validation-replacement} and
$\mb I_m\otimes\widehat{\mb H}_N\toP\Gamma_0$, therefore gives
\begin{equation}
\sqrt N\{\vecop(\widehat{\bm\beta}_{\PP})-\vecop(\bm\beta_0)\}
=
\underbrace{\sqrt N\{\vecop(\widehat{\bm\beta}_{\Human})-
\vecop(\bm\beta_0)\}}_{=:O_N}
+\Gamma_0\mb B_N+o_p(1).
\label{eq:ppi-linear-expansion}
\end{equation}
This is the key expansion. The term $-\mb G\mb e_i\Delta_i^\top \mb D_0$ inside
$\Lambda_i(\mb D_0)$ is precisely the first-order effect of correcting a resident's
outcome on the endogenous peer outcome.

By Theorem~\ref{thm:clt},
\[
\mb O_N\tod\mathcal N\!\left(
0,\Sigma_{\widetilde E}\otimes\mathcal Q_0^{-1}
\right).
\]
For the validation term, conditional on $\mathcal F_N$,
\[
\E_R(\mb B_N\mid\mathcal F_N)=0,
\qquad
\Cov_R(\mb B_N\mid\mathcal F_N)
=\frac{1-f_N}{f_N}\mb S_{\xi,N}
\toP\Omega_R.
\]
To verify the finite-population CLT condition, let
$h_{N,\max}=\max_i\mb e_i^\top\mb P_{\mb\Phi_N}\mb e_i$.
The uniform sieve-Gram lower bound and bounded-basis condition yield
\[
h_{N,\max}
\le\frac{\zeta_0(L_N)^2}
{N\lambda_{\min}(\mb\Phi_N^\top\mb\Phi_N/N)}=o_p(1).
\]
Since $\mb P_{\mb\Phi_N}$ is an orthogonal projection,
\[
\|\mb M_{\mb\Phi_N}\mb K\|_{2,\infty}
\le\|\mb K\|_{2,\infty}+\sqrt{h_{N,\max}}\|\mb K\|_F.
\]
Expanding $\mb K^\top \mb M_{\mb\Phi_N}\Lambda_i(\mb D_0)$ and using
$\sum_jG_{ji}\le c_{G,N}$ therefore gives
\[
\frac{\max_i\|\xi_{Ni}\|_2}{\sqrt N}
\le c_\Delta(1+c_{G,N}\|\mb D_0\|)
\left\{\frac{\|\mb K\|_{2,\infty}}{\sqrt N}
+\sqrt{h_{N,\max}}\frac{\|\mb K\|_F}{\sqrt N}\right\}=o_p(1).
\]
For every fixed vector $\mb s$ such that $\mb s^\top\Omega_\xi \mb s>0$, this and
$f_N\to f\in(0,1)$ imply
\[
\frac{\max_i\{\mb s^\top(\xi_{Ni}-\bar\xi_N)\}^2}
{\min(n_N,N-n_N)\mb s^\top \mb S_{\xi,N}\mb s}\toP0,
\qquad \bar\xi_N=N^{-1}\sum_i\xi_{Ni}.
\]
This is the H\'ajek condition for simple random sampling 
\citep{li2017general}. For directions with
$\mb s^\top\Omega_\xi\mb s=0$, the conditional variance tends to zero and
conditional Chebyshev's inequality applies instead. Applying the
finite-population CLT along almost-surely convergent further subsequences
of these bounds gives, for each fixed $\mb s$,
\begin{equation}
    \varphi_N(\mb s):=
    \E_R\{\exp(\mathrm{i}\cdot\mb s^\top\mb B_N)\mid\mathcal F_N\}
    \toP\exp(-\mb s^\top\Omega_R\mb s/2),
    \qquad
    \Omega_R=\frac{1-f}{f}\Omega_\xi.
    \label{eq:ppi-proof-design-clt}
\end{equation}
Here $\Omega_\xi$ is the deterministic covariance limit in the theorem;
\eqref{eq:ppi-proof-design-clt} specifies precisely the conditional CLT used
below, including a possibly singular limiting covariance.

Finally, $\mb O_N$ is $\mathcal F_N$-measurable, whereas $\mb B_N$ is generated by the
resident-validation sampling design. For fixed conformable vectors $t,s$,
\[
\E\exp\{\mathrm{i}\mb t^\top\mb O_N+\mathrm{i}\mb s^\top\mb B_N\}
=
\E\left[
\exp(\mathrm{i}\mb t^\top\mb O_N)
\E_R\{\exp(\mathrm{i}\mb s^\top\mb B_N)\mid\mathcal F_N\}
\right].
\]
By \eqref{eq:ppi-proof-design-clt}, the inner conditional characteristic
function converges in probability to
$\exp(-\mb s^\top\Omega_R\mb s/2)$; since characteristic functions are bounded by one,
the convergence also holds in $L_1$. Combining this with the oracle CLT shows
that $\mb O_N$ and $\mb B_N$ converge jointly to independent Gaussian vectors.
Applying Slutsky's theorem to \eqref{eq:ppi-linear-expansion} therefore gives
\[
\sqrt N\{\vecop(\widehat{\bm \beta}_{\PP})-\vecop(\bm\beta_0)\}
\tod
\mathcal N\!\left(
0,
\Sigma_{\widetilde E}\otimes\mathcal Q_0^{-1}
+\Gamma_0\Omega_R\Gamma_0^\top
\right),
\]
which is \eqref{eq:ppi-clt}.

\subsection{Consistency of the PPI covariance estimator}
\label{app:ppi-se}

Recall $\widehat{\mb Z}^*=\mb M_{\widehat{\mb \Phi}_N}\widehat{\mb Z}^{\PP}$,
$\widehat{\mb K}^*=\mb M_{\widehat{\mb \Phi}_N}\mb K$,
$\mathcal A_N=(\widehat{\mb Z}^*)^\top \mb P_{\widehat{\mb K}^*}\widehat{\mb Z}^*$, and the
influences $\widehat{\mb b}_i$ of \eqref{eq:bihat}. Writing
$\widehat{\mb Z}_v^*=\mb I_m\otimes\widehat{\mb Z}^*$ and
$\widehat{\mb K}_v^*=\mb I_m\otimes\widehat{\mb K}^*$, we have
$\mb P_{\widehat{\mb K}_v^*}=\mb I_m\otimes \mb P_{\widehat{\mb K}^*}$, so the finite-sample 2SLS
map from a vectorized residual perturbation to the coefficient estimate is
$\widehat{\mb B}_N=\mb I_m\otimes\{\mathcal A_N^{-1}(\widehat{\mb Z}^*)^\top
\mb P_{\widehat{\mb K}^*}\}$ and
$\widehat{\mb B}_N\vecop(\widehat\Lambda_i^*)
=\vecop\{\mathcal A_N^{-1}(\widehat{\mb Z}^*)^\top \mb P_{\widehat{\mb K}^*}
\widehat\Lambda_i^*\}=\widehat{\mb b}_i$, which is \eqref{eq:bihat}. The same
identity reduces the covariance map of Algorithm~\ref{algorithm1} to
$\mathcal V_N(\mb V)=\mb V\otimes\mathcal A_N^{-1}$ for any $m\times m$ matrix $\mb V$,
which is the form used in \eqref{eq:Vtotal}. Finally, with
$\overline\Lambda^*=n^{-1}\sum_{i:R_i=1}\widehat\Lambda_i^*$,
\begin{equation}
\widehat{\mb S}_\Lambda
=\frac{1}{n-1}\sum_{i:R_i=1}
(\widehat\Lambda_i^*-\overline\Lambda^*)^\top
(\widehat\Lambda_i^*-\overline\Lambda^*)\in\R^{m\times m}.
\label{eq:SLambda}
\end{equation}

\begin{Proposition}[Consistency of the PPI covariance estimator]
\label{prop:variance-consistency}
Suppose the conditions of Theorem~\ref{thm:ppi-asymptotic} hold,
$\Sigma_{\widetilde E}$ is positive definite, and the finite-population second
moments of $\Lambda_{Ni}^*(\mb D_0)$ converge with no dominant resident. Then
$\widehat\Sigma_{\widetilde E}\toP\Sigma_{\widetilde E}$,
$N\widehat{\mb V}_{\beta,R}\toP\Gamma_0\Omega_R\Gamma_0^\top$, and consequently
$N\widehat{\mb V}_{\beta,\PP}\toP\Sigma_{\beta,\PP}$, so the Wald standard errors
based on \eqref{eq:Vtotal} are asymptotically valid.
\end{Proposition}

\begin{proof}
At the true parameter,
\eqref{eq:ppi-residual-decomp} implies, after latent residualization and in
the normalized quadratic forms relevant for covariance estimation,
$\widehat{\mb E}_{\PP}^*=\widetilde{\mb E}+\sum_i a_i\Lambda_{Ni}^*(\mb D_0)+o_p(1)$. The
cross term has conditional mean zero because
$\E_R(a_i\mid\mathcal F_N)=0$ and is $O_p(N^{-1/2})$. For the quadratic term,
each $\Lambda_{Ni}^*$ is supported (up to the peer-lag column $\mb G\mb e_i$) on row
$i$, so
\[
\frac{1}{N}\Big\{\sum_i a_i\Lambda_{Ni}^*\Big\}^\top
\Big\{\sum_i a_i\Lambda_{Ni}^*\Big\}
=\frac{1}{N}\sum_{i=1}^N a_i^2\,\Lambda_{Ni}^{*\top}\Lambda_{Ni}^*+o_p(1),
\]
an average of $N$ terms rather than an outer product of a single average, so
the finite-population law of large numbers applies. Since
$\E_R(a_i^2)=f_N(1/f_N-1)^2+(1-f_N)=\frac{1-f_N}{f_N}$,
\[
\frac{1}{N}\Big\{\sum_i a_i\Lambda_{Ni}^*\Big\}^\top
\Big\{\sum_i a_i\Lambda_{Ni}^*\Big\}
-\frac{1-f_N}{f_N}\mb S_{\Lambda,N}\toP0,
\qquad
\mb S_{\Lambda,N}=\frac{1}{N-1}\sum_{i=1}^N
(\Lambda_{Ni}^*-\overline\Lambda_N)^\top(\Lambda_{Ni}^*-\overline\Lambda_N),
\]
where the centering contributes at order $N^{-1}$ and is negligible. Hence
$\widehat{\mb V}_{E,\mathrm{raw}}=\Sigma_{\widetilde E}
+\frac{1-f_N}{f_N}\mb S_{\Lambda,N}+o_p(1)$. The sample covariance
$\widehat{\mb S}_\Lambda$ in \eqref{eq:SLambda} is design-consistent for
$\mb S_{\Lambda,N}$ under SRSWOR, and consistency of $\widehat{\bm \beta}_{\PP}$ and of
the sieve estimates permits replacement of $\mb D_0$ and $\mb M_{\mb\Phi_N}$ by their
sample analogues. Subtracting \eqref{eq:VEdebiased} gives
$\widehat\Sigma_{\widetilde E}\toP\Sigma_{\widetilde E}$.

Let $\mb B_N^{\Human}$ be the
finite-sample 2SLS map evaluated at the oracle human-reference design and
$\mb b_{Ni}^0=\mb B_N^{\Human}\vecop\{\Lambda_{Ni}^*(\mb D_0)\}$, so that
$\mb b_{Ni}^0=N^{-1}\Gamma_0\xi_{Ni}+o_p(N^{-1})$. Conditional on the complete
data, the first-order PPI perturbation of the coefficient is
$\sum_i a_i\mb b_{Ni}^0=N(\bar{\mb b}_S^0-\bar{\mb b}_N^0)$, whose exact SRSWOR variance is
$\frac{N^2(1-f_N)}{n_N}\mb S_{b,N}^0$. Consistency of $\widehat{\bm \beta}_{\PP}$ and of
the PPI design matrices implies that the plug-in influences $\widehat{\mb b}_i$
approximate $\mb b_{Ni}^0$ uniformly at the order required for the sample
covariance, so $\widehat{\mb S}_b$ is design-consistent for $\mb S_{b,N}^0$. Since the
2SLS map converges to $\Gamma_0/N$ on the appropriate score scale,
$N\widehat{\mb V}_{\beta,R}=N\frac{N^2(1-f_N)}{n_N}\widehat{\mb S}_b
\toP\Gamma_0\Omega_R\Gamma_0^\top$.

Continuity of the covariance map and Theorem~\ref{thm:clt} give
$N\{\widehat\Sigma_{\widetilde E}\otimes\mathcal A_N^{-1}\}
\toP\Sigma_{\widetilde E}\otimes\mathcal Q_0^{-1}$; adding the two consistent
components proves the claim. The positive-semidefinite projection applied to
\eqref{eq:VEdebiased} is asymptotically inactive because
$\Sigma_{\widetilde E}$ is positive definite.
\end{proof}

\subsection{Linearized ALR specification}
\label{app:linearized}

The main analysis uses the plug-in ALR outcome, obtained by first constructing
the human-calibrated probability profile and then applying the nonlinear map
$\mb T(\cdot)$. As an alternative we linearize $\mb T$ around the BART profile
$\mb q_i^{\AI}$ using a first-order Taylor expansion. With category $C$ as reference, the $(C-1)\times C$ Jacobian
$\mb J_T(\mb q)=\partial\mb T(\mb q)/\partial\mb q^\top$ is
\[
\mb J_T(\mb q)=
\begin{pmatrix}
q_1^{-1} & 0 & \cdots & 0 & -q_C^{-1}\\
0 & q_2^{-1} & \cdots & 0 & -q_C^{-1}\\
\vdots & & \ddots & & \vdots\\
0 & 0 & \cdots & q_{C-1}^{-1} & -q_C^{-1}
\end{pmatrix},
\]
and the linearized human-calibrated outcome is
\begin{equation}
\widehat{\mb Y_i}^{\Human,\mathrm{lin}}
=\mb Y_i^{\AI}
+\mb J_T\!\left\{\mathcal P_\varepsilon(\mb q_i^{\AI})\right\}
\left(\widehat{\mb q}_i^{\Human}-\mb q_i^{\AI}\right).
\label{eq:YH-linearized}
\end{equation}
Since $\widehat{\mb q}_i^{\Human}$ and $\mb q_i^{\AI}$ both sum to one, the argument of
$\mb J_T$ acts on a valid tangent direction. The construction is affine in the
prediction-powered probability estimator, so for a fixed expansion point the
design unbiasedness of $\widehat{\mb q}_i^{\Human}$ carries directly to the
corresponding ALR target. The trade-off is that \eqref{eq:YH-linearized}
targets a first-order approximation to $\mb T(\mb q_i^{\Human})$ rather than the
nonlinear ALR itself. Peer-effect estimation proceeds exactly as in
Sections~\ref{sec:resident-ppi}--\ref{sec:ppi-se}, with
\[
\widehat{\mb Y_i}^{\PP,\mathrm{lin}}
=\mb Y_i^{\AI}+\frac{R_i}{\pi_i}
\left(\widehat{\mb Y_i}^{\Human,\mathrm{lin}}-\mb Y_i^{\AI}\right),
\]
the peer lag $\mb G\widehat{\mb Y}^{\PP,\mathrm{lin}}$ recomputed before applying the
same latent-adjusted IV--2SLS estimator and resident-sampling covariance
correction.

\subsection{Two-stage PPI++ specification}
\label{app:ppipp}

Ordinary PPI removes the bias introduced by LLM measurements, but does not guarantee improvement in efficiency relative to classical inference. In particular, if the LLM measures are weakly correlated with human preferences, then using LLM measures may even inflate the variance of the estimator. We therefore also develop a power-tuned version of PPPI similar to PPI++ of \citep{angelopoulos2023ppi++} that guarantees asymptotic efficiency compared to classical peer effect inference with available human data alone. The power-tuning can be done at the message level, or at the resident level, or at both levels.

\paragraph{Message-level PPI++.}
For a validated resident let
$\overline{\mb H}_i=s_i^{-1}\sum_{j\in\mathcal S_i}\mb H_{ij}$ and
$\overline{\mb L}_i=s_i^{-1}\sum_{j\in\mathcal S_i}\mb L_{ij}$, and set
\begin{equation}
\widehat{\mb q}_i^{\Human,+}
=\overline{\mb H}_i+\widehat\lambda_M\left(\mb q_i^{\AI}-\overline{\mb L}_i\right),
\label{eq:message-ppipp}
\end{equation}
which reduces to \eqref{eq:message-ppi} at $\lambda_M=1$. It is important to have a single tuning
parameter across all $C$ categories due to  compositional structure. Because
$\overline{\mb H}_i,\overline{\mb L}_i,\mb q_i^{\AI}$ each sum to one, so does
$\widehat{\mb q}_i^{\Human,+}$ for every resident. While it is possible to vary the tuning parameter $\lambda$ across residents we do not pursue it here. Under message SRSWOR let $w_i=(1-s_i/M_i)/s_i$. The
common tuning parameter minimizes the trace of the estimated message-sampling
covariance,
\begin{equation}
\widehat\lambda_M^{\mathrm{raw}}
=\frac{\sum_{i:R_i=1}w_i\operatorname{tr}
\{\widehat{\operatorname{Cov}}_i(\mb H_{ij},\mb L_{ij})\}}
{\sum_{i:R_i=1}w_i\operatorname{tr}
\{\widehat{\operatorname{Var}}_i(\mb L_{ij})\}},
\qquad
\widehat\lambda_M=\min\{1,\max(0,\widehat\lambda_M^{\mathrm{raw}})\},
\label{eq:lambda-message}
\end{equation}
i.e. within-resident covariance and BART variance are pooled over categories
and residents before taking the ratio, which is important at $s_i=5$ where
resident-specific ratios could be noisy. The resulting human-calibrated
outcome is either the plug-in ALR
$\widehat{\mb Y_i}^{\Human,+}=\mb T\{\mathcal P_\varepsilon(\widehat{\mb q}_i^{\Human,+})\}$
or its linearized analogue, obtained by replacing $\widehat{\mb q}_i^{\Human}$ with
$\widehat{\mb q}_i^{\Human,+}$ in \eqref{eq:YH-linearized}. Conditional on the
realized first-stage data these are treated as the human-reference
measurements, as in Remark~\ref{rem:fixed-human}.

\paragraph{Resident-level PPI++.}
One specification is to use the improved profile but retain  the ordinary
resident-level PPI,
\begin{equation}
\widehat{\mb Y_i}^{\mathrm{profile}+}
=\mb Y_i^{\AI}+\frac{R_i}{\pi_i}
\left(\widehat{\mb Y_i}^{\Human,+}-\mb Y_i^{\AI}\right),
\label{eq:profile-ppipp}
\end{equation}
so that only the construction of the calibrated profile is power-tuned. The
full two-stage specification adds a resident-level parameter,
\begin{equation}
\widehat{\mb Y_i}^{++}(\lambda_R)
=\lambda_R\mb Y_i^{\AI}+\frac{R_i}{\pi_i}
\left(\widehat{\mb Y_i}^{\Human,+}-\lambda_R\mb Y_i^{\AI}\right),
\label{eq:resident-ppipp}
\end{equation}
which satisfies $\E_R\{\widehat{\mb Y_i}^{++}(\lambda_R)\mid
\widehat{\mb Y_i}^{\Human,+},\mb Y_i^{\AI}\}=\widehat{\mb Y_i}^{\Human,+}$ for every fixed
$\lambda_R$, so unbiasedness of the rectified moment is preserved for
any choice of $\lambda_R$, while $\lambda_R=1$ recovers \eqref{eq:profile-ppipp}.
We choose $\lambda_R$ to reduce the estimated resident-sampling variance of
the peer-effect coefficients using the resident-level IV influence functions.
For a selected coefficient block, with $\mb b_i^{\Human}$ and $\mb b_i^{\AI}$ the
human-reference and BART influence contributions, we estimate,
\begin{equation}
\widehat\lambda_R
=\frac{\sum_{i:R_i=1}
\langle \mb b_i^{\Human}-\overline{\mb b}^{\Human},
\mb b_i^{\AI}-\overline{\mb b}^{\AI}\rangle}
{\sum_{i:R_i=1}\norm{\mb b_i^{\AI}-\overline{\mb b}^{\AI}}_2^2}.
\label{eq:lambda-resident}
\end{equation}
After constructing either \eqref{eq:profile-ppipp} or
\eqref{eq:resident-ppipp} we recompute the peer lag and apply the same
latent-adjusted IV--2SLS estimator and covariance correction as for ordinary
PPPI. Estimation error in
$\widehat\lambda_M$ and $\widehat\lambda_R$ is treated as $o_p(1)$, as is standard for PPI++. These two specifications at the resident level separates the empirically well-motivated
use of PPI++ for constructing resident language profiles from the further
question of whether power tuning at the downstream peer-effect stage yields
additional efficiency gains.

\setcounter{table}{0}
\renewcommand{\thetable}{A\arabic{table}}
\setcounter{figure}{0}
\renewcommand{\thefigure}{A\arabic{figure}}

\section{Additional Tables and Figures}

\begin{table}[H] \centering \renewcommand*{\arraystretch}
{0.8}
\caption{Summary Statistics}
\label{summarystatsdata}

\begin{tabular}{p{5.5cm}rrrrrrrrr}
\hline
\hline
Variable & N & Mean & SD & Min & p25 & p50 & p75 & p95 & Max \\ 
\hline
\textbf{A. Female Unit}&  &  &  &  &  &  &  & \\
Recidivism & 982 &  &  &  &  &  &  &  &  \\ 
\hspace{0.5cm} 0 & 769 & 78\% &  &  &  &  &  &  &  \\ 
\hspace{0.5cm} 1 & 213 & 22\% &  &  &  &  &  &  &  \\ 
White & 982 &  &  &  &  &  &  &  &  \\ 
\hspace{0.5cm} Non White & 204 & 21\% &  &  &  &  &  &  &  \\ 
\hspace{0.5cm} White & 778 & 79\% &  &  &  &  &  &  &  \\ 
LSI.R & 982 & 25 & 7.2 & 0 & 21 & 25 & 30 & 35 & 57 \\ 
Age & 982 & 31 & 8.3 & 18 & 24 & 29 & 37 & 46 & 60 \\ 
education & 982 &  &  &  &  &  &  &  &  \\ 
\hspace{0.5cm} less than HS & 539 & 55\% &  &  &  &  &  &  &  \\ 
\hspace{0.5cm} GED & 109 & 11\% &  &  &  &  &  &  &  \\ 
\hspace{0.5cm} HS & 232 & 24\% &  &  &  &  &  &  &  \\ 
\hspace{0.5cm} >HS but < College Grad. & 74 & 8\% &  &  &  &  &  &  &  \\ 
\hspace{0.5cm} College Degree & 28 & 3\% &  &  &  &  &  &  &  \\ \\
\textbf{B. Male Unit}&  &  &  &  &  &  &  & \\
Recidivism & 1649 &  &  &  &  &  &  &  &  \\ 
\hspace{0.5cm} 0 & 1229 & 75\% &  &  &  &  &  &  &  \\ 
\hspace{0.5cm} 1 & 420 & 25\% &  &  &  &  &  &  &  \\ 
White & 1649 &  &  &  &  &  &  &  &  \\ 
\hspace{0.5cm} Non White & 689 & 42\% &  &  &  &  &  &  &  \\ 
\hspace{0.5cm} White & 960 & 58\% &  &  &  &  &  &  &  \\ 
LSI.R & 1649 & 26 & 10 & 0 & 22 & 25 & 29 & 35 & 354 \\ 
Age & 1649 & 29 & 8.8 & 18 & 21 & 26 & 35 & 46 & 60 \\ 
education & 1649 &  &  &  &  &  &  &  &  \\ 
\hspace{0.5cm} less than HS & 1018 & 62\% &  &  &  &  &  &  &  \\ 
\hspace{0.5cm} GED & 213 & 13\% &  &  &  &  &  &  &  \\ 
\hspace{0.5cm} HS & 293 & 18\% &  &  &  &  &  &  &  \\ 
\hspace{0.5cm} >HS but < College Grad. & 94 & 6\% &  &  &  &  &  &  &  \\ 
\hspace{0.5cm} College Degree & 31 & 2\% &  &  &  &  &  &  &  \\ 
\hline
\hline
\multicolumn{10}{c}{ \begin{minipage}{15 cm}{\footnotesize{Notes: This table provides the summary statistics for the primary outcome ``recidivism" and the pre-entry covariates. }}
\end{minipage}} \\
\end{tabular}
\end{table}

\begin{table}[!htbp]
\centering
\caption{Survey 1 Questions for Human Validation of Input Labels}
\label{tab:survey1_questions}
\begin{tabular}{p{1cm}p{13.5cm}}
\hline\hline
Q1 &
Does the label ``Rule Violations'' capture a meaningful behavioral dimension in the messages you read? For your reference, rule violations may refer to breaking program rules, procedures, or expected standards of conduct. \\\\

Q2 &
Does the label ``Disruptive Conduct'' capture a meaningful behavioral dimension in the messages you read? For your reference, disruptive conduct may refer to disrespectful, hostile, antisocial, or socially disruptive behavior toward others or the community. \\\\

Q3 &
Does the label ``Personal Growth'' capture a meaningful behavioral dimension in the messages you read? For your reference, Personal Growth may refer to self-improvement, accountability, effort, responsibility, and engagement with the program. \\\\

Q4 &
Does the label ``Community Support'' capture a meaningful behavioral dimension in the messages you read? For your reference, community support may refer to helping, encouraging, mentoring, listening to, or positively supporting peers. \\\\

Q5 &
Considering all four labels together, Personal Growth, Community Support, Rule Violations, and Disruptive Conduct, do you think they together capture the relevant behavioral dimensions in the messages you read? \\
\hline\hline
\end{tabular}
\begin{minipage}{14cm}
\footnotesize
Notes: For Q1--Q4, respondents selected ``Yes'' or ``No.'' For Q5, respondents selected from ``Strongly agree,'' ``Neither agree nor disagree,'' and ``Strongly disagree.''
\end{minipage}
\end{table}

\begin{figure}[!h] 
  \caption{Entry, Exit Dates, Length of Stay and Time to Recidivism}
    \label{time} 

 \begin{minipage}[b]{0.5\linewidth}
       \begin{center}

    \includegraphics[width=\linewidth]{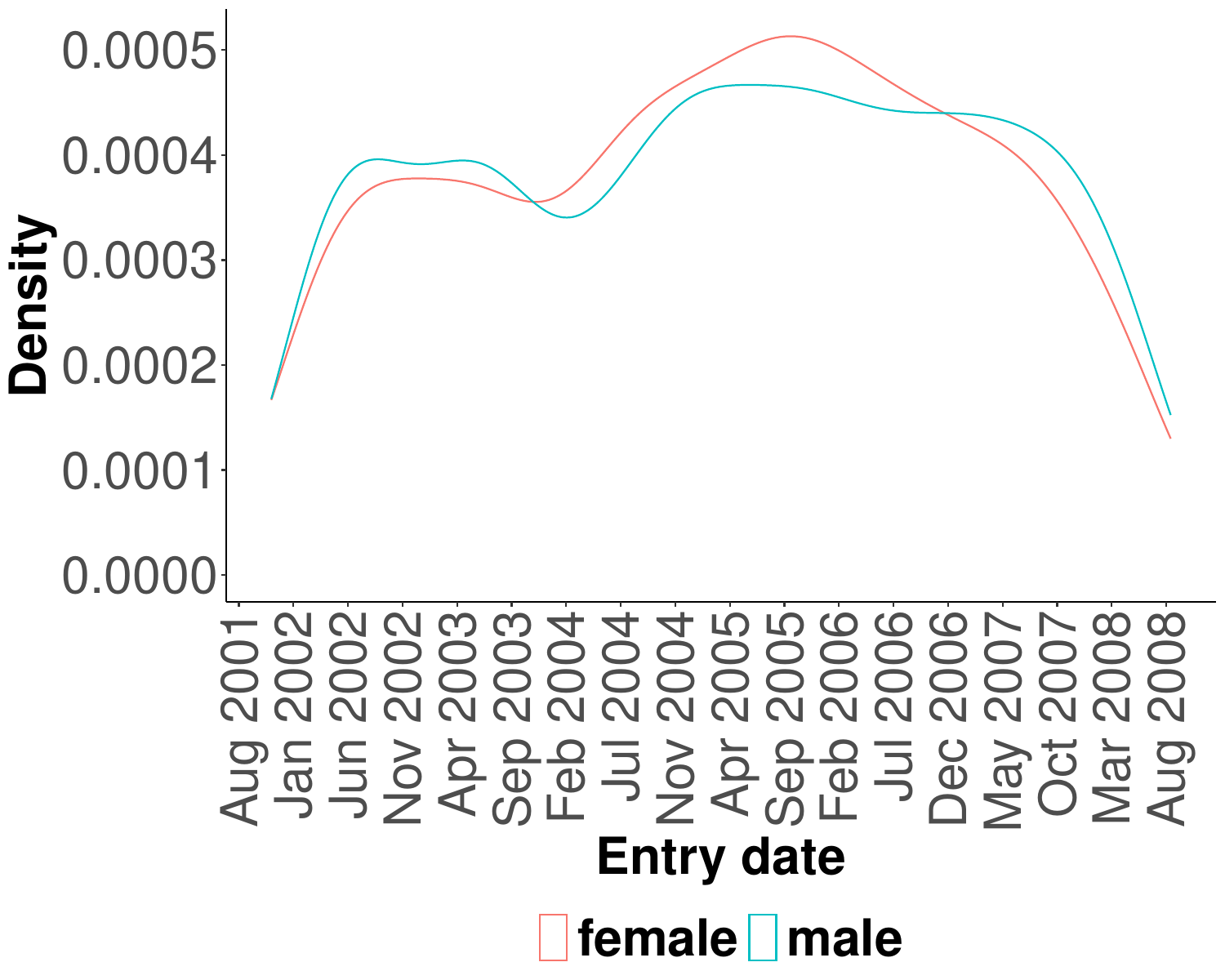}
    \caption*{a. Entry Date} 
        \end{center}

  \end{minipage}
  \begin{minipage}[b]{0.5\linewidth}
       \begin{center}

    \includegraphics[width=\linewidth]{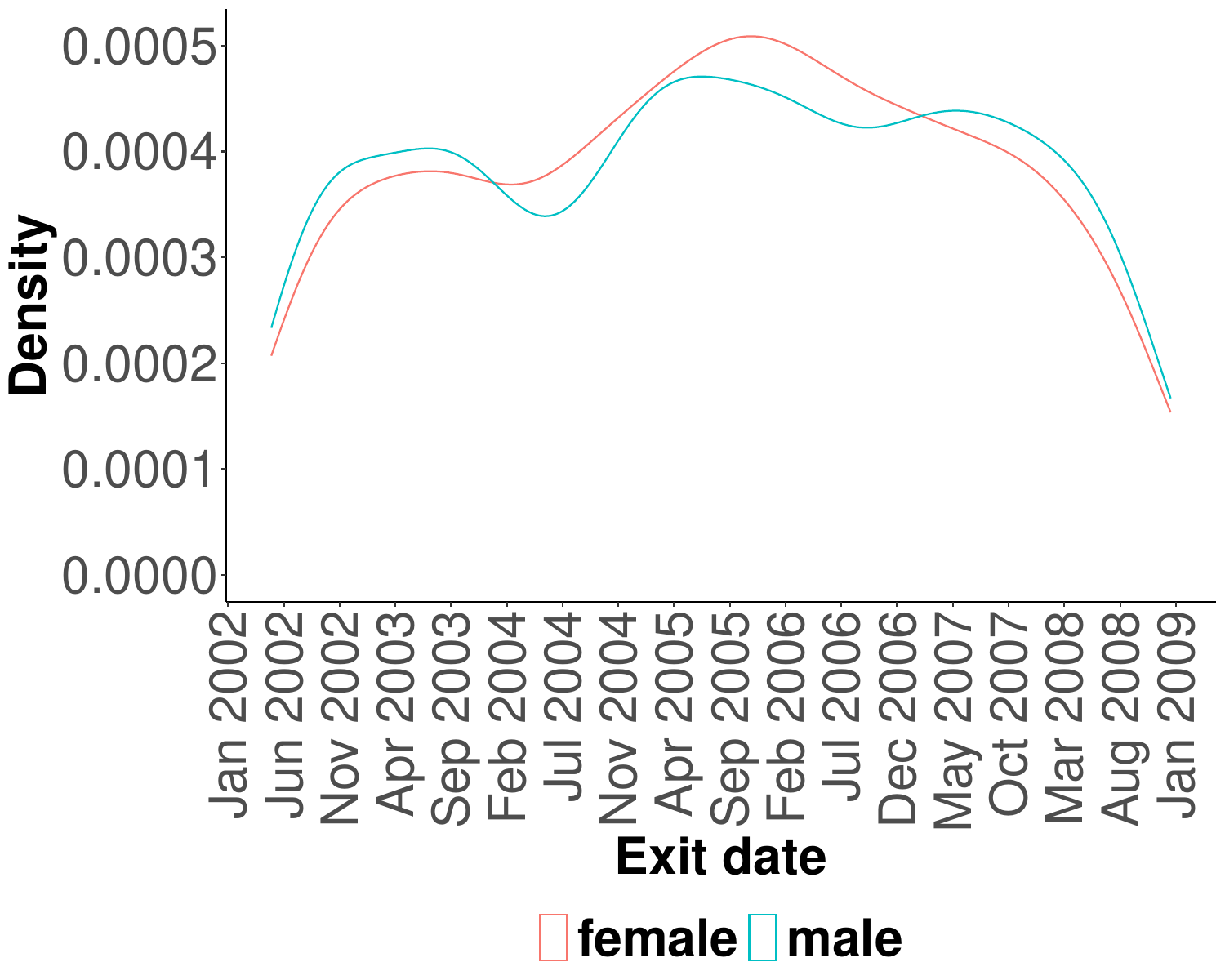}
    \caption*{b. Exit Date} 
        \end{center}

  \end{minipage}
   \begin{minipage}[b]{0.5\linewidth}
       \begin{center}

    \includegraphics[width=\linewidth]{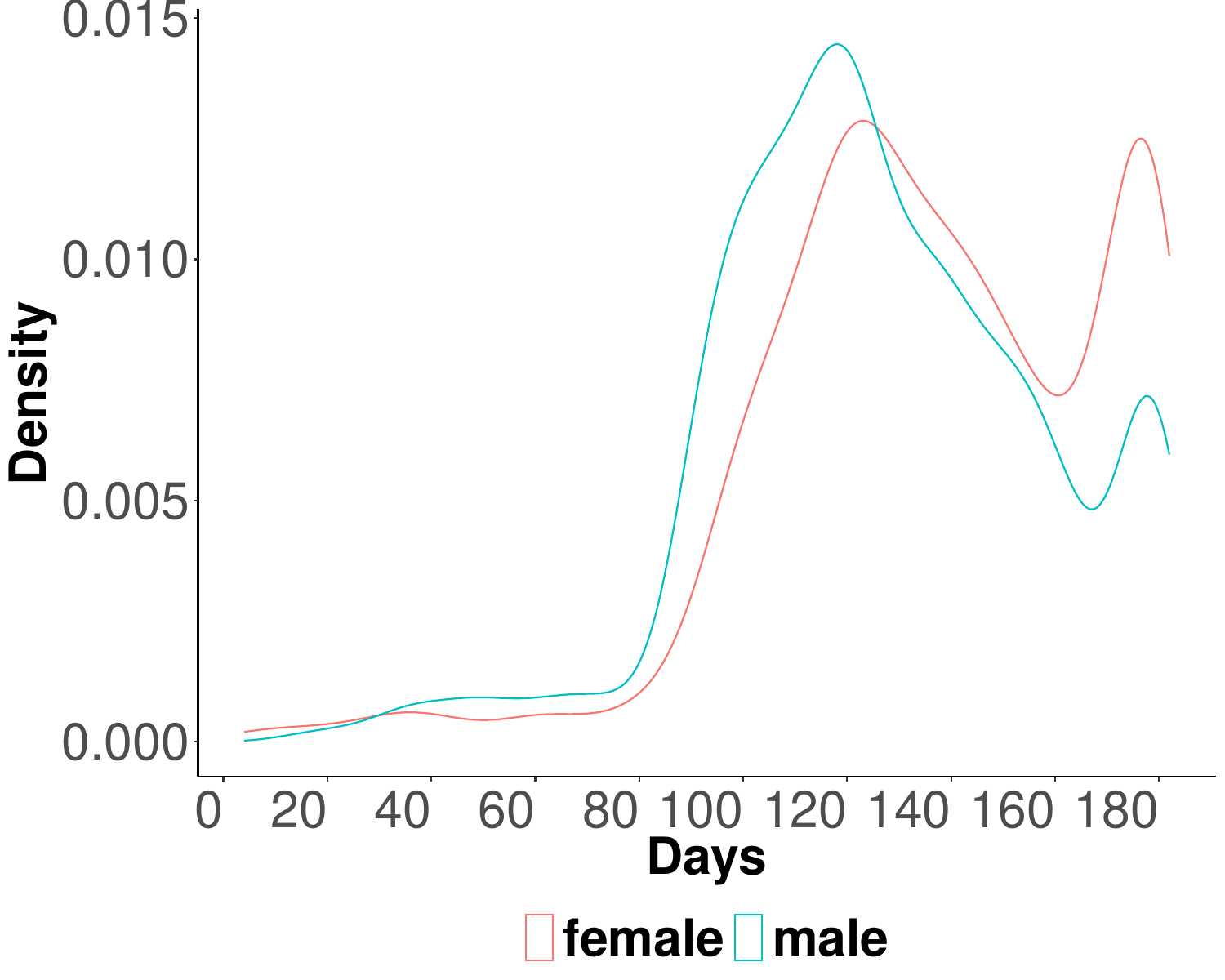}
    \caption*{c. Length of Stay} 
        \end{center}

  \end{minipage}
  \begin{minipage}[b]{0.5\linewidth}
       \begin{center}

    \includegraphics[width=\linewidth]{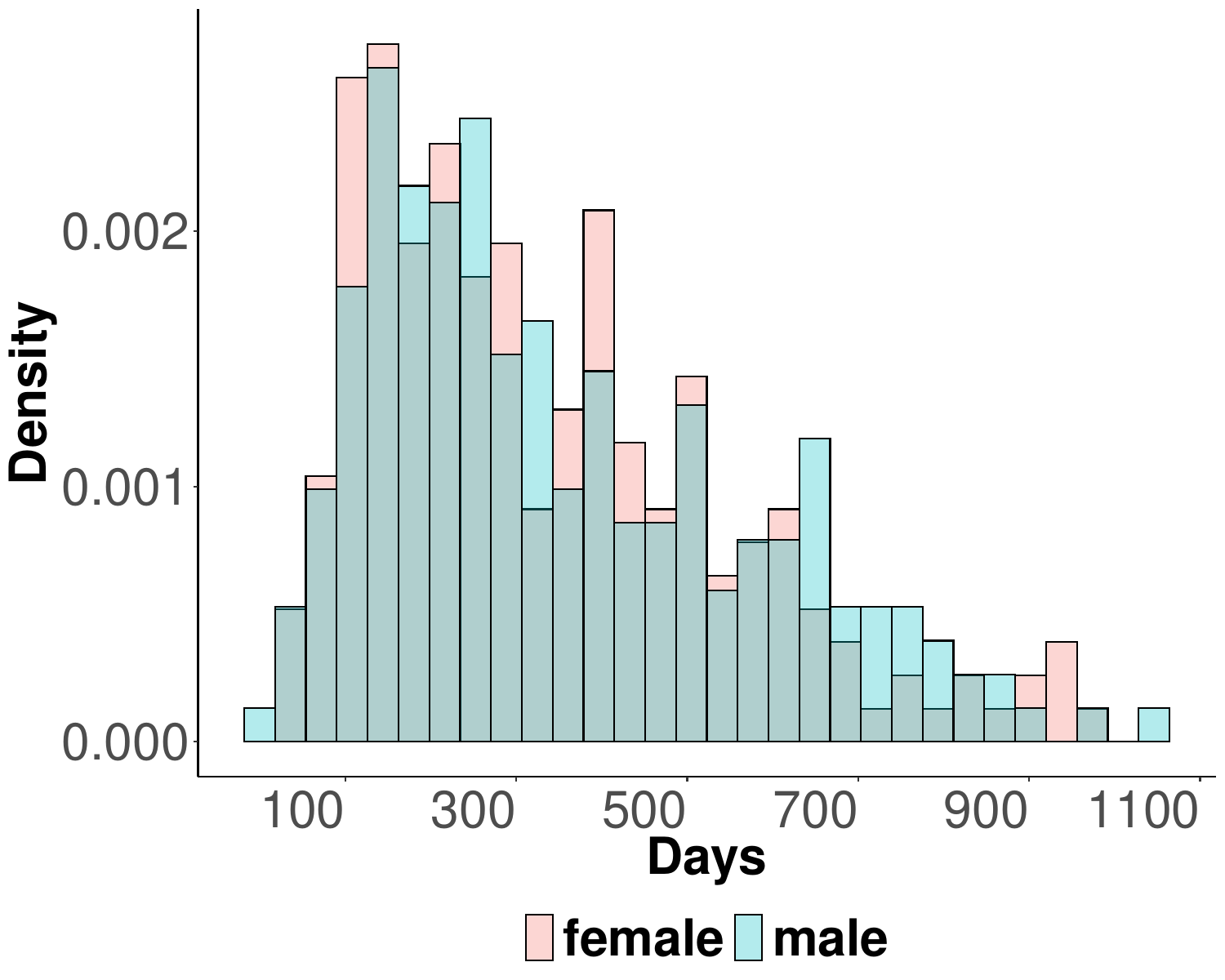}
    \caption*{d. Time to Recidivism} 
        \end{center}

  \end{minipage}
\end{figure}

\begin{figure}[!h] 
  \caption{Exchanges of Affirmations and Corrections (Count)}
    \label{counttexts} 

 \begin{minipage}[b]{0.5\linewidth}
       \begin{center}

    \includegraphics[width=\linewidth]{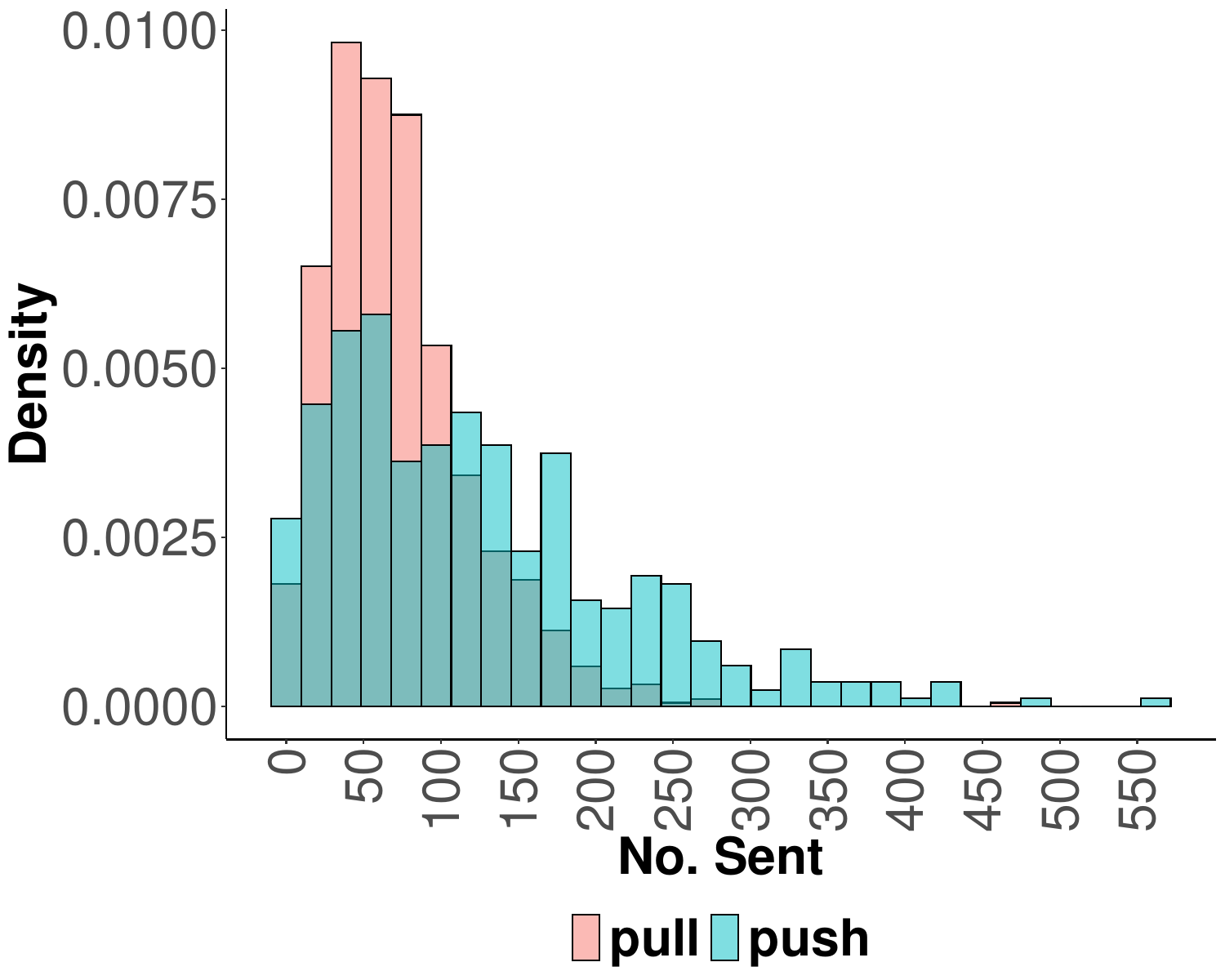}
    \caption*{a. Female} 
        \end{center}

  \end{minipage}
  \begin{minipage}[b]{0.5\linewidth}
       \begin{center}

    \includegraphics[width=\linewidth]{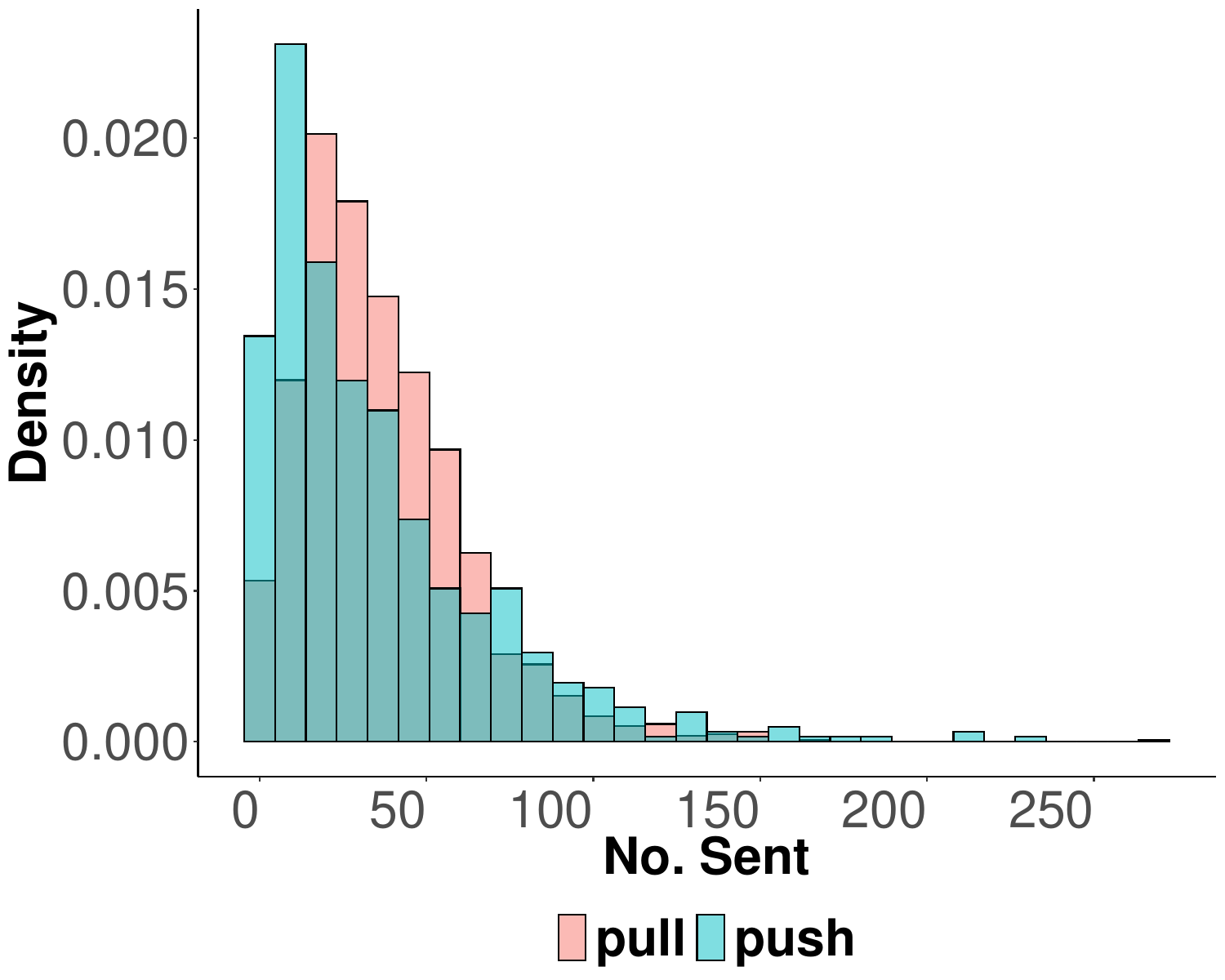}
    \caption*{b. Male} 
        \end{center}

  \end{minipage}
   \begin{minipage}[b]{0.5\linewidth}
       \begin{center}

    \includegraphics[width=\linewidth]{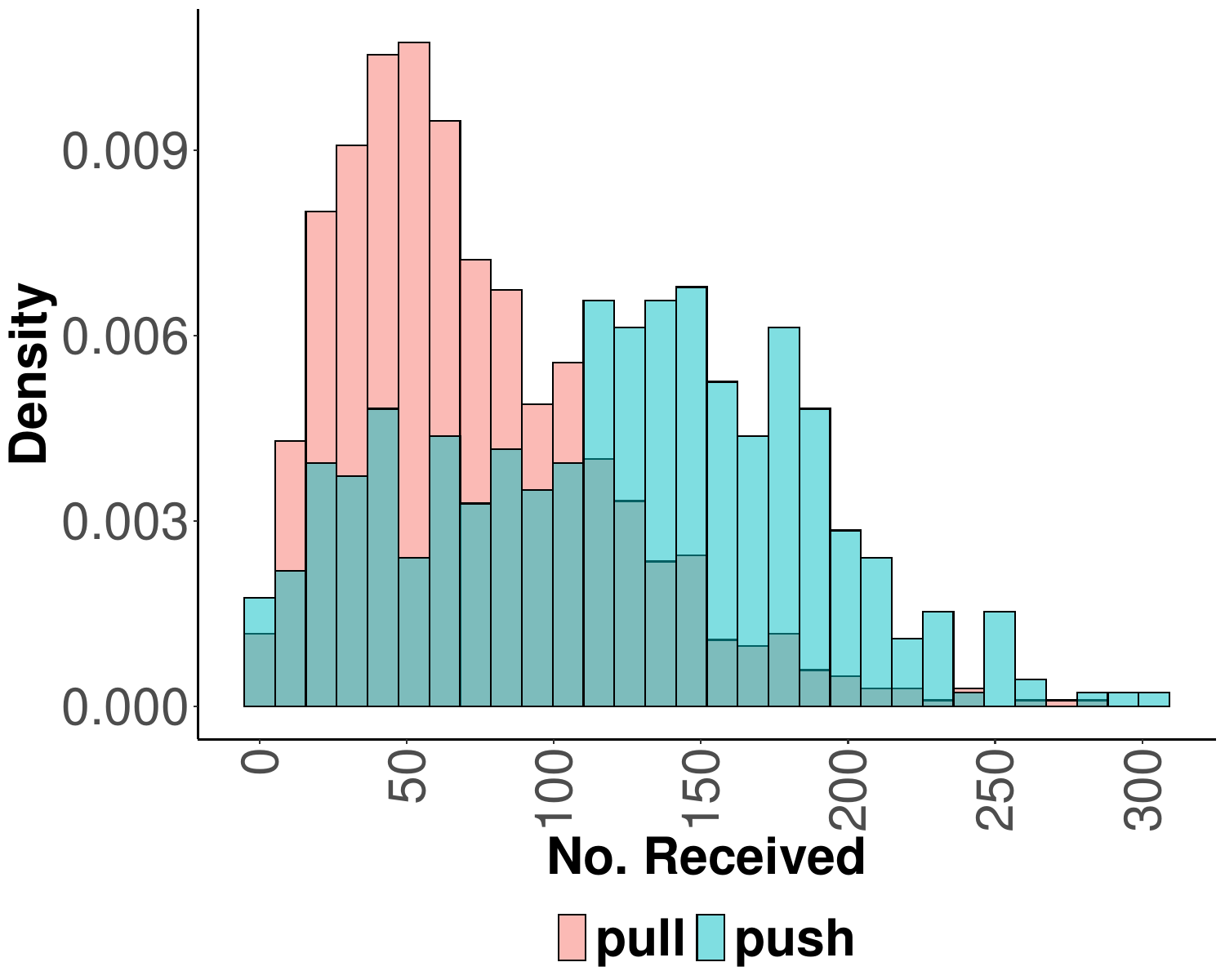}
    \caption*{c. Female} 
        \end{center}

  \end{minipage}
  \begin{minipage}[b]{0.5\linewidth}
       \begin{center}

    \includegraphics[width=\linewidth]{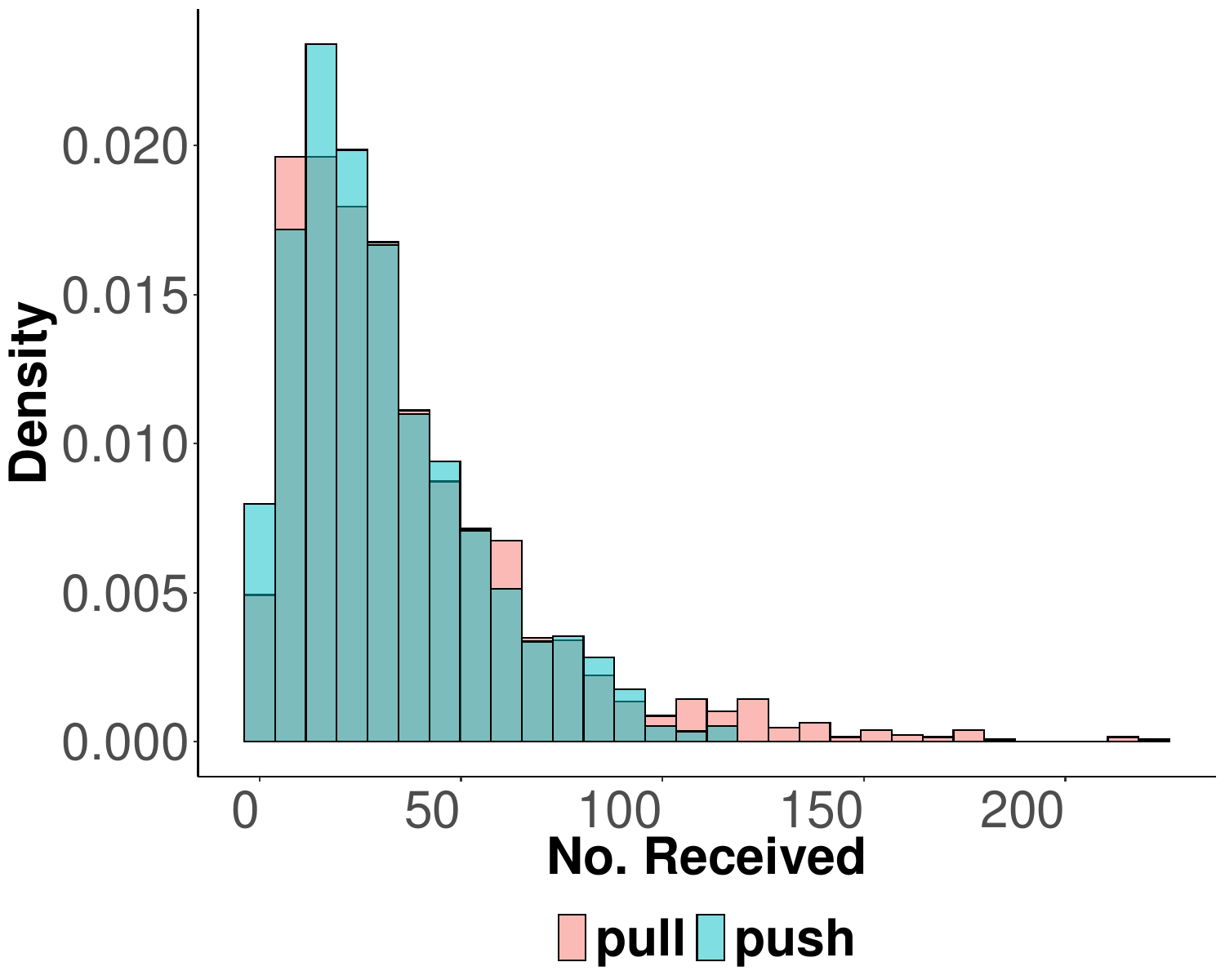}
    \caption*{d. Male} 
        \end{center}

  \end{minipage}
\end{figure}

\begin{figure}[!htbp]
  \centering
  \caption{Predicting message type by zero-shot probabilities on final input labels}
\label{predictaffirmationcorrection}
  \begin{subfigure}[t]{0.45\textwidth}
    \centering
    \includegraphics[width=\linewidth]{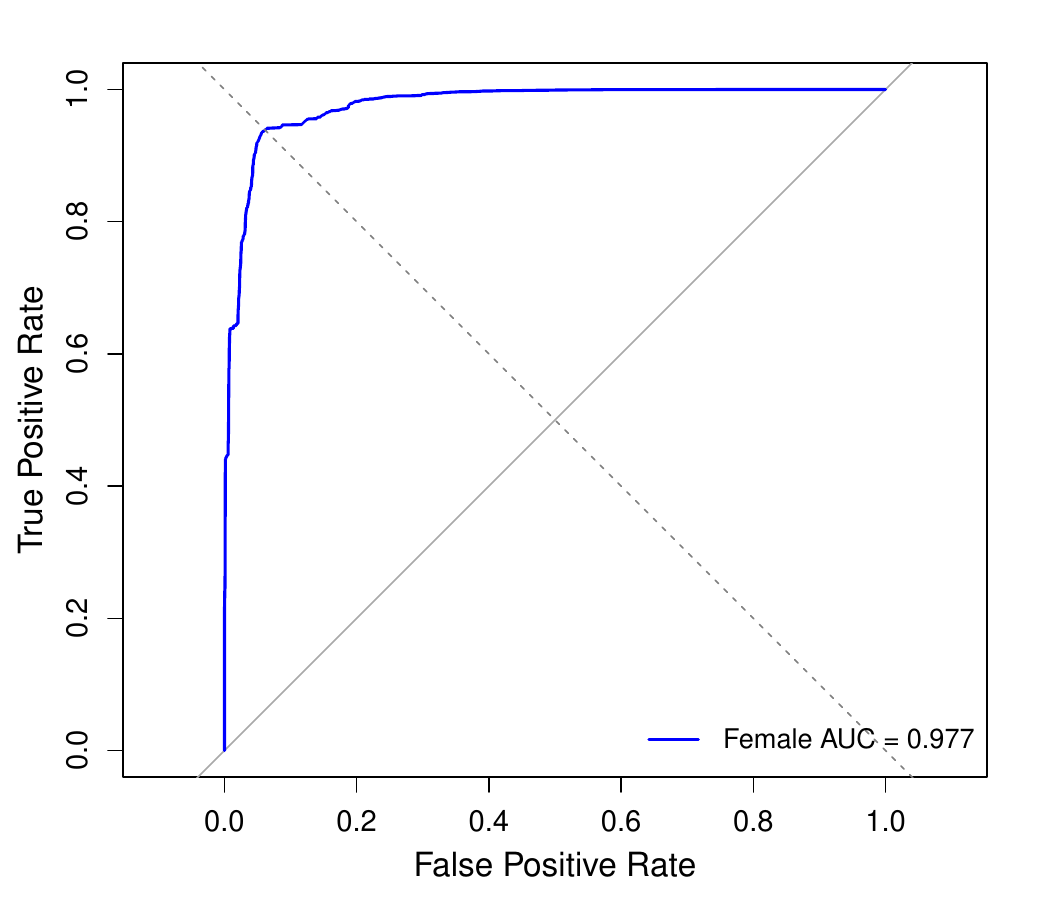}
    \caption{Female}
  \end{subfigure}
  \begin{subfigure}[t]{0.45\textwidth}
    \centering
    \includegraphics[width=\linewidth]{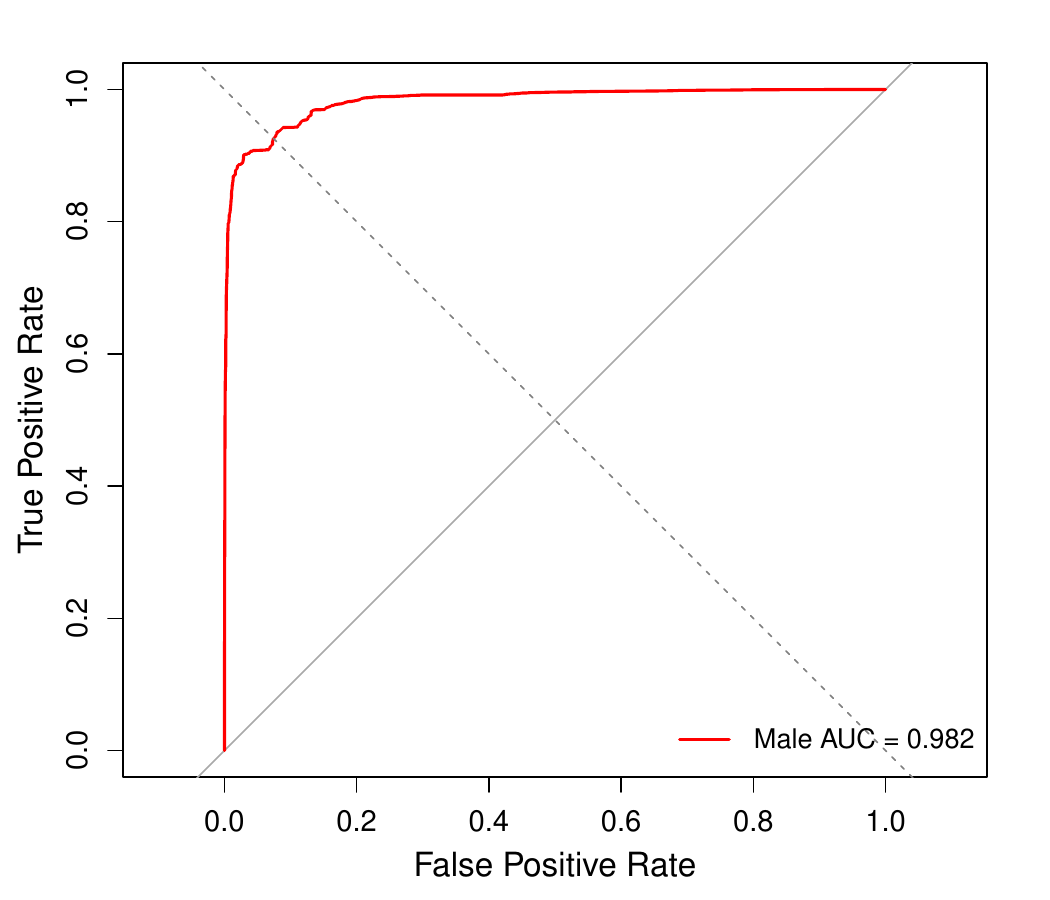}
    \caption{Male}
  \end{subfigure}

\begin{minipage}{15.0 cm}{\footnotesize{Notes: This analysis is conducted on the held-out data not used for label generation.}}
\end{minipage} 
\end{figure}

\begin{figure}[!htbp]
  \centering
  \caption{Validation of the zero-shot method using placebo labels}
\label{comparisonplacebo}
  \begin{subfigure}[t]{0.8\textwidth}
    \centering
    \includegraphics[width=\linewidth]{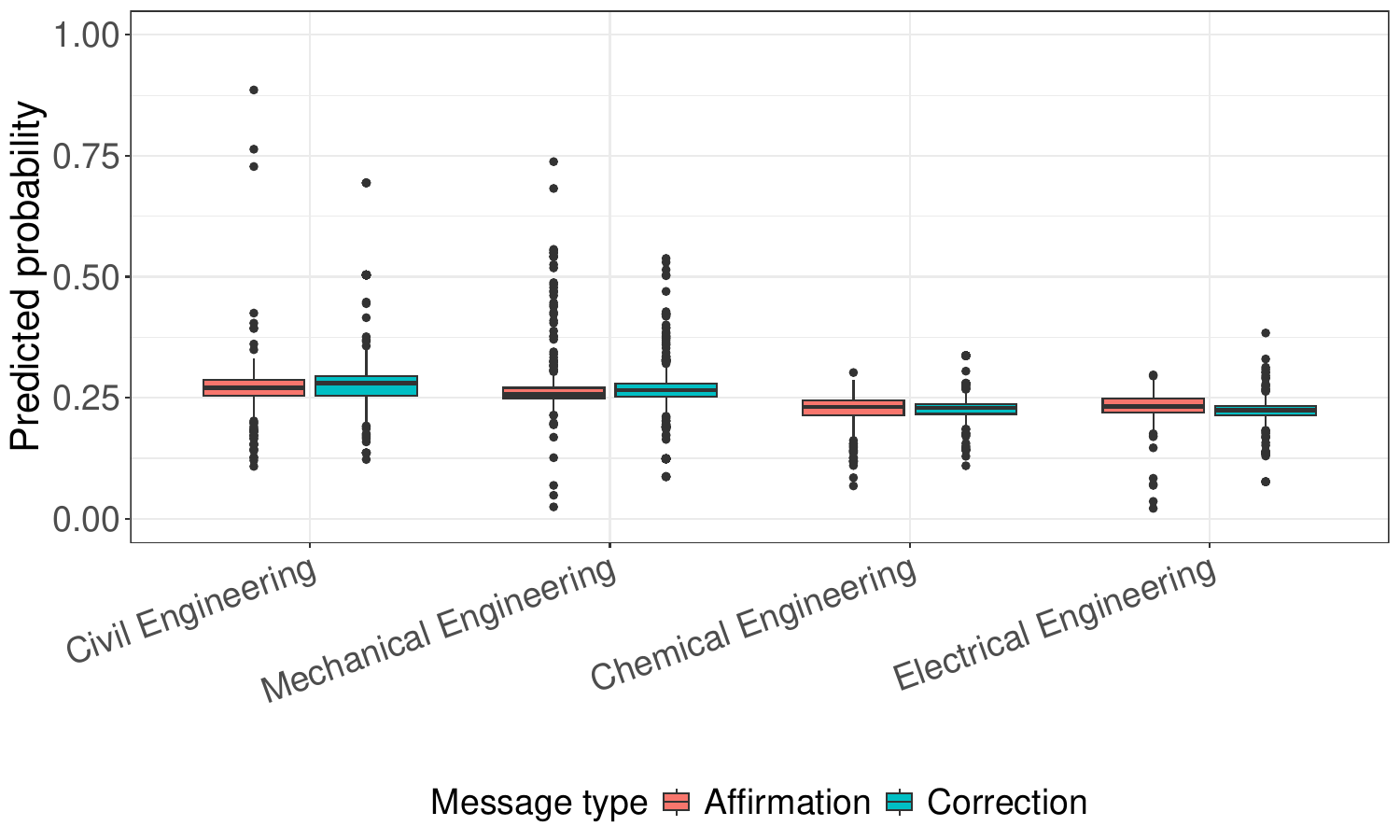}
    \caption{Arbitary labels}
  \end{subfigure}
  \begin{subfigure}[t]{0.8\textwidth}
    \centering
    \includegraphics[width=\linewidth]{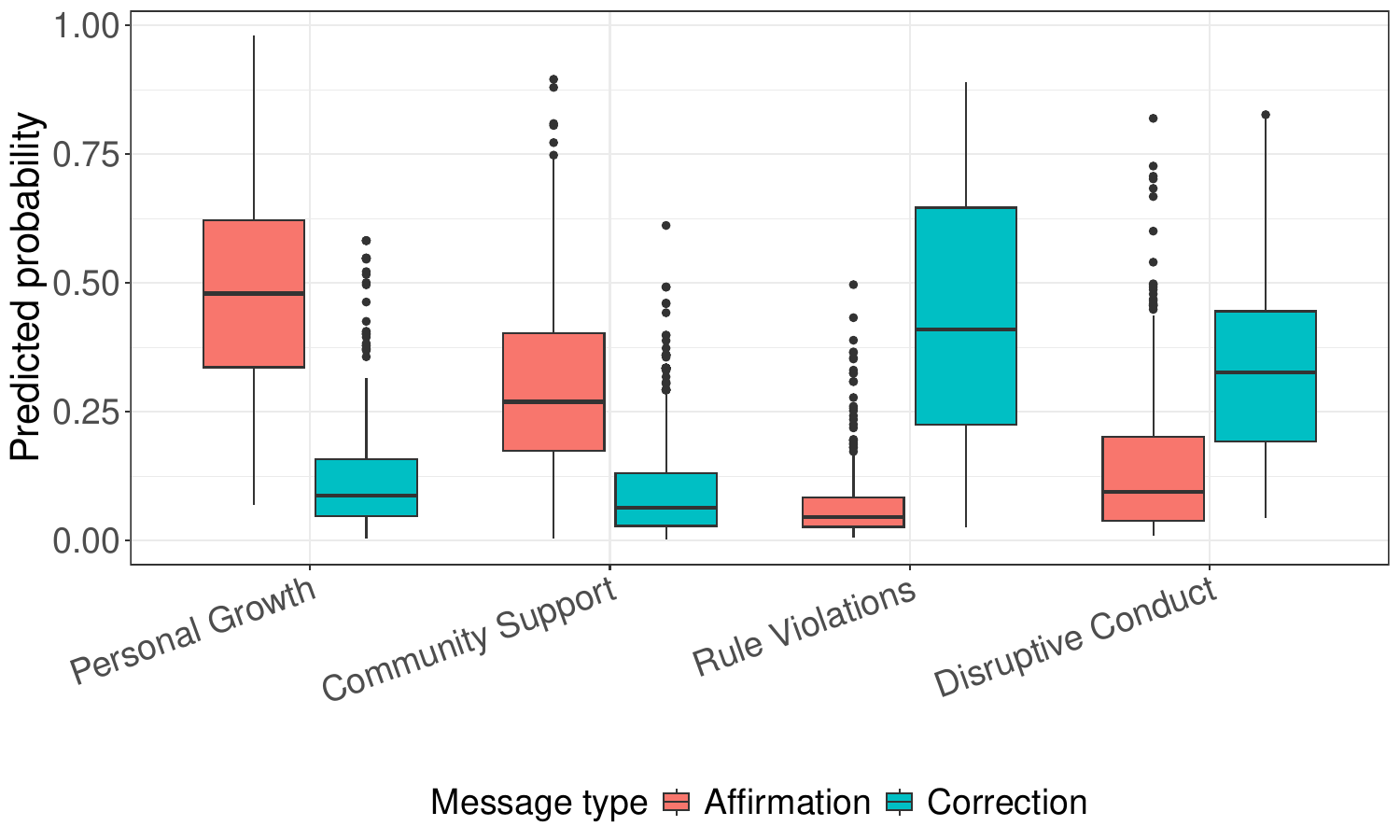}
    \caption{Behavioral labels}
  \end{subfigure}

\begin{minipage}{15.0 cm}{\footnotesize{Notes: This analysis is conducted on a randomly drawn sample of 1000 messages from the held-out data. Panel (a) shows the zero-shot predictions for four arbitrary labels and their distributions by message type. Panel (b) displays the zero-shot probabilities on the four chosen labels.}}
\end{minipage} 
\end{figure}

\begin{table}[!htbp]
\centering
\caption{Human validation of the four input labels on Prolific}
\label{tab:survey1_validation_main}
\begin{tabular}{p{5cm}p{2cm}p{2cm}p{2cm}}
\hline\hline
\multicolumn{4}{l}{{a. Individual labels}}\\
Label & Yes & No & \% Yes \\
\hline
Personal Growth & 67 & 1 & 98.5 \\
Community Support & 62 & 6 & 91.2 \\
Rule Violations & 67 & 1 & 98.5 \\
Disruptive Conduct & 65 & 3 & 95.6 \\
\hline
\multicolumn{4}{l}{{b. Joint assessment}}\\
Response & Count & Percent & \\
\hline
Strongly agree & 61 & 89.7 & \\
Neither agree nor disagree & 6 & 8.8 & \\
Strongly disagree & 1 & 1.5 & \\
\hline\hline
\multicolumn{4}{p{12.0cm}}{\footnotesize Notes: This table summarizes responses from 68 completed Prolific participants. Messages used for this validation exercise are a subset of the messages used for label generation.}\\
\end{tabular}
\end{table}

\begin{figure}[!htbp] 
  \caption{Resident-level bias of BART and PPI++ relative to human-validated profiles for sender profiles}
    \label{ppi_human_senderprofs} 

 \begin{minipage}[b]{0.5\linewidth}
       \begin{center}

    \includegraphics[width=\linewidth]{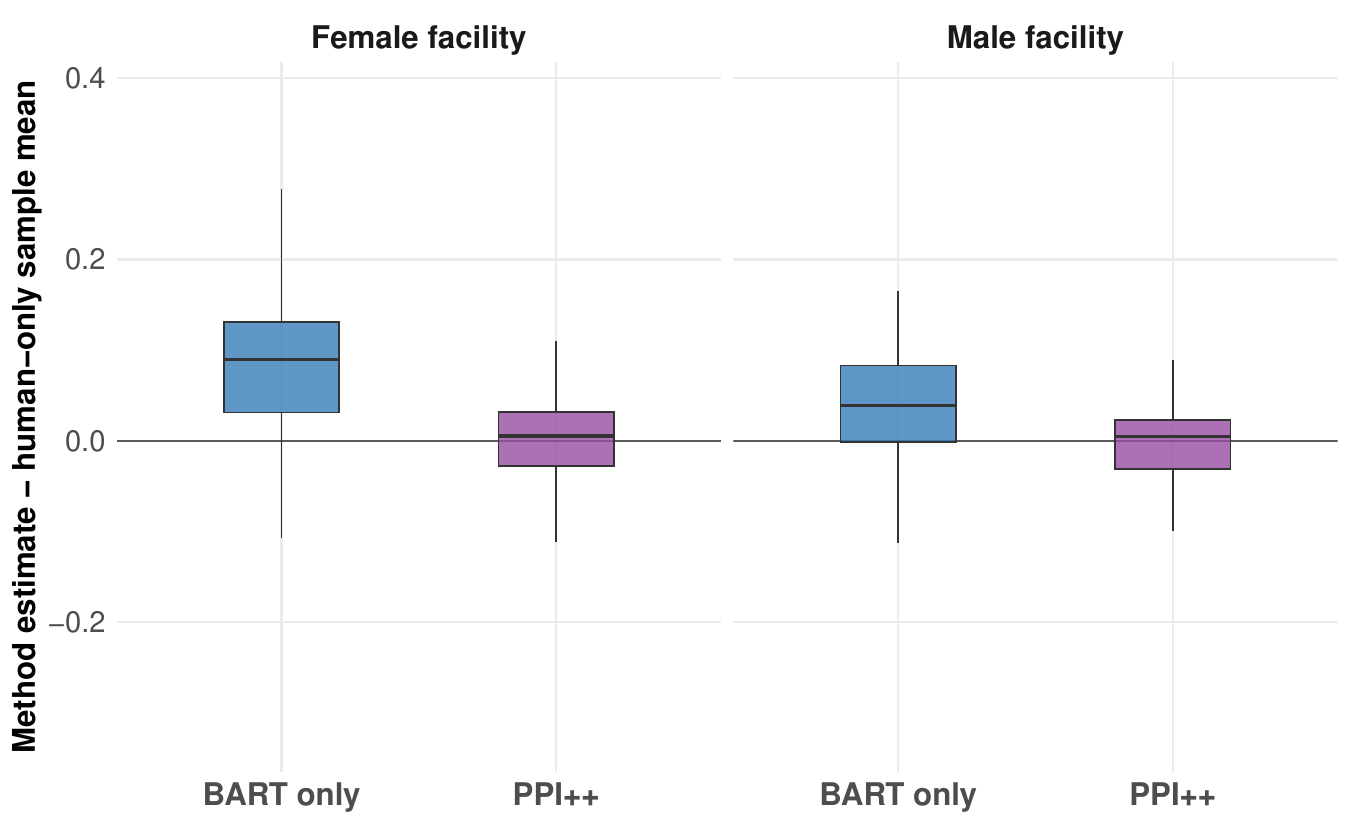}
    \caption*{a. Personal Growth} 
        \end{center}

  \end{minipage}
  \begin{minipage}[b]{0.5\linewidth}
       \begin{center}

    \includegraphics[width=\linewidth]{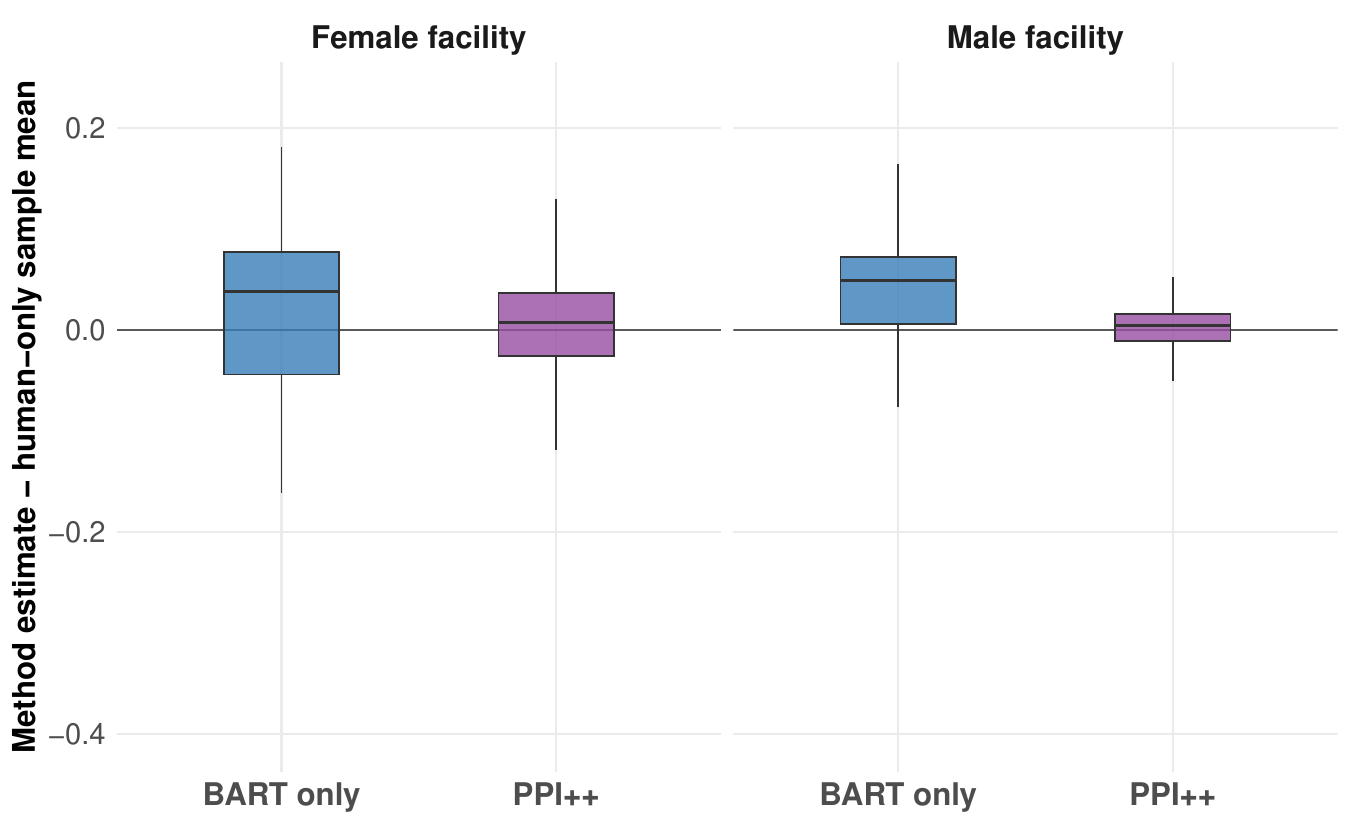}
    \caption*{b. Community Support} 
        \end{center}

  \end{minipage}

   \begin{minipage}[b]{0.5\linewidth}
       \begin{center}

    \includegraphics[width=\linewidth]{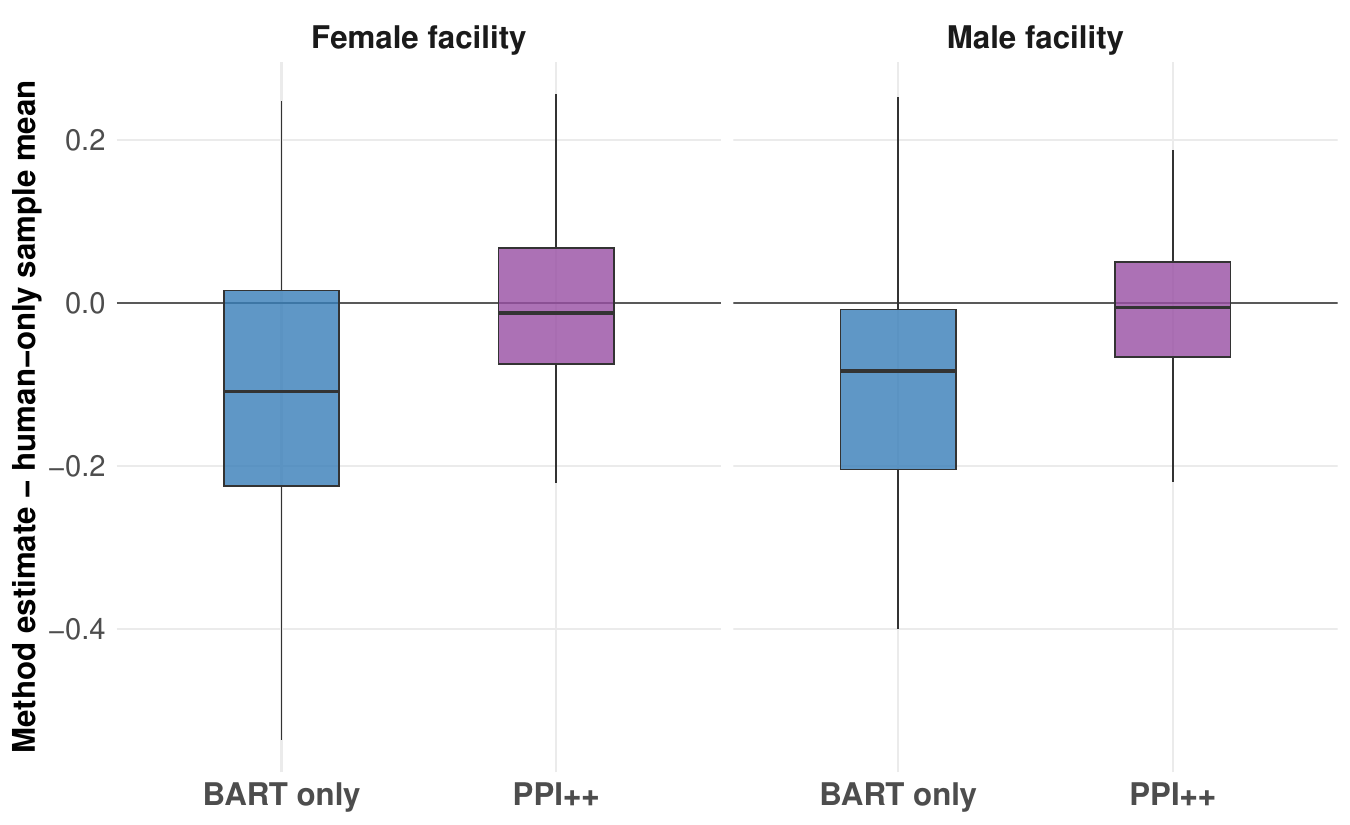}
    \caption*{c. Rule Violations} 
        \end{center}

  \end{minipage}
  \begin{minipage}[b]{0.5\linewidth}
       \begin{center}

    \includegraphics[width=\linewidth]{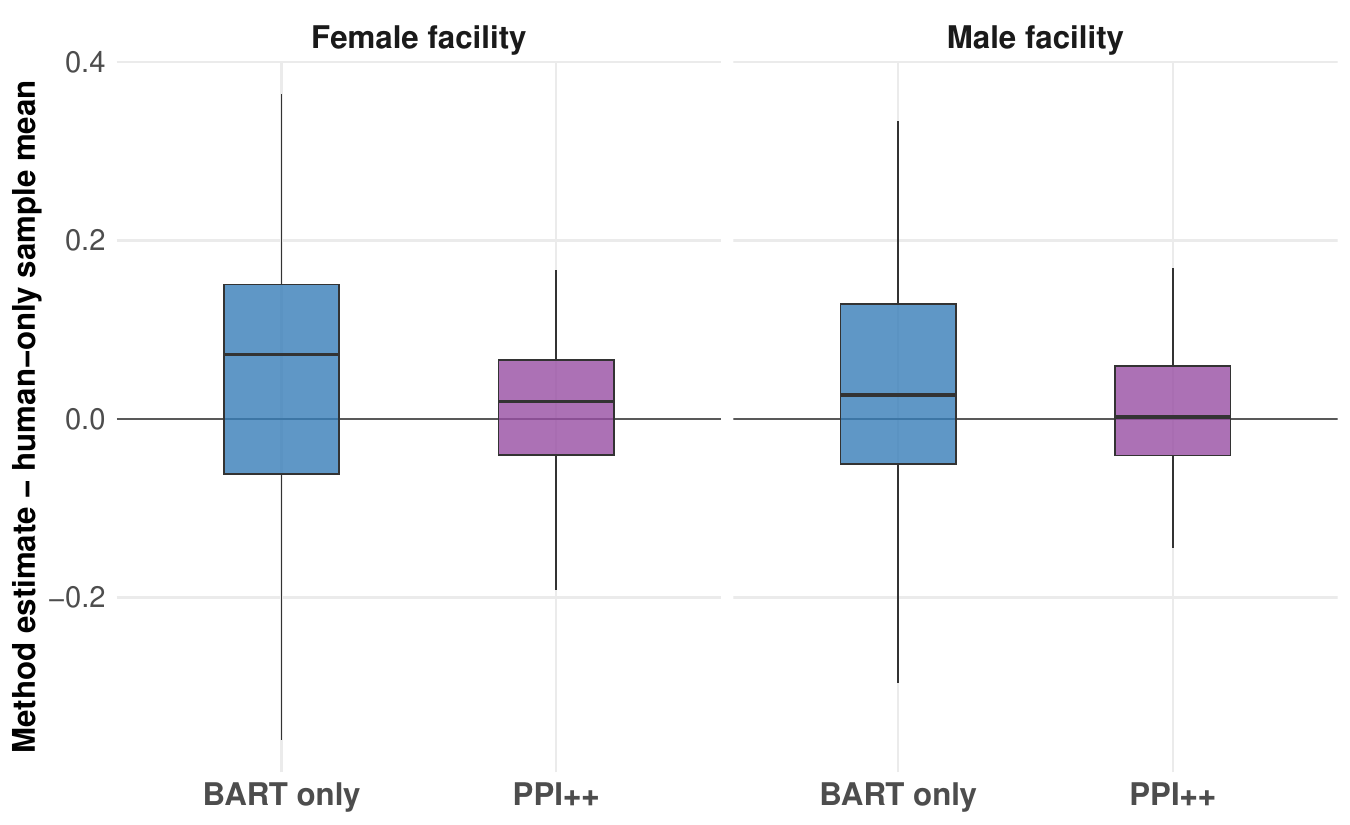}
    \caption*{d. Disruptive Conduct} 
        \end{center}

  \end{minipage}
\begin{minipage}{15.0 cm}{\footnotesize{Notes: The human mean is constructed by averaging multinomial classifications from Prolific raters over sampled messages for each resident. ``BART only'' uses zero-shot probabilities from BART aggregated over all of the resident's messages. ``PPI++'' combines the human-labeled sample with BART predictions on the resident's full message set (not including any messages used in label generation data).}}
\end{minipage} 
\end{figure}

\begin{table}[!htbp]
\caption{Cragg-Donald Weak Instruments Test}
\label{tab:weak_instruments}
\begin{center}
\begin{tabular}{p{3cm}p{2.0cm}p{2.0cm}p{2.0cm}}
\hline\hline
Gender & Profile Type & CD Statistic & Critical Value \\
\midrule
\textbf{Female} & Sender & 12.94 & 10.25 \\
                & Receiver & 13.31 & 10.25 \\
\textbf{Male}   & Sender & 14.84 & 10.25 \\
                & Receiver & 12.42 & 10.25 \\
\hline\hline
\multicolumn{4}{c}{ \begin{minipage}{12 cm}{\footnotesize{Notes: This analysis is done on the held-out messages that were not used in zero-shot label generation.}}
\end{minipage}} \\
\end{tabular}%
\end{center}
\end{table}

\begin{table}[!htbp]
\begin{center}
\caption{Peer Effects in Class Probabilities (Female)}
\label{rhotablewithadjust}
\resizebox{0.85\textwidth}{!}{%
\begin{tabular}{lP{2.5cm}P{3cm}P{3cm}}
\hline\hline
 & Personal Growth & Community Support & Rule Violations \\
\midrule

\textbf{a. Sender Profiles}&&&\\
Peer Personal Growth ALR   & 0.847   & -0.665  & 0.635   \\
                           & (0.430) & (0.434) & (0.307) \\
Peer Community Support ALR & 0.230   & 1.997   & -0.694  \\
                           & (0.531) & (0.535) & (0.379) \\
Peer Rule Violations ALR   & -0.133  & -0.008  & 1.618   \\
                           & (0.496) & (0.500) & (0.354) \\
 \textbf{b. Receiver Profiles}                       &&&\\
Peer Personal Growth ALR   & 2.303   & 1.375   & 0.431   \\
                           & (0.415) & (0.440) & (0.279) \\
Peer Community Support ALR & -1.228  & -0.461  & -0.457  \\
                           & (0.490) & (0.519) & (0.329) \\
Peer Rule Violations ALR   & 0.575   & 0.033   & 1.536   \\
                           & (0.352) & (0.372) & (0.236) \\
\hline\hline
\multicolumn{4}{c}{ \begin{minipage}{15 cm}{\footnotesize{Notes: The adjustment is done using a latent space model that allows for the inclusion of covariates (proximity in entry dates). The peer effects analysis is done on the held-out messages that were not used in zero-shot label generation.}}
\end{minipage}} \\
\end{tabular}%
}
\end{center}
\end{table}

\begin{table}[!h]
\caption{Peer Effects in Text Class Probabilities (Female) IV 2SLS without Adjustment}
\label{rhotablewithoutadjust}
\begin{center}
\begin{tabular}{lP{2cm}p{2cm}p{2cm}}
\hline\hline
 & Personal Growth ALR & Community Support ALR& Rule Violations ALR\\
\midrule

\textbf{a. Sender Profiles}&&&\\
Peer Personal Growth ALR   & 1.050   & -0.104  & 0.196   \\
                           & (0.240) & (0.245) & (0.168) \\
Peer Community Support ALR & -0.047  & 1.088   & -0.172  \\
                           & (0.303) & (0.309) & (0.212) \\
Peer Rule Violations ALR   & -0.017  & -0.242  & 1.378   \\
                           & (0.339) & (0.347) & (0.238) \\
 \textbf{b. Receiver Profiles}                       &&&\\
Peer Personal Growth ALR   & 1.042   & 0.218   & -0.193  \\
                           & (0.223) & (0.235) & (0.152) \\
Peer Community Support ALR & 0.170   & 0.867   & 0.363   \\
                           & (0.287) & (0.303) & (0.196) \\
Peer Rule Violations ALR   & 0.745   & 0.305   & 1.661   \\
                           & (0.269) & (0.284) & (0.183) \\
\hline\hline
\multicolumn{4}{c}{ \begin{minipage}{13 cm}{\footnotesize{Notes: This table deploys the IV 2SLS for multivariate outcomes but does not adjust for the latent homophily driven by endogeneity in network formation.}}
\end{minipage}} \\
\end{tabular}
\end{center}
\end{table}

\begin{table}[!h]
\caption{Peer Effects in Class Probabilities (Male)}
\label{rhotablewithadjustmale}
\begin{center}
\resizebox{0.85\textwidth}{!}{%
\begin{tabular}{lP{2.5cm}P{3cm}P{3cm}}
\hline\hline
 & Personal Growth & Community Support & Rule Violations\\
\midrule

\textbf{a. Sender Profiles}&&&\\
Peer Personal Growth ALR   & 0.368   & -0.777  & 0.031   \\
                           & (0.459) & (0.462) & (0.299) \\
Peer Community Support ALR & 0.644   & 1.693   & 0.013   \\
                           & (0.466) & (0.470) & (0.303) \\
Peer Rule Violations ALR   & -0.542  & -0.267  & 0.932   \\
                           & (0.369) & (0.372) & (0.240) \\
 \textbf{b. Receiver Profiles}                       &&&\\
Peer Personal Growth ALR   & 1.684   & 0.671   & -0.030  \\
                           & (0.405) & (0.434) & (0.267) \\
Peer Community Support ALR & -0.533  & 0.401   & -0.025  \\
                           & (0.390) & (0.418) & (0.257) \\
Peer Rule Violations ALR   & -0.095  & -0.085  & 1.126   \\
                           & (0.285) & (0.306) & (0.188) \\
\hline\hline
\multicolumn{4}{c}{ \begin{minipage}{15.0 cm}{\footnotesize{Notes: The adjustment is done using a latent space model that allows for the inclusion of covariates (proximity in entry dates). The peer effects analysis is done on the held-out messages that were not used in zero-shot label generation.}}
\end{minipage}} 
\end{tabular}%
}
\end{center}
\end{table}

\begin{table}[!h]
\caption{Peer Effects in Text Class Probabilities (Male) IV 2SLS without Adjustment}
\label{rhotablewithoutadjustmale}
\begin{center}
\begin{tabular}{lP{2cm}p{2cm}p{2cm}}
\hline\hline
 & Personal Growth ALR& Community Support ALR& Rule Violations ALR\\
\midrule

\textbf{a. Sender Profiles}&&&\\
Peer Personal Growth ALR   & 0.823   & -0.249  & 0.005   \\
                           & (0.330) & (0.333) & (0.217) \\
Peer Community Support ALR & 0.203   & 1.211   & 0.009   \\
                           & (0.345) & (0.348) & (0.227) \\
Peer Rule Violations ALR   & -0.152  & -0.238  & 1.126   \\
                           & (0.303) & (0.305) & (0.199) \\
 \textbf{b. Receiver Profiles}                       &&&\\
Peer Personal Growth ALR   & 1.291   & 0.332   & 0.016   \\
                           & (0.253) & (0.273) & (0.166) \\
Peer Community Support ALR & -0.247  & 0.681   & -0.045  \\
                           & (0.263) & (0.283) & (0.172) \\
Peer Rule Violations ALR   & 0.129   & 0.003   & 1.063   \\
                           & (0.223) & (0.241) & (0.147) \\
\hline\hline
\multicolumn{4}{c}{ \begin{minipage}{14 cm}{\footnotesize{Notes: This table deploys the IV 2SLS for multivariate outcomes but does not adjust for the endogeneity in network formation.}}
\end{minipage}} \\
\end{tabular}
\end{center}
\end{table}

\clearpage

\section{Additional methodological and theoretical details}
\label{details_methods}

\begin{remark}[Identifiability under orthogonal transformations]
\label{orthogonal}
The latent positions in the RDPG model are identifiable only up to a common orthogonal transformation. Specifically, for any $H\in\mathcal{O}(d)$, $U_i$ and $H^\top U_i$ generate the same network distribution. We therefore assume that the class $\mathcal{H}$ containing the conditional mean functions $h^f$, $f\in\{Y,Z,K\}$, is closed under orthogonal transformations, i.e., $h\in\mathcal{H}$ implies $h\circ H\in\mathcal{H}$ for every $H\in\mathcal{O}(d)$. We implement this condition by using the polynomial sieve consisting of all monomials of total degree at most $J_N$ (called tensor-poly in the rest of the paper). The residualized variables $\widetilde Y_i,\widetilde Z_i,\widetilde K_i$ and the moment condition in \eqref{momenteq} are invariant to the choice of latent-position orientation. Hence, under  Proposition~\ref{identify}, $\beta_0$ is uniquely identified even though $U$ is identified only up to orthogonal transformation. For the additive and multiplicative latent space model, the latent positions within the constrained parameter space described above are identified, and hence these restrictions are not needed.
\end{remark}

\begin{remark}[Asymptotic relevance of network instruments]
\label{relevance}
Recent work by \citet{hayes2024minimax} points out that in a linear-in-means peer effects model, identification does not necessarily imply that the parameters remain estimable as the network grows.
In particular, when nodal covariates are independent of network formation,
have a common population mean, and node degrees diverge, the quantities $GX$ and $G^2X$ may become
asymptotically close to a common value. This can asymptotically produce near-collinearity
even when the usual finite-sample rank condition is satisfied. This phenomenon
need not arise, however, when there is persistent heterogeneity in the
covariates and in the composition of neighborhoods. For example, our general
latent-space formulation can allow network formation to depend on both latent
traits and observed nodal characteristics,
$
A_{ij}=f(U_i,U_j,W_i,W_j,\xi_{ij}).
$
If the $X_i$'s are heterogeneous fixed covariates or their conditional
means vary with observed characteristics $W_i$, then this heterogeneity can induce
nonvanishing cross-sectional variation in $GX$ and $G^2X$ after adjustment
for latent homophily. Then the network averages need not converge to a common
constant. Our asymptotic theory relies on the strong-relevance
condition in Assumption~\ref{ass:iv-relevance}, which requires the residualized
network instruments to retain nondegenerate variation as $N$ grows and rules
out such asymptotically collinear sequences.
\end{remark}

\clearpage

\section{Multi-Head Self-Attention and transformer blocks}
\label{attn}

The first step to processing text is tokenizing it into words or sub-words $z_k$ for $k = 1, \ldots, K$, forming a sequence $\mb{z} = (z_1, \ldots, z_K)$ of tokens where $K$ is the total number of tokens for each written text. For each token position $k$, both BART and BERT construct an input embedding as a sum of token embedding $\mb{e}_k^{\mathrm{token}}$ and a positional embedding $\mb{e}_k^{\mathrm{pos}}$ as $\mb{e}_k = \mb{e}_k^{\mathrm{token}} + \mb{e}_k^{\mathrm{pos}}.$ In BERT, an additional segment embedding $\mb{e}_k^{\mathrm{seg}}$ is included if input text consists of a sentence pair: $\mb{e}_k = \mb{e}_k^{\mathrm{token}} + \mb{e}_k^{\mathrm{pos}} + \mb{e}_k^{\mathrm{seg}}.$ These vectors then form an input embedding matrix $\mb{E} = [\mb{e}_1, \ldots, \mb{e}_K]^\top \in \mathbb{R}^{K \times H}$. Here $H$ denotes the dimension of each embedding vector. Precisely, $H = 768$ for \textit{bert-base-uncased} and $H = 1024$ for \textit{bart-large}.

A sequence of transformer blocks processes $\mb{E}$, where the output of each block is fed as an input into the next block. Let the output of the $l$-th block be $\mb{E}^{(l)} \in \mathbb{R}^{K \times H}$, with the initial input being $\mb{E}^{(0)} = \mb{E}$. 

To compute attention within each encoder block $l = 1, \ldots,L$, the model first generates query ($\mb{Q}^{(l,a)}$), key ($\mb{K}^{(l,a)}$), and value matrices ($\mb{V}^{(l,a)}$) for each attention head $a = 1, \ldots, A$. These matrices are obtained by multiplying the previous layer's output $\mb{E}^{(l-1)}$ using head-specific learnable weight matrices $\mb{W}_B^{(l,a)}$ as follows,  $\mb{B}^{(l,a)} = \mb{E}^{(l-1)}\mb{W}_B^{(l,a)} \in \mathbb{R}^{K \times d}$ for $\mb{B} \in \{\mb{Q,K,V}\}$, $\mb{W}_B^{(l,a)} \in \mathbb{R}^{H\times d}$ and $d = H/A$. 
For reference, \textit{bert-base-uncased} uses $L=12$ encoder layers with $A=12$ attention heads, while \textit{bart-large} uses $L=12$ encoder layers with $A=16$ attention heads (and hence $d=64$ in both models). The $k$-th row of the matrix $\mb{Q}^{(l,a)}$ is denoted $\mb{q}_k^{(l,a)} \in \mathbb{R}^d$, with equivalent notation $\mb{k}_k^{(l,a)}$ and $\mb{v}_k^{(l,a)}$ for $\mb{K}^{(l,a)}$ and $\mb{V}^{(l,a)}$ respectively.

The query vector represents what the $k$-th token seeks, and the key vector indicates what each token offers others. The attention score between the $k$-th query and all $K$ keys is given by $\mb{s}_k^{(l,a)} = \left[ s_{k1}^{(l,a)}, s_{k2}^{(l,a)}, \ldots, s_{kK}^{(l,a)} \right]^\top \in \mathbb{R}^K$. Each element in this vector is the scaled dot-product between the corresponding query and key vectors, 
  $  s_{kj}^{(l,a)} = \frac{\mb{q}_k^{(l,a)\top} \mb{k}_j^{(l,a)}}{\sqrt{d}}. $ The attention weights are then obtained by normalizing $s_{kj}^{(l,a)}$ using the softmax function as $m_{kj}^{(l,a)} = \frac{\exp(s_{kj}^{(l,a)})}{\sum_{q=1}^K \exp(s_{kq}^{(l,a)})}.$ Let $\mb{M}^{(l,a)} \in \mathbb{R}^{K\times K}$ be the attention weights matrix for all tokens in $\mb{z}$. The output matrix is given by $\mb{H}^{(l,a)} = \mb{M}^{(l,a)} \mb{V}^{(l,a)} \in \mathbb{R}^{K \times d}$. Next, the outputs from $A$ heads are combined as: $\mb{H}^{(l)}=\left[ \mb{H}^{(l,1)}, \mb{H}^{(l,2)},\ldots,  \mb{H}^{(l,A)}\right] \in \mathbb{R}^{K \times H}$. 
  The column-concatenated matrix $\mb{H}^{(l)}$ is then processed within the encoder block to generate the block's final output $\mb{E}^{(l)}$ (Figure \ref{fig:BERT Architecture}).

To numerically encode the semantic content of each message, we use the output from the 11th encoder block (i.e., the penultimate block) of \textit{bert-base-uncased} model to construct an ``embedding profile" for each participant which is a high-dimensional vector summarizing their language use across multiple messages during their stay in TC. We choose the penultimate layer as that layer is widely thought to contain the unsupervised representation of data in the representation learning framework \citep{bengio2013representation}, while the last layer is trained or fine-tuned to be task-specific. One can, in principle, also choose an embedding representation from an earlier layer. Let $\mb{f}^{\mathrm{token}}_{c,k} \in \mathbb{R}^H$ denote the embedding of the $k$-th token in the $c$-th message, extracted from the 11th encoder. We obtain the message-level embedding by taking the average of these vectors over all tokens in the message as
$\mb{f}^{\mathrm{msg}}_c=\frac{1}{N_c} \sum_{k=1}^{N_c} \mb{f}^{\mathrm{token}}_{c,k}$, where $N_c$ is the number of tokens in message $c$. We then compute the ``sender embedding profile'' for a sender $s$ by taking an average of the the message-level embeddings, $\mb{f}^{\mathrm{sender}}_s=\frac{1}{N_s} \sum_{c \in \mathcal{C}_s} \mb{f}_c^{\mathrm{msg}}$, where $\mathcal{C}_s$ denotes the set of messages sent by sender $s$, and $N_s = |\mathcal{C}_s|$ is the number of such messages. The embedding profile for each receiver can be obtained analogously. 

\begin{figure}[!h]
    \centering
        \caption{BERT Architecture}
    \label{fig:BERT Architecture}
    \includegraphics[width=\linewidth]{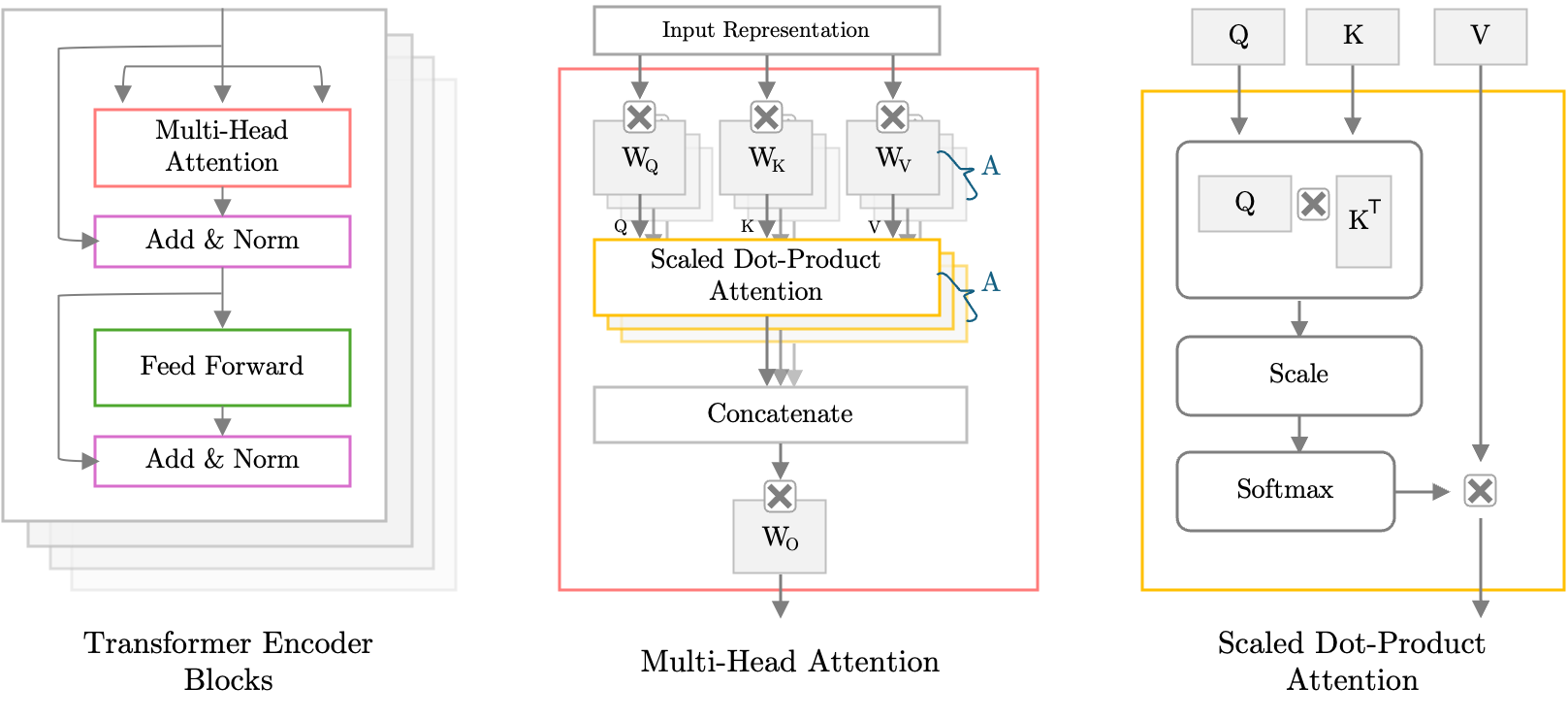}
    \begin{minipage}{15.0 cm}{\footnotesize{Notes: The architecture has been reproduced from the framework in \cite{vaswani2017attention}.}}
\end{minipage} 
\end{figure}

\subsection{Zero-Shot Classification}
\label{zeroshotappendix1}

Figure \ref{fig:bart-classificationl} illustrates how a TC text message is classified into user-defined categories using this framework. In our application, each peer message is treated as a premise, and we define four psychologically meaningful word strings: ``Rule Violations", ``Disruptive Conduct", ``Community Support", and ``Personal Growth". Each of the four strings is converted into a hypothesis using a fixed template: ``This example is \{\}." 

For each message $c$, the model generates four premise-hypothesis pairs by combining the message with each label-based-hypothesis indexed by $\ell=1, \ldots,4$. 
For each pair, the model obtains a probability vector for entailment $\mb{s}_{c} =(p_{c,\ell}^{\mathrm{entail}})_{\ell=1}^4 \in \mathbb{R}^4$, which reflect how strongly the input message supports each of the user-defined categories. These four entailment scores are then passed through a softmax function to obtain a probability vector over the labels: $\boldsymbol{\pi}_c = \mathrm{softmax}(\mb{s}_c)$. Finally, we aggregate these scores at the level of individuals to create personal class profiles, analogous to the embedding-based profiles described earlier, but in a more tractable lower-dimensional space. 

\begin{figure}[!h]
    \centering
        \caption{Process of zero-shot classification}
    \label{fig:bart-classificationl}
    \includegraphics[width=0.75\linewidth]{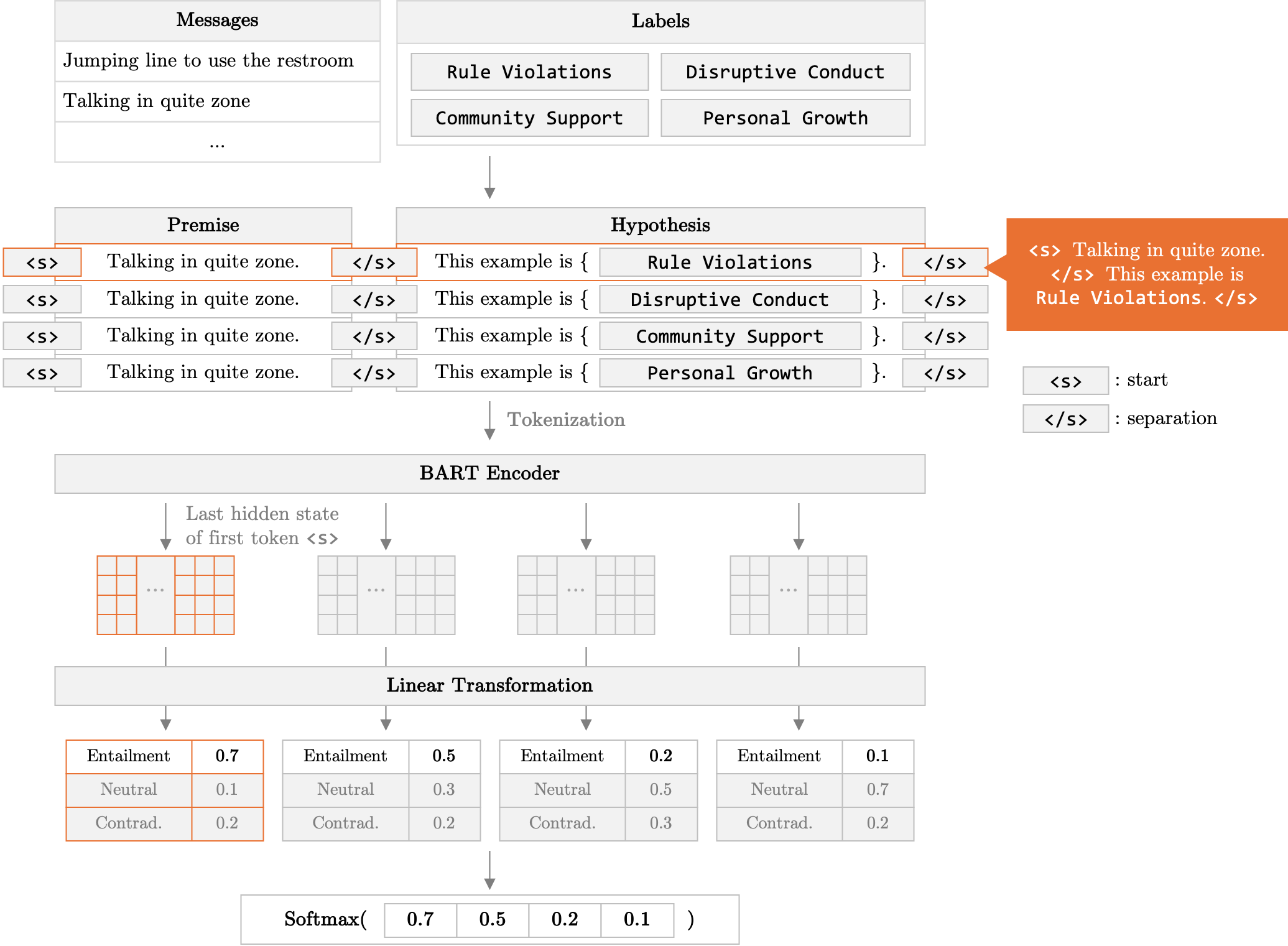}
\end{figure}

\end{document}